\documentclass[12pt]{article}
\usepackage[utf8]{inputenc}
\usepackage[sectionbib]{natbib}

\title{Robust Gradient Descent Estimation for Tensor Models under Heavy-Tailed Distributions}
\author{Xiaoyu Zhang\footnote{School of Mathematical Sciences,
Tongji University}, 
Di Wang\footnote{School of Mathematical Sciences,
Shanghai Jiao Tong University}, 
Guodong Li\footnote{Department of Statistics and Actuarial Science,
University of Hong Kong}, 
and Defeng Sun\footnote{Department of Applied Mathematics,
Hong Kong Polytechnic University}, }

\usepackage[margin=1in]{geometry}
\usepackage{graphicx}
\usepackage{mathrsfs}
\usepackage{multirow}
\usepackage{amsfonts}
\usepackage{amsmath}
\usepackage{comment}
\usepackage{amsthm}
\usepackage{color}
\usepackage{mathtools}
\usepackage{algorithm}
\usepackage{sectsty}
\usepackage[colorlinks]{hyperref}
\usepackage{amsmath}
\usepackage{amssymb}
\usepackage[toc,page]{appendix}
\usepackage[mathscr]{euscript} 
\usepackage[toc]{appendix}

\let\counterwithin\relax
\usepackage{chngcntr}
\usepackage{apptools}
\usepackage{booktabs}
\usepackage{lscape}
\usepackage{arydshln}
\usepackage{soul}
\usepackage{setspace}
\usepackage{epstopdf}
\usepackage{xr}
\usepackage{booktabs}
\usepackage{array}
\usepackage{makecell}

\makeatletter
\newcommand*{\addFileDependency}[1]{
  \typeout{(#1)}
  \@addtofilelist{#1}
  \IfFileExists{#1}{}{\typeout{No file #1.}}
}
\makeatother

\newcommand*{\myexternaldocument}[1]{%
    \externaldocument{#1}%
    \addFileDependency{#1.tex}%
    \addFileDependency{#1.aux}%
}

\myexternaldocument{Supplement}

\AtAppendix{\counterwithin{lemma}{section}}

\newtheorem{assumption}{Assumption}
\newtheorem{definition}{Definition}
\newtheorem{lemma}{Lemma}
\newtheorem{proposition}{Proposition}
\newtheorem{theorem}{Theorem}

\newtheorem{remark}{Remark}
\newtheorem{example}{Example}
\theoremstyle{definition}

\hypersetup{
    colorlinks=true,
    linkcolor=blue,
    citecolor=blue,
    filecolor=magenta,
    urlcolor=magenta,
}

\DeclareMathOperator*{\argmin}{arg\,min}

\newcommand{\bm}{\mathbf}
\newcommand{\bbm}{\boldsymbol}
\newcommand{\cm}[1]{\mbox{\boldmath$\mathscr{#1}$}}

\mathtoolsset{showonlyrefs}

\begin{document}

\setlength{\parindent}{16pt}

\maketitle

\vspace{-1cm}
\begin{abstract}

Low-rank tensor models are widely used in statistics. However, most existing methods rely heavily on the assumption that data follows a sub-Gaussian distribution. To address the challenges associated with heavy-tailed distributions encountered in real-world applications, we propose a novel robust estimation procedure based on truncated gradient descent for general low-rank tensor models. We establish the computational convergence of the proposed method and derive optimal statistical rates under heavy-tailed distributional settings of both covariates and noise for various low-rank models. Notably, the statistical error rates are governed by a local moment condition, which captures the distributional properties of tensor variables projected onto certain low-dimensional local regions. Furthermore, we present numerical results to demonstrate the effectiveness of our method.

\end{abstract}


\newpage

\setlength\abovedisplayskip{2pt}
\setlength\belowdisplayskip{2pt}

\section{Introduction}

\subsection{Low-Rank Tensor Modeling}

Low-rank tensor models have emerged as powerful tools for analyzing multiway data, which consist of observations with interactions across multiple modes or dimensions. Such data arise in a wide range of applications, including time series collected across multiple sensors, medical imaging data, and user–item interactions in recommendation systems. By leveraging low-dimensional structures, tensor methods enable dimension reduction, improve interpretability, and enhance computational scalability. These advantages have led to the growing use of tensor models in fields such as biomedical imaging \citep{Zhou13}, time series forecasting \citep{chen2022factor}, and collaborative filtering \citep{tarzanagh2022regularized}.

Despite significant progress in both convex and nonconvex optimization for tensor estimation, a major limitation remains. Most existing methods rely on strong distributional assumptions, such as sub-Gaussianity or boundedness of the noise or covariates. These assumptions are crucial for ensuring theoretical guarantees, including convergence rates and risk bounds, and they also contribute to the stability of optimization algorithms \citep{ZX18,raskutti2019convex,han2022optimal}. However, heavy-tailed distributions are common in many real-world applications. For example, biomedical signals such as electroencephalography (EEG) and functional magnetic resonance imaging (fMRI) data often exhibit skewness and outliers. Financial time series can contain extreme events and heavy-tailed noise. Sensor data collected in Internet of Things (IoT) applications or climate monitoring systems are frequently corrupted or contaminated. As a result, methods that assume light-tailed noise and/or Gaussian covariates may produce biased, unstable, or unreliable estimates when applied to such data.

The growing interest in robust estimation methods for high-dimensional low-rank matrix and tensor models underscores the pressing need for solutions that can handle heavy-tailed data. In terms of methodology to achieve robustness, the existing works can be broadly classified into two approaches: \textit{loss robustification} and \textit{data robustification}. The seminal Huber regression method \citep{huber1964robust, sun2020adaptive} exemplifies the first approach, where the standard least squares loss is replaced with a robust variant. For instance, \citet{tan2023sparse} applied the adaptive Huber regression with regularizations to sparse reduced-rank regression in the presence of heavy-tailed noise. \citet{shen2022computationally} employed the least absolute deviation (LAD) and Huber loss functions for low-rank matrix and tensor trace regression. While these loss-robustification methods provide robust control over residuals, they focus solely on the residuals' deviations and do not address the heavy-tailedness of the covariates. Moreover, robust loss functions like LAD and Huber loss cannot be easily generalized to more complex tensor models beyond linear trace regression.

Alternatively, \citet{fan2021shrinkage} proposed a robust low-rank matrix estimation procedure via data robustification. This method applies appropriate shrinkage to the data, constructs robust moment estimators from the shrunk data, and ultimately derives a robust estimate for the low-rank parameter matrix. The primary objective of data robustification is to mitigate the influence of samples with large deviations, thereby producing a robust estimate. However, when applied to low-rank matrix and tensor models, the data robustification procedure has limitations. Specifically, it overlooks the inherent structure of the model and fails to exploit the low-rank decomposition. As shown in Section \ref{sec:4}, not all information in the data contributes effectively to estimating the tensor decomposition. Consequently, the data robustification approach may be suboptimal for low-rank tensor estimation.

In this article, we propose a computationally scalable and theoretically grounded framework for robust tensor estimation. Our approach addresses the challenges of heavy-tailed data by introducing \textit{gradient robustification}. Instead of modifying the loss function or preprocessing the data, we stabilize the gradient updates themselves. Specifically, we develop a robust gradient descent algorithm that uses entrywise gradient truncation to reduce the influence of outliers or heavy-tailed noise. Rather than computing the full sample-mean gradient, which is sensitive to extreme values, we truncate gradient entries that exceed a carefully chosen threshold. This ensures that each gradient update is driven by reliable and representative components of the signal. Our method is model-agnostic and applies to a wide range of tensor estimation tasks, including tensor linear regression, logistic regression, and principal component analysis (PCA). Importantly, our approach does not require sub-Gaussian assumptions. Instead, we operate under mild local moment conditions that constrain the tail behavior of the data in low-dimensional subspaces defined by the Tucker decomposition. This localization leads to sharper and more adaptive statistical guarantees and allows us to handle both heavy-tailed covariates and noise.

We summarize our main contributions as follows.\vspace{-0.2cm}
\begin{itemize}
\item[1.] We develop a general and computationally scalable robust gradient descent framework for low-rank tensor estimation, applicable to a wide range of tensor learning tasks, including tensor linear regression, tensor logistic regression, and tensor PCA.\vspace{-0.3cm}
\item[2.] We establish that the method achieves optimal statistical error rates under \textit{the most relaxed moment assumptions}, specifically finite $(1+\epsilon)$-th and $(2+2\lambda)$-th moments for noise and covariates, respectively, without requiring sub-Gaussianity.\vspace{-0.3cm}
\item[3.] We introduce the concept of local moment conditions, a novel technical tool that characterizes the distributional properties of tensor components along low-rank directions and leads to sharper statistical guarantees.
\end{itemize}

The remainder of this article is organized as follows. Section \ref{sec:2} introduces the robust gradient descent algorithm and provide the computational convergence analysis. In Section \ref{sec:3}, we apply the method to tensor linear regression, logistic regression, and PCA, and we establish theoretical guarantees under local moment conditions. We present simulation experiments to validate our theoretical findings in Section \ref{sec:4} and provide a real-data application in Section \ref{sec:5}. We conclude with a discussion of extensions and future directions in Section \ref{sec:6}. Technical proofs, implementation details, discussions, and additional numerical results are provided in the supplementary materials.

\subsection{Related Literature}

This article is related to a large body of literature on nonconvex methods for low-rank matrix and tensor estimation. The gradient descent algorithm and its variants have been extensively studied for low-rank matrix models \citep{chen2015fast, tu2016low, wang2017unified, ma2018implicit} and low-rank tensor models \citep{xu2017efficient, chen2019non, han2022optimal, tong2022accelerating, tong2022scaling}. For simplicity, we focus on the robust alternatives to the standard gradient descent, although the proposed technique can be extended to other gradient-based methods. Robust gradient methods have also been explored for low-dimensional statistical models in convex optimization \citep{prasad2020robust}. Our work differs from the existing work as we consider the general low-rank tensor estimation framework under the heavy-tailed distribution setting.

Robust estimation against heavy-tailed distributions is another emerging topic in high-dimensional statistics. Various robust $M$-estimators have been proposed for mean estimation \citep{catoni2012challenging, bubeck2013bandits, devroye2016sub} and high-dimensional linear regression \citep{fan2017estimation, loh2017statistical, sun2020adaptive, wang2020tuning}. More recently, robust methods for low-rank matrix and tensor estimation have been developed in \citet{fan2021shrinkage}, \citet{tan2023sparse}, \citet{wang2023rate}, \citet{shen2022computationally}, \citet{shen2023quantile}, and \citet{barigozzi2023robust}. Compared to these existing methods, our proposed approach can achieve the same or even better convergence rates under \textit{the most relaxed distribution assumptions} on both covariates and noise, as highlighted in Table \ref{table:moment_comparison}.

\begin{table}[H]
    \caption{Comparison of robust estimation methods in covariate moment or distribution requirements and noise moment requirements ($0<\lambda,\epsilon\leq 1$)}\vspace{0.2cm}
    \label{table:moment_comparison}
    \begin{tabular}{m{1.7cm}|m{1.2cm}m{7cm}m{2.2cm}m{2.2cm}}
        \toprule
        Method & Param. Shape & \makecell{Model} & Covariate mom./dist. & Noise mom. \\
        \midrule
        \multirow{4}{1.7cm}{\makecell{Adaptive \\ Huber \\regression}}  & Vector &\makecell{ High-dim. linear regression\\ \citep{sun2020adaptive}} & 4th & $(1+\epsilon)$-th  \\
         & Matrix & \makecell{High-dim. multi-response regression \\ \citep{tan2023sparse}} & Bounded & $(1+\epsilon)$-th \\
        & Tensor & \makecell{High-dim. tensor trace regression \\ \citep{shen2022computationally}} & Gaussian & $(1+\epsilon)$-th\\ 
        & Tensor & \makecell{Tensor PCA \citep{shen2023quantile} \\ \citep{barigozzi2023robust} }  & - & 2nd \\
        \midrule
        \multirow{3}{1.7cm}{\makecell{~\\ Data \\ shrinkage}} & Matrix & \makecell{High-dim. matrix trace regression \\ \citep{fan2021shrinkage}} & 4th & 2nd \\
        & Vector & \makecell{High-dim. logistic regression \\ \citep{zhu2021taming}} & 4th & - \\
        & Matrix & \makecell{High-dim. vector autoregression \\ \citep{wang2023rate}} & $(2+2\lambda)$-th & $(2+2\lambda)$-th \\
        \midrule
        \multirow{2}{1.7cm}{\makecell{Robust \\ gradient \\ descent}} & Vector & \makecell{Low-dim. linear model and generalized \\ linear model \citep{prasad2020robust}} & 4th  & 2nd \\ 
        & Tensor & \makecell{High-dim. tensor linear model and \\generalized linear model ({\color{blue}our proposal})} & \boldmath\textbf{$(2+2\lambda)$-th} & \boldmath\textbf{$(1+\epsilon)$-th} \\
        \bottomrule
    \end{tabular}
\end{table}

\subsection{Notation}\label{sec:1.3}

Throughout this article, we denote vectors by boldface small letters (e.g. $\bm{x}$), matrices by boldface capital letters (e.g. $\bm{X}$), and tensors by boldface Euler letters (e.g. $\cm{X}$), respectively. We introduce the tensor notations and operations used in the article, and their formal definitions and properties are relegated to Appendix \ref{append:tensor_algebra_and_notations} of supplementary materials. For generic $\cm{X}\in\mathbb{R}^{p_1\times\cdots\times p_d}$, $\cm{Y}\in\mathbb{R}^{p_1\times\cdots\times p_{d_0}}$ with $d_0\leq d$, and $\bm{Y}_k\in\mathbb{R}^{q_k\times p_k}$ for $k=1,\dots,d$, the mode-$k$ matricization of $\cm{X}$ is denoted as $\cm{X}_{(k)}$; the generalized inner product of $\cm{X}$ and $\cm{Y}$ is denoted as $\langle\cm{X},\cm{Y}\rangle$; the mode-$k$ multiplication of $\cm{X}$ and $\bm{Y}$ is denoted as $\cm{X}\times_k\bm{Y}_k$. For any $\cm{X}$ and $\cm{Y}$, their tensor outer product is denoted as $\cm{X}\circ\cm{Y}$.

We use $C$ to denote a generic positive constant. For any two sequences $x_k$ and $y_k$, we write $x_k\gtrsim y_k$ if there exists a constant $C>0$ such that $x_k\geq Cy_k$ for all $k$. Additionally, we write $x_k\asymp y_k$ if $x_k\gtrsim y_k$ and $y_k\gtrsim x_k$. For a generic matrix $\bm{X}$, we let $\bm{X}^\top$, $\|\bm{X}\|_\text{F}$, $\|\bm{X}\|$, $\text{vec}(\bm{X})$, and $\sigma_j(\bm{X})$ denote its transpose, Frobenius norm, operator norm, vectorization, and the $j$-th largest singular value, respectively. For any real symmetric matrix $\bm{X}$, let $\lambda_{\min}(\bm{X})$ and $\lambda_{\max}(\bm{X})$ denote its minimum and maximum eigenvalues.

\section{Methodology}\label{sec:2}

\subsection{Gradient Descent with Robust Gradient Estimates}

We consider a general framework for low-rank tensor estimation, where the loss function $\overline{\mathcal{L}}(\cm{A};z_i)$ depends on a $d$-th order parameter tensor $\cm{A}$ and an observed data point $z_i$. Suppose the parameter tensor admits a Tucker low-rank decomposition \citep{kolda2009tensor}
\begin{equation}\label{eq:Tucker}
    \cm{A}=\cm{S}\times_1\bm{U}_1\times_2\bm{U}_2\cdots\times_d\bm{U}_d=\cm{S}\times_{j=1}^d\bm{U}_j,
\end{equation}
where $\cm{S}\in\mathbb{R}^{r_1\times r_2\times\dots\times r_d}$ is the core tensor and each $\bm{U}_j\in\mathbb{R}^{p_j\times r_j}$ is the factor matrix. Throughout the article, we assume that the order $d$ is fixed and the ranks $(r_1,r_2,\cdots,r_d)$ are known. For brevity, we denote the tuple of components as $\bm{F}=(\cm{S},\bm{U}_1,\dots,\bm{U}_d)$ and define the loss function with respect to $\bm{F}$ as 
\begin{equation}
    \mathcal{L}(\bm{F};z_i) = \overline{\mathcal{L}}(\cm{S}\times_{j=1}^d\bm{U}_j;z_i)\quad\text{and}\quad\mathcal{L}_n(\bm{F};\mathcal{D}_n) = \frac{1}{n}\sum_{i=1}^n\mathcal{L}(\bm{F};z_i).
\end{equation}

A standard estimation method is to minimize the following regularized loss function
\begin{equation}\label{eq:RegularizedEmpricalRisk}
    \mathcal{L}_n(\bm{F};\mathcal{D}_n)+\frac{a}{2}\sum_{j=1}^d\|\bm{U}_j^\top\bm{U}_j-b^2\bm{I}_{r_j}\|_\text{F}^2,
\end{equation}
where $a,b>0$ are tuning parameters.
The regularization terms help prevent rank deficiency and ensures balanced scaling across factor matrices \citep{han2022optimal}. This optimization problem is typically solved via gradient descent. Under suitable initialization, the estimation error depends on the intrinsic low-rank structure and the data distribution.

However, when the data $z_i$ are heavy-tailed, the standard gradient descent approach may suffer from suboptimal performance. This is because the partial gradients of the loss,
\begin{equation}
    \nabla_{\bm{U}_k} \mathcal{L}_n(\bm{F}; \mathcal{D}_n) = \frac{1}{n} \sum_{i=1}^n \nabla_{\bm{U}_k} \mathcal{L}(\bm{F}; z_i)\quad\text{and}\quad \nabla_{\scalebox{0.7}{\cm{S}}} \mathcal{L}_n(\bm{F}; \mathcal{D}_n) = \frac{1}{n} \sum_{i=1}^n \nabla_{\scalebox{0.7}{\cm{S}}} \mathcal{L}(\bm{F}; z_i),
\end{equation}
are sample means and thus sensitive to extreme values. To improve robustness, we propose replacing them with robust gradient estimates, denoted by $\bm{G}_k(\bm{F})$ and $\cm{G}_0(\bm{F})$, which are designed to maintain stability under heavy-tailed distributions.

Our robust gradient descent algorithm is presented in Algorithm \ref{alg:1}. At each iteration, the standard partial gradients are replaced with their robust alternatives, and the regularization term is retained to ensure numerical stability of the factor matrices. The algorithm is applicable to a broad class of tensor estimation problems, and the robustness of the procedure depends crucially on the quality of the gradient estimates $\bm{G}_k(\bm{F})$ and $\cm{G}_0(\bm{F})$.

\begin{algorithm}[!htp]
    \setstretch{1.2}
    \caption{Robust gradient descent algorithm}
    \label{alg:1}
    \begin{flushleft}
    \textbf{input}: $\bm{F}^{(0)}$, $a,b>0$, step size $\eta>0$, and number of iterations $T$\\
    \textbf{for} $t=0$ to $T-1$\\
    \hspace*{1cm}\textbf{for} $k=1$ to $d$\\
    \hspace*{2cm}$\bm{U}_k^{(t+1)} \gets \bm{U}_k^{(t)} - \eta\cdot\bm{G}_k(\bm{F}^{(t)}) - \eta a\bm{U}_k^{(t)}(\bm{U}_k^{(t)\top}\bm{U}_k^{(t)}-b^2\bm{I}_{r_k})$\\
    \hspace*{1cm}\textbf{end for}\\
    \hspace*{1cm}$\cm{S}^{(t+1)} \gets \cm{S}^{(t)}-\eta\cdot\cm{G}_0(\bm{F}^{(t)})$\\
    \textbf{end for}\\
    \textbf{return} $\cm{\widehat{A}}=\cm{S}^{(T)}\times_1\bm{U}_1^{(T)}\dots\times_d\bm{U}_d^{(T)}$\vspace*{-0.45cm}
    \end{flushleft}
\end{algorithm}

\subsection{Local Convergence Analysis}

While Algorithm \ref{alg:1} employs robust gradient estimates in place of standard gradients, these robust gradients are not necessarily derived from an explicit robust loss function. Consequently, the algorithm does not correspond to minimizing a well-defined objective function in the traditional sense, which complicates the analysis of its convergence properties. To establish theoretical guarantees, we introduce a set of conditions that link the behavior of the robust gradients to the underlying optimization landscape.

We begin by imposing a condition on the expected gradient of the original loss function $\overline{\mathcal{L}}$. This condition ensures that, for low-rank tensors, the gradient provides meaningful descent directions toward the population optimum.

\begin{definition}[Restricted correlated gradient]\label{def:1}
    The loss function $\overline{\mathcal{L}}$ satisfies the restricted correlated gradient (RCG) condition: for any $\cm{A}$ such that $\textup{rank}(\cm{A}_{(k)})\leq r_k$, $1\leq k\leq d$,
    \begin{equation}
        \langle\mathbb{E}[\nabla\overline{\mathcal{L}}(\cm{A};z_i)],\cm{A}-\cm{A}^*\rangle\geq\frac{\alpha}{2}\|\cm{A}-\cm{A}^*\|_\textup{F}^2 + \frac{1}{2\beta}\|\mathbb{E}\nabla\overline{\mathcal{L}}(\cm{A};z)\|_\textup{F}^2,
    \end{equation}
    where the RCG parameters $\alpha$ and $\beta$ satisfy $0<\alpha\leq\beta$.
\end{definition}

This condition implies that the expected gradient is well-aligned with the direction of improvement $\cm{A}-\cm{A}^*$, and that the gradient norm carries meaningful curvature information. It generalizes the notion of restricted strong convexity/smoothness to the setting of low-rank tensor estimation, but is stated directly in terms of the gradient rather than the loss itself, which is crucial when the loss may not have finite moments under heavy-tailed distribution.

\begin{remark}
    In settings where the risk $\mathbb{E}[\overline{\mathcal{L}}(\cm{A};z_i)]$ has finite second moments, the RCG condition is closely related to (and often implied by) standard notions of restricted strong convexity and smoothness. However, for heavy-tailed or non-sub-Gaussian data, it is more natural to impose such conditions directly on the gradient.
\end{remark}

Next, we impose conditions on the robust gradient estimator $\bm{G}_k(\bm{F})$ and $\cm{G}_0(\bm{F})$ in Algorithm \ref{alg:1}. These conditions ensure that the robust gradients remain close to the population gradients, even in the presence of outliers or heavy tails.

\begin{definition}[Stability of robust gradients]\label{def:2}
    For the given $\bm{F}$, the robust gradient functions are stable if there exist positive constants $\phi$ and $\xi_k$, for $0\leq k\leq d$, such that
    \begin{equation}
        \begin{split}
            \|\bm{G}_k(\bm{F})-\mathbb{E}[\nabla_{\bm{U}_k}\mathcal{L}(\bm{F};z_i)]\|_\textup{F}^2 & \leq \phi \|\cm{S}\times_{j=1}^d\bm{U}_j-\cm{A}^*\|_\textup{F}^2 + \xi_k^2,\\
            \text{and }\|\cm{G}_0(\bm{F})-\mathbb{E}[\nabla_{\scalebox{0.7}{\cm{S}}}\mathcal{L}(\bm{F};z_i)]\|_\textup{F}^2 & \leq \phi \|\cm{S}\times_{j=1}^d\bm{U}_j-\cm{A}^*\|_\textup{F}^2 + \xi_0^2.
        \end{split}
    \end{equation}
\end{definition}
These bounds control how much the robust gradient deviates from the true population gradient. The term $\phi\|\cm{A}-\cm{A}^*\|_\text{F}^2$ captures error that grows with the optimization error, while $\xi_k^2$ represents the inherent estimation error of the robust gradient estimator. The universal constant $\phi$ governs the sensitivity of all gradient components, and $\xi_k$'s reflect component-specific accuracy.

For the ground truth $\cm{A}^*$, denote its largest and smallest singular values across all directions by $\bar{\sigma}=\max_{1\leq k\leq d}\|\cm{A}^*_{(k)}\|$ and $\underline{\sigma}=\min_{1\leq k\leq d}\sigma_{r_k}(\cm{A}^*_{(k)})$. The condition number of $\cm{A}^*$ is then given by $\kappa=\bar{\sigma}/\underline{\sigma}$. Due to the inhenrent rotational invariance in Tucker decompositions, we measure estimation error in a component-wise rotation-invariant fashion. For an estimate $\bm{F}=(\cm{S}, \bm{U}_1, \dots, \bm{U}_d)$, define the error
\begin{equation}\label{eq:est_error}
    \text{Err}(\bm{F})=\min_{\bm{O}_k\in\mathbb{O}^{r_k},1\leq k\leq d}\left\{\|\cm{S}-\cm{S}^*\times_{j=1}^d\bm{O}_k^\top\|_\text{F}^2+\sum_{k=1}^d\|\bm{U}_k-\bm{U}_k^*\bm{O}_k\|_\text{F}^2\right\},
\end{equation}
where the true decomposition satisfies $\|\bm{U}_k^*\|=b$ and the orthogonal matrices $\bm{O}_k$'s account for the unidentification of the Tucker decomposition. For the $t$-th iteration of Algorithm \ref{alg:1}, where $t=0,1,\dots,T$, denote the estimated parameters as $\bm{F}^{(t)}$ and $\cm{A}^{(t)}=\cm{S}^{(t)}\times_{j=1}^d\bm{U}_j^{(t)}$. The corresponding estimation error is then given by $\text{Err}(\bm{F}^{(t)})$.

We are now ready to state the local convergence guarantee for Algorithm \ref{alg:1}, under the RCG and gradient stability conditions.

\begin{theorem}\label{thm:1}
    Suppose that the loss function $\overline{\mathcal{L}}$ satisfies the RCG condition with parameters $\alpha$ and $\beta$ as in Definition \ref{def:1}, and that the robust gradient functions at each step $t$ satisfy the stability condition with parameters $\phi$ and $\xi_k$ as in Definition \ref{def:2}, for all $k=0,1,\dots,d$ and $t=1,2,\dots,T$. If the initial estimation error satisfies
    $\textup{Err}(\bm{F}^{(0)})\lesssim \alpha\beta^{-1}\bar{\sigma}^{2/(d+1)}\kappa^{-2}$, $\phi\lesssim\alpha^2\kappa^{-4}\bar{\sigma}^{2d/(d+1)}$, $a\asymp\alpha\kappa^{-2}\bar{\sigma}^{(2d-2)/(d+1)}$, $b\asymp\bar{\sigma}^{1/(d+1)}$, and $\eta\asymp\alpha\beta^{-1}\kappa^2$, then for $t=1,2,\dots,T$,
    \begin{equation}
        \textup{Err}(\bm{F}^{(t)}) \leq ~ (1-C\alpha\beta^{-1}\kappa^{-2})^t\cdot\textup{Err}(\bm{F}^{(0)}) + C\alpha^{-2}\bar{\sigma}^{-4d/(d+1)}\kappa^4\sum_{k=0}^d\xi_k^2,
    \end{equation}
    and
    \begin{equation}
        \|\cm{A}^{(t)}-\cm{A}^*\|_\textup{F}^2 \lesssim ~ \kappa^2(1-C\alpha\beta^{-1}\kappa^{-2})^t\cdot\|\cm{A}^{(0)}-\cm{A}^*\|_\textup{F}^2 + \bar{\sigma}^{-2d/(d+1)}\alpha^{-2}\kappa^4\sum_{k=0}^d\xi_k^2.
    \end{equation}
\end{theorem}

Theorem \ref{thm:1} estalibshes the linear convergence of the robust gradient descent iterates provided the robust gradients are well-behaved and the initialization is sufficiently accurate. In each upper bounds provided, the first term corresponds to optimization error that decays exponentially with the number of iterations, reflecting the improvement in the solution as the gradient descent progresses. The second term, on the other hand, captures the statistical error, which depends on the accuracy of the robust gradient estimators. Notably, the iterates does not necessarily converge to a fixed estimate, but maybe to a region of estimates with equivalent statistical properties. Thus, fast convergence relies on both a good initialization and high-quality robust gradient estimates. We provide a general strategy for robust gradient construction on Section \ref{sec:2.3}, and discuss model-specific initialization methods in Section \ref{sec:3}.

\subsection{Robust Gradient Estimation via Entrywise Truncation}\label{sec:2.3}

The robust gradient estimates in Algorithm \ref{alg:1} play a central role in ensuring stability under heavy-tailed distribution. In this subsection, we propose a concrete and general-purpose method for constructing such robust gradients, based on entrywise truncation, a simple yet powerful technique that provides robustness by controlling the influence of extreme values.

Recall that the standard partial gradients are essentially sample means of gradients across observations. As is well known, sample means are highly sensitive to outliers, especially when the data exhibit heavy tails. This sensitivity directly translates to instability in the gradient estimates used for optimization. To mitigate this issue, we adopt a strategy inspired by robust mean estimation: rather than using the raw gradient components, we replace them with truncated versions that bound the influence of extreme values. This approach is computationally simple and leverages techniques that have been shown to achieve near-optimal statistical performance under weak moment conditions \citep{fan2021shrinkage}.

Let $\bm{M}\in\mathbb{R}^{p\times q}$, and let $\tau>0$ be a user-specified truncation threshold. We define the entrywise truncation operator $\text{T}(\cdot,\cdot):\mathbb{R}^{p\times q}\times\mathbb{R}^+\to\mathbb{R}^{p\times q}$ as
\begin{equation}
    \text{T}(\bm{M},\tau)_{j,k}=\text{sgn}(\bm{M}_{j,k})\min(|\bm{M}_{j,k}|,\tau),~~\text{for }j=1,\dots,p,~~k=1,\dots,q,
\end{equation}
where $\text{sgn}(\cdot)$ denotes the sign function. This operator truncates each entry of $\bm{M}$ to have magnitude no greater than $\tau$, while preserving its sign. The same operation extends naturally to tensors. The truncation parameter $\tau$ plays a critical role in balancing the trade-off between truncation bias and robustness. A smaller $\tau$ increases robustness by suppressing outliers more aggressively, but may introduce bias if set too conservatively; a larger $\tau$ retains more of the original gradient signal, but risks amplifying the influence of anomalous observations. This trade-off is carefully balanced in our theoretical analysis, where the statistical error depends explicitly on the accuracy of the truncated gradient estimates.

Using the entrywise truncation operator, we construct robust estimators for the partial gradients of the loss function. The robust gradient estimators with respect to $\bm{U}_k$ and \cm{S} are 
\begin{equation}
    \begin{split}
        \bm{G}_k(\bm{F};\tau) = & \frac{1}{n}\sum_{i=1}^n\text{T}(\nabla_{\bm{U}_k}\mathcal{L}(\bm{F};z_i),\tau)
        = \frac{1}{n}\sum_{i=1}^n\text{T}(\nabla\overline{\mathcal{L}}(\cm{S}\times_{j=1}^d\bm{U}_j;z_i)_{(k)}(\otimes_{j\neq k}\bm{U}_j)\cm{S}_{(k)}^\top,\tau),\\
        \cm{G}_0(\bm{F};\tau) & = \frac{1}{n}\sum_{i=1}^n\text{T}(\nabla_{\scalebox{0.7}{\cm{S}}}\mathcal{L}(\bm{F};z_i),\tau)
        = \frac{1}{n}\sum_{i=1}^n\text{T}(\nabla\overline{\mathcal{L}}(\cm{S}\times_{j=1}^d\bm{U}_j;z_i)\times_{j=1}^d\bm{U}_j^\top,\tau).
    \end{split}
\end{equation}
Note that the truncation-based robust gradient estimator is generally applicable to a wide range of tensor models. In Sections \ref{sec:4} and \ref{sec:5}, we will show both theoretically and numerically that the entrywise truncation using a single parameter $\tau$ can achieve optimal estimation performance under various distributional assumptions.

\section{Applications to Tensor Models}\label{sec:3}

In this section, we apply the proposed robust gradient descent algorithm, equipped with entrywise truncated gradient estimators, to three fundamental tensor models: tensor linear regression, tensor logistic regression, and tensor PCA. These models cover a broad range of appliations, including multi-response regression, binary classification with tensor covariates, and unsupervised tensor signal extraction. In each setting, we assume that either the covariates, the noise, or both exhibit heavy-tailed behavior, motivating the need for robust estimation methods that go beyond sub-Gaussian assumptions.

For all models, we let $\bar{p}=\max_{1\leq j\leq d}p_j$ denote the maximum dimension across tensor modes, and define the effective dimension of the Tucker decomposition as
\begin{equation}
    d_{\text{eff}} = \sum_{k=1}^dp_kr_k + \prod_{k=1}^d r_k,
\end{equation}
which corresponds to the total number of free parameters in the low-rank tensor representation. Our theoretical analysis is based on a novel local moment condition, which will be introduced in next subsection, and serves as the foundation for our theoretical guarantees across all three models.

\subsection{Local Moments for Partial Gradients}\label{sec:3.1}

In the tensor models considered in this article, the partial gradients depend on both the tensor factors and the observed data. A key structural insight is that when the estimated factor matrices $\bm{U}_k$ are close to their ground truth counterparts $\bm{U}_k^*$, these gradients depend primarily on low-dimensional projections of the data onto the subspaces spanned by the factors. This motivates the use of \textit{local moment conditions} that capture the tail behavior of these projected components, rather than imposing global moment assumptions on the full data distribution. Such localized conditions are particularly useful in high-dimensional settings, where global moment constraints may be overly restrictive or unrealistic.

For a given sample $z_i$ and fixed estimates $\bm{F}=(\cm{S},\bm{U}_1,\dots,\bm{U}_d)$, the partial gradients with respect to the factor matrices and core tensors are given by
\begin{equation}
    \nabla_{\bm{U}_k}\mathcal{L}(\bm{F};z_i) = (\nabla\overline{\mathcal{L}}(\cm{A};z_i)\times_{j\neq k}\bm{U}_j^\top)_{(k)}\cm{S}_{(k)}^\top~~\text{and}~~\nabla_{\scalebox{0.7}{\cm{S}}}\mathcal{L}(\bm{F};z_i) = \nabla\overline{\mathcal{L}}(\cm{A};z_i)\times_{j=1}^d\bm{U}_j^\top,
\end{equation}
where $\cm{A}=\cm{S}\times_{j=1}^d\bm{U}_j$ is the reconstructed tensor, and $\nabla\overline{\mathcal{L}}(\cm{A};z_i)$ is the gradient of the underlying loss with respect to the full tensor. Observe that these partial gradients are obtained via multilinear projections of the full gradient onto the subspaces defined by the factor matrices $\bm{U}_j$. Consequently, their statistical behavior is governed not by the global distribution of the data, but by the distribution of these projected components, specifically, in neighborhoods of the true factor directions.

To formalize this, we introduce local moment conditions that characterize the tail behavior of the projected data and gradients in the vicinity of the true factor subspaces. For a given ground truth $\bm{U}_j^*\in\mathbb{R}^{p_j\times r_j}$ and a small radius $\delta\in[0,1]$, we define the set of unit vectors that lie within an angular distance of approximately $\arcsin(\delta)$ from the column space of $\bm{U}_j^*$ as 
\begin{equation}
    \mathcal{V}(\bm{U}_j^*,\delta)=\{\bm{v}\in\mathbb{R}^{p_j}:\|\bm{v}\|_2=1~\text{and}~\sin\arccos(\|\mathcal{P}_{\bm{U}^*_j}\bm{v}\|_2)\leq\delta\},
\end{equation}
where $\mathcal{P}_{\bm{U}^*_j} = \bm{U}_j^* (\bm{U}_j^{*\top} \bm{U}_j^*)^\dagger \bm{U}_j^{*\top}$ is the orthogonal projector onto the column space of $\bm{U}_j^*$, and $\dagger$ denotes the Moore--Penrose pseudo-inverse. The parameter $\delta$ controls the maximum allowable angular deviation of a unit vector $\bm{v}$ from the subspace spanned by $\bm{U}_j^*$. 

Equipped with these sets, we define two types of local moments that quantify the tail behavior of tensor-valued quantities. For a random tensor $\cm{T}\in\mathbb{R}^{p_1\times\cdots\times p_d}$ and fixed ground truth factors $\{\bm{U}_j^*\}_{j=1}^d$, moment order $\eta>0$, and radius $\delta\in[0,1]$, we define:
\begin{definition}[Local moments]
    The $\eta$-th all-mode local moment of $\cm{T}$ with radius $\delta$ is 
    \begin{equation}
        \textup{LM}_0(\cm{T};\eta,\delta,\{\bm{U}_j^*\}_{j=1}^d) = \sup_{\bm{v}_j\in\mathcal{V}(\bm{U}_j^*,\delta)}\mathbb{E}\left[|\cm{T}\times_{j=1}^d\bm{v}_j^\top|^{\eta}\right].
    \end{equation}
    Also, for $1\leq k\leq d$, its $\eta$-th mode-$k$-excluded local moment with radius $\delta$ is defined as
    \begin{equation}
        \textup{LM}_k(\cm{T};\eta,\delta,\{\bm{U}_j^*\}_{j=1}^d) = \sup_{\bm{v}_j\in\mathcal{V}(\bm{U}_j^*,\delta),~1\leq l\leq p_k}\mathbb{E}\left[|\cm{T}\times_{j=1,j\neq k}^d\bm{v}_j^\top\times_k\bm{c}_l^\top|^{\eta}\right],
    \end{equation}
    where $\bm{c}_l$ is the coordinate vector whose $j$-th entry is one and others zero.
\end{definition}

The all-mode local moment $\text{LM}_0$ captures the distributional behavior of the full projected gradient tensor $\nabla\overline{\mathcal{L}}(\cm{A};z_i)\times_{j=1}^d\bm{U}_j^\top$, which involves contributions from all factor modes. The mode-$k$-excluded local moment $\text{LM}_k$, on the other hand, focuses on the projection of the gradient onto all modes except the $k$-th, and is thus tailored to the estimation of the partial gradients with respect to $\bm{U}_k$. These definitions generalize traditional moment conditions by localizing them to the low-dimensional subspaces relevant to the underlying tensor structure.

When $\delta=1$, the sets $\mathcal{V}(\bm{U}_j^*,1)$ encompass the entire unit sphere, and the local moments reduce to their global counterparts:
\begin{equation}
    \text{LM}_0(\cm{T}; \eta, 1, \{\bm{U}_j^*\}_{j=1}^d) = \sup_{\|\bm{v}_j\|_2 = 1} \mathbb{E}\left[ \big| \cm{T} \times_{j=1}^d \bm{v}_j^\top \big|^\eta \right],
\end{equation}
with a similar reduction for $\text{LM}_k$. Thus, local moments provide a natural generalization of global moment assumptions, allowing for substantially weaker conditions in scenarios where the data exhibit heavy tails or high dimensionality, provided that the relevant projections behave benignly. To illustrate this advantage, consider the following example.

\begin{example}
    Let $\cm{X}\in\mathbb{R}^{p\times p\times p}$, where $\textup{vec}(\cm{X})\sim N(\bm{0}_{p^3},\bm{\Sigma}_{0.5})$ and $\bm{\Sigma}_{0.5}=0.5\textup{diag}(\bm{1}_{p^3})+0.5\bm{1}_{p^3}\bm{1}_{p^3}^\top$. Suppose that the ground truths are $\bm{U}_k^*=(1,\bm{0}_{p-1}^\top)^\top$ for $k=1,2,3$. Then, the global second moment of $\cm{X}$ is $\textup{LM}(\cm{X};2,1,\{\bm{U}_j^*\}_{j=1}^3) = (p^3+1)/2$, which grows with dimension. However, when restricted to directions $\bm{v}$ within an angular radius $\delta_1\leq p^{-3/2}$ (for $\textup{LM}_0$) and $\delta_2\leq p^{-1}$ (for each $\textup{LM}_k$), the corresponding local second moments are bounded by 2.
\end{example}
This example underscores the key advantage of our local moment framework: by focusing on the low-dimensional subspaces aligned with the true factors, we can work under much weaker moment assumptions than would be required if the full data distribution were considered. This property is essential for establishing robust and statistically optimal estimation under heavy-tailed distributions, as will be formalized in the subsequent theoretical analysis. More discussions of the local moment conditions are provided in Appendix \ref{append:local_moments} of supplementary materials.

\subsection{Heavy-Tailed Tensor Linear Regression}\label{sec:3.2}

We begin with tensor linear regression, a natural extension of classical linear models to tensor-valued predictors or responses. Given $0\leq d_0\leq d$, consider the model
\begin{equation}\label{eq:linearregression}
    \cm{Y}_i = \langle \cm{A}^*,\cm{X}_i \rangle + \cm{E}_i,\quad i=1,2,\dots,n,
\end{equation}
where $\cm{X}_i\in\mathbb{R}^{p_1\times\cdots\times p_{d_0}}$ is the $d_0$-th order tensor covariate, $\cm{Y}_i\in\mathbb{R}^{p_{d_0+1}\times\cdots\times p_d}$ is the $(d-d_0)$-th order tensor response, $\cm{E}_i\in\mathbb{R}^{p_{d_0+1}\times\cdots\times p_d}$ is the noise tensor with $\mathbb{E}[\cm{E}_i|\cm{X}_i]=\bm{0}$, $\cm{A}^*\in\mathbb{R}^{p_1\times \cdots\times p_d}$ is the coefficient tensor with Tucker ranks $(r_1,\dots,r_d)$. The goal is to estimate the coefficient tensor $\cm{A}^*$ from noisy observations, even when both the covariates and noise are heavy-tailed.

We adopt the least squares loss function for each observation
\begin{equation}
    \mathcal{L}(\bm{F}; z_i) = \frac{1}{2} \|\cm{Y}_i - \langle \cm{S} \times_{j=1}^d \bm{U}_j, \cm{X}_i \rangle\|_\text{F}^2,
\end{equation}
where $\bm{F}=(\cm{S},\bm{U}_1,\dots,\bm{U}_d)$. A key insight is that the partial gradients depend not on the high-dimensional raw data $\cm{X}_i$ and $\cm{Y}_i$, but on their low-dimensional projections induced by the factor matrices. Specifically, define the transformed variables $\overline{\cm{X}}_i = \cm{X}_i \times_{j=1}^{d_0} \bm{U}_j^\top$, $\overline{\cm{Y}}_i = \cm{Y}_i \times_{j=1}^{d - d_0} \bm{U}_{d_0 + j}^\top$, $\overline{\cm{X}}_{i,k}=\cm{X}_i\times_{j=1,j\neq k}^{d_0}\bm{U}_j^\top$ for $k=1,\dots,d_0$, and $\overline{\cm{Y}}_{i,k} =\cm{Y}_i\times_{j=1,j\neq k-d_0}^{d-d_0}\bm{U}_{d_0+j}^\top$ for $k=d_0+1,\dots,d$. 
For a truncation threshold $\tau>0$, the robust gradient estimators are
\begin{equation}\label{eq:robustgradients_lm2}
    \begin{split}
        \bm{G}_k(\bm{F}; \tau) &= \frac{1}{n} \sum_{i=1}^n \text{T}\big( \big[\overline{\cm{X}}_{i,k} \circ (\langle \cm{S}\times_{j=d_0+1}^{d}\bm{U}_j^\top\bm{U}_j, \overline{\cm{X}}_i \rangle - \overline{\cm{Y}}_i)\big]_{(k)} \cm{S}_{(k)}^\top, \tau \big), \quad\text{for }k = 1, \dots, d_0, \\
        \bm{G}_k(\bm{F}; \tau) &= \frac{1}{n} \sum_{i=1}^n \text{T}\big( \big[\overline{\cm{X}}_{i} \circ (\langle \cm{S}\times_k\bm{U}_k\times_{j=d_0+1,j\neq k}^{d}\bm{U}_j^\top\bm{U}_j, \overline{\cm{X}}_i \rangle - \overline{\cm{Y}}_{i,k})\big]_{(k)} \cm{S}_{(k)}^\top, \tau \big), \\
        & \quad\quad\quad\quad\quad\quad\quad\quad\quad\quad\quad\quad\quad\quad\quad\quad\quad\quad\quad\quad\quad\quad\quad\quad\text{for }k = d_0+1, \dots, d,\\
        \cm{G}_0(\bm{F}; \tau) &= \frac{1}{n} \sum_{i=1}^n \text{T}\big( \big[\overline{\cm{X}}_i \circ (\langle \cm{S}\times_{j=d_0+1}^d\bm{U}_j^\top\bm{U}_j, \overline{\cm{X}}_i \rangle - \overline{\cm{Y}}_i)\big] \times_{j=1}^d \bm{U}_j^\top, \tau \big).
    \end{split}
\end{equation}

As only the low-dimensional transformed data $\overline{\cm{X}}_i$, $\overline{\cm{Y}}_i$, $\overline{\cm{X}}_{i,k}$, and $\overline{\cm{Y}}_{i,k}$ appear in the truncated gradients in \eqref{eq:robustgradients_lm2}, it is crucial to characterize their distributional properties. Similar to $\overline{\cm{Y}}_i$ and $\overline{\cm{Y}}_{i,k}$, we can also define $\overline{\cm{E}}_i$ and $\overline{\cm{E}}_{i,k}$ as the transformed noise. We assume that the covariate $\cm{X}_i$ and noise $\cm{E}_i$ satisfy certain local moment bounds when projected onto subspaces defined by the true factors $\{\bm{U}_j^*\}_{j=1}^d$. These are formalized in Assumption \ref{asmp:1}, stated as follows.
\begin{assumption}\label{asmp:1}
    For some $\epsilon\in(0,1]$, $\lambda\in(0,1]$, and $\delta\in[0,1]$, the followings hold:
    \begin{itemize}
        \item[(a)] The vectorized covariate $\textup{vec}(\cm{X}_i)$ has the mean zero and positive definite variance $\bm{\Sigma}_x$ satisfying $0<\alpha_x\leq\lambda_{\min}(\bm{\Sigma}_x)\leq\lambda_{\max}(\bm{\Sigma}_x)\leq\beta_x$.
        \item[(b)] Conditioned on $\cm{X}_i$, the noise tensor $\cm{E}_i$ has the $(1+\epsilon)$-th local moment $M_{e,1+\epsilon,\delta} = \max_{0\leq k\leq d-d_0}[\textup{LM}_k(\cm{E}_i;1+\epsilon,\delta,\{\bm{U}_j^*\}_{j={d_0+1}}^d)]$.
        \item[(c)] $\cm{X}_i$ has the $(2+2\lambda)$-th global moment $M_{x,2+2\lambda} = \max_{0\leq k\leq d_0}\left[\textup{LM}_k(\cm{X}_i;2+2\lambda,1,\{\bm{U}_j^*\}_{j={1}}^{d_0})\right]$. In addition, as $1+\epsilon\leq 2+2\lambda$, let the $(1+\epsilon)$-th local moment of $\cm{X}_i$ be $M_{x,1+\epsilon,\delta} = \max_{0\leq k\leq d_0}\left[\textup{LM}_k(\cm{X}_i;1+\epsilon,\delta,\{\bm{U}_j^*\}_{j={1}}^{d_0})\right]$.
    \end{itemize}
\end{assumption}
These conditions are substantially weaker than sub-Gaussian or even fourth-moment assumptions, as they depend only on the behavior of the data in low-dimensional aligned subspaces, precisely where the gradients are concentrated due to the Tucker decomposition. Under Assumption \ref{asmp:1}, if all $\bm{U}_k$'s lie in the neighborhood of radius $\delta$ around their ground truth, all entries of $\overline{\cm{E}}_i$ and $\overline{\cm{X}}_{i}$ have a finite $(1+\epsilon)$-th moment bounded by $M_{e,1+\epsilon,\delta}$ and $M_{x,1+\epsilon,\delta}$, respectively.

Denote the estimator obtained by the robust gradient descent algorithm with gradient truncation parameter $\tau$ as $\cm{\widehat{A}}$, and the corresponding estimation error by $\textup{Err}(\widehat{\bm{F}})$ as in \eqref{eq:est_error}. Based on the bounded local moment conditions, we have the following guarantees.

\begin{theorem}\label{thm:linearregression}
    For tensor linear regression in \eqref{eq:linearregression}, suppose Assumption \ref{asmp:1} holds with the radius satisfying $\delta\geq\min\{\bar{\sigma}^{-1/(d+1)}\sqrt{\textup{Err}^{(0)}}+\kappa^2\alpha_x^{-1}\bar{\sigma}^{-1}d_\textup{eff}^{1/2}[M_{\textup{eff},1+\epsilon,\delta}\log(\bar{p})/n]^{\epsilon/(1+\epsilon)},1\}$. If the truncation parameter $\tau$ satisfies $\tau\asymp\bar{\sigma}^{d/(d+1)}[nM_{\textup{eff},1+\epsilon,\delta}/\log(\bar{p})]^{1/(1+\epsilon)}$, the sample size $n$ satisfies
    \begin{equation}\label{eq:condition1}
        n\gtrsim \left[\sqrt{\bar{p}}\alpha_x^{-1}\kappa^2 M_{x,2+2\lambda}\bar{\sigma}^\lambda\right]^{\frac{1+\max(\lambda,\epsilon)}{\lambda}}\log(\bar{p}),
    \end{equation}
    and the conditions of $a$, $b$, and $\eta$ in Theorem \ref{thm:1} hold with $\alpha=\alpha_{x}/2$ and $\beta=\beta_{x}/2$, then with probability at least $1-C\exp(-C\log(\bar{p}))$, after sufficient iterations of Algorithm \ref{alg:1},
    \begin{equation}\label{eq:error_1}
        \textup{Err}(\widehat{\bm{F}}) \lesssim \kappa^4\alpha_x^{-2}\bar{\sigma}^{-2d/(d+1)}d_\textup{eff}\left[M_{\textup{eff},1+\epsilon,\delta}^{1/\epsilon}\log(\bar{p})/n\right]^{2\epsilon/(1+\epsilon)},
    \end{equation}
    and
    \begin{equation}\label{eq:error2}
        \|\cm{\widehat{A}}-\cm{A}^*\|_\textup{F} \lesssim\ \kappa^2\alpha_{x}^{-1}d_{\textup{eff}}^{1/2}\left[M_{\textup{eff},1+\epsilon,\delta}^{1/\epsilon}\log(\bar{p})/n\right]^{\epsilon/(1+\epsilon)},
    \end{equation}
    where $M_{\textup{eff},1+\epsilon,\delta}=M_{x,1+\epsilon,\delta}\cdot M_{e,1+\epsilon,\delta}$ is the effective $(1+\epsilon)$-th local moment.
\end{theorem}

The statistical guarantees for model \eqref{eq:linearregression} are summarized in Theorem \ref{thm:linearregression}, which provides bounds on both the estimation error $\textup{Err}(\widehat{\bm{F}})$ and the Frobenius norm error $\|\cm{\widehat{A}}-\cm{A}^*\|_\text{F}$. The conditions and results depend critically on $\lambda$ and $\epsilon$, namely the moment order of the covariates and noise. Specifically, the sample size requirement in \eqref{eq:condition1} depends on $\lambda$ and $\epsilon$. Given other parameters fixed, we need $n\gtrsim \bar{p}^{(1+\max(\lambda,\epsilon))/(2\lambda)}$. Larger values of $\lambda$ (i.e., higher-order moment availability) relax the sample complexity, while smaller $\lambda$ lead to strong requirement but still valid estimation for the covariates in heavy-tailed regimes. 

On the other hand, the rate of convergence is solely governed by $\epsilon$. When $\epsilon=1$, the noise has finite second local moment, and the estimator attains a fast rate, matching those under Gaussian conditions \citep{,raskutti2019convex,han2022optimal}. When $\epsilon<1$, the noise exhibit heavier tails, and the rate slows down but remains minimax optimal for the given moment condition \citep{sun2020adaptive,tan2023sparse}. The truncation threshold $\tau$ is chosen adaptively based on the effective noise level and sample size, ensuring that the truncated gradients remain statistically well-behaved. In practice, $\tau$ can be chosen via cross-validation, where the detailed implementations are provided in Appendix \ref{append:init} of supplementary materials.

Our analysis relies on local moment conditions, which capture the tail behavior of the data in the low-dimensional subspaces defined by the true tensor factors. The localization leads to moment assumptions that are substantially weaker than global ones, and consequently, our statistical guarantees are potentially much sharper. These advantages are empirically validated through simulation experiments in Section \ref{sec:4}.

\begin{remark}
    Compared with the existing methods, we relax the distributional condition on the covariates. For example, Huber regression is widely-used for robust estimation of linear regression, with the loss function
    \begin{equation}\label{eq:huber}
        \mathcal{L}_\textup{H}(\bm{F};\mathcal{D}_n)=\frac{1}{2}\sum_{i=1}^n\ell_\nu(\cm{Y}_i-\langle\cm{S}\times_{j=1}^d\bm{U}_j,\cm{X}_i\rangle),
    \end{equation}
    where $\ell_\nu(\cm{T})=\sum_{i_1,\dots,i_d}\ell_\nu(\cm{T}_{i_1\dots i_d})$ for any tensor $\cm{T}$, $\ell_\nu(x)=x^2\cdot1(|x|\leq\nu)+(2\nu x-\nu^2)\cdot1(|x|>\nu)$ is the Huber loss, and $\nu>0$ is the robustness parameter. The partial gradients of $\mathcal{L}_\textup{H}$ are
    \begin{equation}
        \begin{split}
            \nabla_{\bm{U}_k}\mathcal{L}_{\textup{H}}(\bm{F};\mathcal{D}_n)&=\frac{1}{n}\sum_{i=1}^n[\cm{X}_i\circ\textup{T}(\langle\cm{A},\cm{X}_i\rangle-\cm{Y}_i,\nu)]_{(k)}(\otimes_{j\neq k}\bm{U}_j)\cm{S}_{(k)}^\top,\\
            \text{and }\nabla_{\scalebox{0.7}{\cm{S}}}\mathcal{L}_{\textup{H}}(\bm{F};\mathcal{D}_n)&=\frac{1}{n}\sum_{i=1}^n[\cm{X}_i\circ\textup{T}(\langle\cm{A},\cm{X}_i\rangle-\cm{Y}_i,\nu)]\times_{j=1}^d\bm{U}_j^\top,
        \end{split}
    \end{equation}
    where the gradients bound the residuals $(\langle\cm{A},\cm{X}_i\rangle-\cm{Y}_i)$ solely, and have no control on the covariates $\cm{X}_i$. Hence, the covariates are typically assumed to be sub-Gaussian or bounded \citep{fan2017estimation,sun2020adaptive,tan2023sparse,shen2022computationally}. In contrast, the proposed method, based on gradient robustification, can handle both heavy-tailed covariates and noise without stringent moment conditions on the covariates themselves. 
\end{remark}

The convergence guarantees in Theorem \ref{thm:1} depend on a good initialization error $\text{Err}(\bm{F}^{(0)})$. To initialize the estimate for tensor linear regression, we propose to reformulate it to
\begin{equation}
    \text{vec}(\cm{Y}_i) = \text{mat}(\cm{A}^*)\text{vec}(\cm{X}_i) + \text{vec}(\cm{E}_i),\quad i=1,\dots,n,
\end{equation}
where $\text{mat}(\cdot)$ is a tensor matricization. Due to the low-rank structure of $\cm{A}^*$, $\text{mat}(\cm{A}^*)$ is a low-rank matrix. Similarly to \citet{sun2020adaptive} and \citet{fan2021shrinkage}, we perform data truncation to the covariates and apply the reduced-rank Huber regression model in \eqref{eq:huber} with a nuclear norm penalty by \citet{tan2023sparse}. After obtaining the initial value of $\cm{A}$, we apply the higher-order orthogonal iterations \citep[HOOI]{delathauwer2000multilinear} to obtain the intial values of $\bm{F}$. The details of initialization and correponding theoretical guarantees are relegated to Appendix \ref{append:init} of supplementary materials.

\subsection{Heavy-Tailed Tensor Logistic Regression}

For the generalized linear model, conditioned on the tensor covariate $\cm{X}_i$, the response variable $y_i$ follows the distribution
\begin{equation}\label{eq:GLM}
    \mathbb{P}(y_i|\cm{X}_i)\propto\exp\left\{\frac{y_i\langle\cm{X}_i,\cm{A}^*\rangle-\Phi(\langle\cm{X}_i,\cm{A}^*\rangle)}{c(\gamma)}\right\},~~i=1,2,\dots,n,
\end{equation}
where $\Phi(\cdot)$ is a convex link function, and $c(\gamma)$ is a normalization constant that may depend on additional parameters $\gamma$. The corresponding negative log-likelihood loss function is
\begin{equation}
    \overline{\mathcal{L}}(\cm{A};z_i)=\Phi(\langle\cm{X}_i,\cm{A}\rangle) - y_i\langle\cm{X}_i,\cm{A}\rangle.
\end{equation}

A widely studied instance of this model is logistic regression, where $\Phi(t)=\log(1+\exp(t))$. For this model, since $y_i$ is a binary variable, we assume that the covariate $\cm{X}_i$ may follow a heavy-tailed distribution. Similarly to tensor linear regression, for a given $\bm{F}$, we consider the multilinear transformations induced by the factor matrices: $\overline{\cm{X}}_i=\cm{X}_i\times_{j=1}^d\bm{U}_j^\top$ and $\overline{\cm{X}}_{i,k}=\cm{X}_i\times_{j=1,j\neq k}^d\bm{U}_j^\top$. These transformations project the original covariate tensors onto the low-dimensional subspaces defined by the factor matrices $\bm{U}_j$, which align with the low-rank structure of the coefficient tensor $\cm{A}^*$. The partial gradients of the logistic loss with respect to the factor matrices and core tensor are
\begin{equation}
    \begin{split}    
        \nabla_{\bm{U}_k}\mathcal{L}(\bm{F};z_i) & = \left(\frac{\exp(\langle\overline{\cm{X}}_i,\cm{S}\rangle)}{1+\exp(\langle\overline{\cm{X}}_i,\cm{S}\rangle)}-y_i\right)(\overline{\cm{X}}_{i,k})_{(k)}\cm{S}_{(k)}^\top,~1\leq k\leq d_0,\\
        \text{and }\nabla_{\scalebox{0.7}{\cm{S}}}\mathcal{L}(\bm{F};z_i) & = \left(\frac{\exp(\langle\overline{\cm{X}}_i,\cm{S}\rangle)}{1+\exp(\langle\overline{\cm{X}}_i,\cm{S}\rangle)}-y_i\right)\overline{\cm{X}}_i.
    \end{split}
\end{equation}

As in the case of tensor linear regression, the partial gradients depend on the low-dimensional transformations $\overline{\cm{X}}_i$ and $\overline{\cm{X}}_{i,k}$. Therefore, to derive sharp statistical guarantees, it is essential to characterize their distributional properties. In this article, we impose the following local moment conditions on the covariate tensor $\cm{X}_i$.
\begin{assumption}\label{asmp:logistic}
    For some $\lambda\in(0,1]$ and $\delta\in[0,1]$, $\cm{X}_i$ satisfies:
    \begin{itemize}
        \item[(a)] $\textup{vec}(\cm{X}_i)$ has mean zero and a positive definite variance matrix $\bm{\Sigma}_x$, with $0<\alpha_x\leq\lambda_{\min}(\bm{\Sigma}_x)\leq \lambda_{\max}(\bm{\Sigma}_x)\leq\beta_x$.
        \item[(b)] $\cm{X}_i$ has the $(2+2\lambda)$-th global moment $M_{x,2+2\lambda}=\max_{0\leq k\leq d}[\textup{LM}_{k}(\cm{X}_i;2+2\lambda,1,\{\bm{U}_j^*\}_{j=1}^{d})]$, and has the second local moment $M_{x,2,\delta}=\max_{0\leq k\leq d}[\textup{LM}_{k}(\cm{X}_i;2,\delta,\{\bm{U}_j^*\}_{j=1}^{d})]$.
    \end{itemize}
\end{assumption}

By definition, the local second moment $M_{x,2,\delta}$ is typically much smaller than the global variance bound $\beta_x$. The statistical convergence rate of the robust estimator, denoted as $\cm{\widehat{A}}$, is governed by the local moment $M_{x,2,\delta}$, while the $(2+2\lambda)$-th global moment $M_{x,2+2\lambda}$ influences the sample size requirement.

\begin{theorem}\label{thm:logistic}
    For low-rank tensor logistic regression, suppose Assumption \ref{asmp:1} holds with some $\epsilon\in(0,1]$ and $\delta\geq\min\{\bar{\sigma}^{-1/(d+1)}\sqrt{\textup{Err}^{(0)}}+C\kappa^2\alpha_x^{-1}\bar{\sigma}^{-1}\sqrt{d_\textup{eff}M_{\textup{eff},\delta}\log(\bar{p})/n},1\}$. If
    \begin{equation}
        \tau\asymp\bar{\sigma}^{d/(d+1)}[nM_{x,2,\delta}/\log(\bar{p})]^{1/2},\quad n\gtrsim \bar{p}^{1/\lambda}\alpha_x^{-2/\lambda}\kappa^{4/\kappa}M_{x,2+2\lambda}^{2/\lambda}\bar{\sigma}^2\log(\bar{p}),
    \end{equation}
    and other conditions of $a$, $b$ and $\eta$ in Theorem \ref{thm:1} hold with $\alpha=\alpha_x/2$ and $\beta=\beta_x/2$, then with probability at least $1-C\exp(-C\log(\bar{p}))$, after sufficient iterations of Algorithm \ref{alg:1},
    \begin{equation}
        \textup{Err}(\widehat{\cm{S}},\widehat{\bm{U}}_1,\dots,\widehat{\bm{U}}_d) \lesssim \kappa^4\alpha_x^{-2}\bar{\sigma}^{-2d/(d+1)}d_{\textup{eff}}M_{x,2,\delta}\log(\bar{p})/n.
    \end{equation}
    and
    \begin{equation}
        \|\cm{\widehat{A}}-\cm{A}^*\|_\textup{F}\lesssim \kappa^2\alpha_x^{-1}d_{\textup{eff}}^{1/2}\sqrt{M_{x,2,\delta}\log(\bar{p})/n}.
    \end{equation}
\end{theorem}

Theorem \ref{thm:logistic} establishes sharp statistical convergence rates for robust tensor logistic regression under local moment assumptions on covariates. Similarly to Theorem \ref{thm:linearregression}, the sample size requirement depends on $\lambda$. Notably, the convergence rate matches that of the vanilla gradient descent algorithm under Gaussian design \citep{chen2019non}, demonstrating that our method achieves robust and optimal estimation even when the covariates are heavy-tailed. Moreover, if the initial estimation error satsifies $\text{Err}^{(0)}\leq\bar{\sigma}^{2/(d+1)}$ and the sample size $n$ is sufficiently large, the local radius $\delta$ remains less than one, ensuring that the statistical guarantees hold under the local moment conditions.

\begin{remark}

    The proposed method extends to tensor logistic regression and improves over existing approaches \citep{prasad2020robust,zhu2021taming} by requiring only a local $(2+2\lambda)$-th moment condition on the covariates, rather than fourth or higher moments. By leveraging the low-rank structure, the statistical behavior of the gradients is governed by the projected covariates within the low-dimensional subspaces defined by the true factors. This localization enables relaxed moment assumptions and improved convergence rates tailored to tensor models.

\end{remark}

For initialization of tensor logistic regression, we propose to vectorize the tensor covariates, perform vector norm truncation to them, apply the robust estimation similar to \citet{zhu2021taming}, and perform tensor HOOI to initialize $\bm{F}$. The detailed implementation and theoretical guarantees are relegated to Appendix \ref{append:init} of supplementary materials.

\subsection{Heavy-Tailed Tensor PCA}

Another important statistical model for tensor data is tensor principal component analysis (PCA). Specically, we consider the model
\begin{equation}\label{eq:PCA}
    \cm{Y} = \cm{A}^* + \cm{E} = \cm{S}^*\times_{j=1}^d\bm{U}_j^* + \cm{E},
\end{equation}
where $\cm{Y}\in\mathbb{R}^{p_1\times\cdots\times p_d}$ is the observed tensor, $\cm{A}^*=\cm{S}^*\times_{j=1}^d\bm{U}_j^*$ is the low-rank signal tensor, and $\cm{E}$ is a mean-zero noise tensor. In the existing literature, most theoretical analyses of tensor PCA focus on  Gaussian or sub-Gaussian noise \citep{richard2014statistical,zhang2019optimal,han2022optimal}.

In this subsection, we consider the setting where the noise tensor $\cm{E}$ is heavy-tailed. We propose estimating the low-rank signal $\cm{A}^*$ using robust gradient descent with truncated gradient estimators. The loss function for tensor PCA is given by
\begin{equation}
    \mathcal{L}(\cm{S},\bm{U}_1,\dots,\bm{U}_d;\cm{Y})=\|\cm{Y}-\cm{S}\times_{j=1}^d\bm{U}_j\|_\text{F}^2/2.
\end{equation}
The partial gradient with respect to the factor matrices and the core tensor are
\begin{equation}
    \begin{split}    
        \nabla_{\bm{U}_k}\mathcal{L}(\cm{S},\bm{U}_1,\dots,\bm{U}_d) & =  (\cm{S}\times_{j=1,j\neq k}^d\bm{U}_j^\top\bm{U}_j\times_{k}\bm{U}_k-\cm{Y}\times_{j=1,j\neq k}^d\bm{U}_j^\top)_{(k)}\cm{S}_{(k)}^\top,~~k=1,\dots,d,\\
        \text{and  }\nabla_{\scalebox{0.7}{\cm{S}}}\mathcal{L}(\cm{S},\bm{U}_1,\dots,\bm{U}_d) & =\cm{S}\times_{j=1}^d\bm{U}_j^\top\bm{U}_j-\cm{Y}\times_{j=1}^d\bm{U}_j^\top.
    \end{split}
\end{equation}

These gradients are computed via multilinear transformation of the observed tensor onto the subspaces spanned by the estimated factor matrices, specifically $\cm{Y}\times_{j=1,j\neq k}^d\bm{U}_j^\top$ and $\cm{Y}\times_{j=1}^d\bm{U}_j^\top$. Consequently, the statistical behavior of the gradients depend primarily on the projected noise components in these subspaces, rather than the ambient noise distribution. To ensure robustness in the presence of heavy-tailed noise, we impose a local $(1+\epsilon)$-th moment condition on the noise tensor $\cm{E}$, as formalized in the following assumption.
\begin{assumption}\label{asmp:PCA}
    For some $\epsilon\in(0,1]$ and $\delta\in[0,1]$, $\cm{E}$ has the local $(1+\epsilon)$-th moment
    \begin{equation}
        M_{e,1+\epsilon,\delta} = \max_{0\leq k\leq d}[\textup{LM}_k(\cm{E};1+\epsilon,\delta,\{\bm{U_j^*}\}_{j=1}^d)].
    \end{equation}
\end{assumption}

In constrast to many existing statistical analyses of tensor PCA, our method does not require the entries of the random noise $\cm{E}$ to be independent or idetically distributed. This is a key feature of our approach, as it allows us to handle more general noise structures, including those with dependencies and heavy tails.

For the estimator obtained by the robust gradient descent, denoted as $\cm{\widehat{A}}$, as well as the  estimation error $\textup{Err}(\widehat{\cm{S}},\widehat{\bm{U}}_1,\dots,\widehat{\bm{U}}_d)$, we have the following convergence guarantees.

\begin{theorem}\label{thm:PCA}
    For tensor PCA in \eqref{eq:PCA}, suppose Assumption \ref{asmp:PCA} holds with some $\epsilon\in(0,1]$ and $\delta\geq{\min}(\bar{\sigma}^{-1/(d+1)}\sqrt{\textup{Err}^{(0)}}+C\bar{\sigma}^{-1}d_\textup{eff}^{1/2}M_{e,1+\epsilon,\delta}^{1/(1+\epsilon)},1)$. If the truncation parameter $\tau$ satisfies $\tau\asymp\kappa^{2/\epsilon}\bar{\sigma}^{d/(d+1)}M_{e,1+\epsilon,\delta}^{1/(1+\epsilon)}$, the minimal signal strength $\underline{\sigma}$ satisfies 
    \begin{equation}\label{eq:SNR}
        \underline{\sigma}/M_{e,1+\epsilon,\delta}^{1/(1+\epsilon)}\gtrsim\sqrt{\bar{p}},
    \end{equation}
    and other conditions of $a$, $b$, and $\eta$ in Theorem \ref{thm:1} hold with $\alpha=\beta=1/2$, then with probability at least $1-C\exp(-C\bar{p})$, after sufficient iterations of Algorithm \ref{alg:1},
    \begin{equation}
        \textup{Err}(\widehat{\cm{S}},\widehat{\bm{U}}_1,\dots,\widehat{\bm{U}}_d) \lesssim \bar{\sigma}^{-2d/(d+1)}d_{\textup{eff}}M_{e,1+\epsilon,\delta}^{2/(1+\epsilon)}
    \end{equation}
    and
    \begin{equation}
        \|\cm{\widehat{A}}-\cm{A}^*\|_\textup{F}\lesssim d_{\textup{eff}}^{1/2}M_{e,1+\epsilon,\delta}^{1/(1+\epsilon)}.
    \end{equation}
\end{theorem}

Under the local $(1+\epsilon)$-th moment condition for the noise tensor $\cm{E}$, the convergence rate of the proposed robust gradient descent method is shown to be comparable to that of vanilla gradient descent under Gaussian noise \citep{ZX18}, and achieves minimax optimality \citep{han2022optimal}. Specically, when $\epsilon=1$, the signal-to-noise ratio (SNR) requirement in \eqref{eq:SNR} is identical to the SNR condition under the sub-Gaussian noise setting \citep{ZX18}. This demonstrates that our method is capable of effectively handling heavy-tailed noise, while still achieving optimal statistical performance. Furthermore, similar to tensor linear regression and logistic regression, if the signal strength satisfies $\bar{\sigma}\gtrsim\sqrt{\bar{p}}$ and the initial error is sufficiently small (i.e., $\textup{Err}^{(0)}<\bar{\sigma}^{2/(d+1)}$), then the local radius $\delta$ remains below one all along the iterations. This demonstrates that our robust gradient framework is not limited to supervised learning, but also enables reliable unsupervised tensor analysis under minimal assumptions.

\begin{remark}
    Our method accommodates noise tensors $\cm{E}$ with only a $(1+\epsilon)$-th moment, relaxing the common sub-Gaussian assumption. Although the noise may be correlated, only the projection onto local low-dimensional regions, characterized by $M_{e,1+\epsilon,\delta}$, affects estimation. Furthermore, $M_{e,1+\epsilon,\delta}$ may grow unbounded, allowing for large noise magnitudes in localized regions, provided the signal-to-noise ratio satisfies $\underline{\sigma}/M_{e,1+\epsilon,\delta}^{1/2}\gtrsim\sqrt{\bar{p}}$ to ensure consistent estimation.
\end{remark}

\section{Simulation Experiments}\label{sec:4}

In this section, we conduct four simulation experiments to validate the theoretical insights from Section \ref{sec:3} and to empirically demonstrate the advantages of the proposed robust gradient descent (RGD) method over existing approaches. We focus on the tensor linear regression as a primary case study in the main text, with extensions to tensor logistic regression and PCA provided in Appendix \ref{append:numerical} of the supplementary materials. We consider two tensor linear regression models.
\begin{equation}\label{eq:sim_model1}
    \text{Model I:}\quad y_i = \langle\cm{A}^*,\cm{X}_i\rangle + e_i,\quad i=1,\dots,n,
\end{equation}
where $\cm{X}_i\in\mathbb{R}^{10\times 10\times 10}$ is a tensor covariate, $y_i$ and $e_i$ are scalar response and noise.
\begin{equation}\label{eq:sim_model2}
    \text{Model II:}\quad \bm{y}_i = \langle\cm{A}^*,\bm{X}_i\rangle + \bm{e}_i,\quad i=1,\dots,n,
\end{equation}
where $\bm{X}_i\in\mathbb{R}^{10\times 10}$ is the matrix covariate , $\bm{y}_i,\bm{e}_i\in\mathbb{R}^{10}$ are the vector response and noise. In both models, we set the coefficient tensor as $\cm{A}^*=\sqrt{10}\cdot\bm{1}_{10}\circ\bm{1}_{10}\circ\bm{1}_{10}=\cm{S}^*\times_{j=1}^3\bm{U}_j^*$.

The first three experiments are designed to verify how the tail behaviors of the covariates and noise, quantified by $\lambda$ and $\epsilon$, as well as the local moment, are related to the computational and statistical performance of the proposed method. The last experiment includes a comparative study between RGD and competing methods, including vanilla gradient descent (VGD) and Huber regression (HUB) in \eqref{eq:huber}, to assess robustness in real-world settings. 

\subsection{Experiment 1: Dependence on Tail Behavior of Covariates} 

In both models, we consider that all entries in $\cm{X}_i$ (or $\bm{X}_i$) are independent and follow the Student's $t_{2+2\lambda}$ distibution, and all entries in $\bm{e}_i$ (or $e_i$) are independent and follow the $t_{1.5}$ distribution. We vary $\lambda\in\{0.1,0.4,0.7,1.0,1.3,1.6\}$ and set the sample size as $n=10\times 2^{m}$, where $m\in\{1,2,3,4,5\}$. For the generated data, we apply the proposed RGD method with initial values set to the ground truth, $a=b=1$, step size $\eta=10^{-3}$, truncation threshold $\tau=\sqrt{n/\log(\bar{p})}$, and number of iterations $T=300$. 

In this experiment, we aim to verify whether the RGD iterates converge and to explore the relationship between the emprical convergence rate and $\lambda$. According to Theorem \ref{thm:linearregression}, if the iterates converge, then $\|\cm{A}^{(t)}-\cm{A}^*\|_\text{F}^2$ lie in a region with radius smaller than $M_{\text{eff},2,\delta}^{1/2}d_{\text{eff}}\log(\bar{p})/n$. To empirically assess convergence, we compute the sample standard deviation of $\|\cm{A}^{(t)}-\cm{A}^*\|_\text{F}^2$ over iterations $t=251,\dots,300$, and label the algorithm as having converged only if this quantity is smaller than $\bar{p}\log(\bar{p})/(100n)$. 

For each pair of $\lambda$ and $m$, we replicate the entire procedure  200 times and summarize the proportion of replications that achieve convergence versus $m$ in Figure \ref{fig:Exp1}. The results confirm that the smaller value of $\lambda$, corresponding to heavy-tailed covariates, leads to a greater sample size requirement for convergence. However, for $\lambda\geq1$, the convergence patterns across different $m$ are similar, which is consistent with the theoretical sample size requirement derived in Theorem \ref{thm:linearregression}.

\subsection{Experiment 2: Dependence on Tail Behavior of Noise} 

In both models, we consider that all entries of the covariates follow either a standard Gaussian distribution or a $t_3$ distribution, and all entries of the noise follow a $t_{1+\epsilon}$ distribution. We vary $\epsilon\in\{0.1,0.4,0.7,1.0,1.3,1.6\}$ and set the sample size as $n=200\times 2^{m}$, where $m\in\{1,2,3,4,5\}$. For the generated data, we apply the proposed RGD method with the same tuning parameters as in Experiment 1, except that the truncation threshold is adjusted to $\tau=(n/\log(\bar{p}))^{1/(1+\epsilon_{\text{eff}})}$, where $\epsilon_{\text{eff}}=\min(1,\epsilon)$. According to Theorem \ref{thm:linearregression}, after a sufficent number of iterations, $-\log(\|\cm{A}^{(T)}-\cm{A}^*\|_\text{F}^2) = C(\bar{p},\epsilon) + C[\epsilon_{\text{eff}}/(1+\epsilon_{\text{eff}})]m$, where $C(\bar{p},\epsilon)$ is a constant depending on $\bar{p}$ and $\epsilon$. 

Therefore, for each pair of $\epsilon$ and $m$, we replicate the procedure 200 times and summarize the average of negative log errors versus $m$ in Figure \ref{fig:Exp2}. For each value of $\epsilon$, the average negative log errors exhibit a linear relationship with respect to $m$. Notably, the slope of this linear relationship shows a smooth transition: when $\epsilon\in(0,1)$, the slope increases as $\epsilon$ increases; when $\epsilon\geq1$, the slopes stablize. These empirical findings verify the smooth transition in statistical convergence rate as stated in Theorem \ref{thm:linearregression}.

\subsection{Experiment 3: Dependence on Local Moment Conditions}

We consider the vectorized covariate $\text{vec}(\cm{X}_i)$ (or $\text{vec}(\bm{X}_i)$) follows a multivariate Gaussian distribution with mean zero and covariance $(\otimes_{j=1}^{d_0}\bm{\Sigma}_\theta)$, where $\bm{\Sigma}_\theta=0.5\bm{I}_{10}+0.5\bm{v}_\theta\bm{v}_\theta^\top$, where $\bm{v}_\theta=\sin(\theta)\bm{1}_{10} + \cos(\theta)\bm{w}$ and $\bm{w}=(1,-1,1,-1,\dots,1,-1)^\top\in\mathbb{R}^{10}$. For Model I, the noise term $e_i$ follows a standard Gaussian distribution. For Model II, the vectorized noise $\text{vec}(\bm{e}_i)$ follows a multivariate Gaussian distribution with mean zero and covariance $\bm{\Sigma}_\theta$. In this setup, the entries in covariates or noise are dependent, and the dependency is governed by the angle parameter $\theta\in[0,\pi/2]$. Specifically, when $\theta=\pi/2$, the vector $\bm{v}_\theta$ aligns with $\bm{1}_{10}$, which coincides with the true factor directions $\bm{U}_1^*=\bm{U}_2^*=\bm{U}_3^*$, resulting in a large local moment condition. When $\theta=0$, the correlation direction $\bm{v}_\theta=\bm{w}$ is orthogonal to the true factors, leading to a much smaller local moment. Thus, in this experiment, the local moment of the data varies with $\theta$, while the global moment remains unchanged. The details of local moment computing are relegated in the supplementary materials.

We consider $\theta=\theta_0\pi/8$ with $\theta_0\in\{0,1,2,3,4\}$ and set $n\in\{300,400,500,600,700\}$. For each pair of $\theta_0$ and $n$, we replicate the procedure 200 times and summarize the average of $\|\cm{A}^{(T)}-\cm{A}^*\|_\text{F}^2$ versus $n$ in Figure \ref{fig:Exp3}. As $\theta_0$ increases, the local moments increase, and the average estimation errors increase accordingly, further validating the importance of leveraging local moment conditions as emphasized in our theoretical analysis.

\subsection{Experiment 4: Comparison with Other Methods}

For both models, four distributional cases are adopted: (1) $N(0,1)$ covariate and $N(0,1)$ noise; (2) $N(0,1)$ covariate and $t_{1.2}$ noise; (3) $t_{2.1}$ covariate and $N(0,1)$ noise; and (4) $t_{2.1}$ covariate and $t_{1.2}$ noise. All entries in covariates and noise are independent, and we set $n=500$. We apply the proposed RGD algorithm, as well as the vanilla gradient descent (VGD) and adaptive Huber regression (HUB) as competitors, to the data generated from each model. For all methods, intial values are obtained in a data-driven manner as suggested in Appendix \ref{append:init} of the supplentary materials. We set $a=b=1$, $\eta=10^{-3}$, $T=300$, and the truncation parameter $\tau$ is selected via five-fold cross-validation. 

For each model and distributional setting, we replicate the procedure 200 times and summarize the average of $\log(\|\cm{A}^{(T)}-\cm{A}^*\|_\text{F}^2)$, as well as their upper and lower quartiles, for the above four cases in Figure \ref{fig:Exp4}. When both the covariate and noise are light-tailed, the performances of three estimation methods are nearly identical. However, in heavy-tailed cases, the performance of VGD deteriorates significantly, with estimation errors much larger than those of the other two methods. Overall, the RGD method consistently yields the smallest estimation errors across all three methods. These numerical findings confirm the robustness and efficiency of the proposed method in handling heavy-tailed data.

\section{Real Data Example: Chest CT Images}\label{sec:5}

In this section, we apply the proposed robust gradient descent (RGD) estimation approach to the publicly available COVID-CT dataset \citep{yang2020covid} (the data can be downloaded from \href{ https://github.com/UCSD-AI4H/COVID-CT}{ https://github.com/UCSD-AI4H/COVID-CT}), which consists of chest CT scans collected for COVID-19 diagnosis. The dataset includes 317 COVID-19 positive scans and 397 negative scans, sourced from four open-access databases. Each scan is a $150 \times 150$ greyscale image with a binary label indicating the disease status. 

Medical imaging data such as CT scans often exhibit non-Gaussian, heavy-tailed noise due to variability in imaging conditions, patient anatomy, and disease manifestation. This is supported by the empirical kurtosis analysis shown in Figure \ref{fig:CT_histogram}, which displays the distribution of pixel-level kurtosis values for COVID and non-COVID scans. The high kurtosis observed in both groups indicates substantial deviations from Gaussianity, suggesting that traditional methods relying on light-tailed assumptions may be suboptimal for this task.

To classify COVID-positive scans based on their visual characteristics, we employ a two-dimensional low-rank tensor logistic regression model $(d=2, p_1=p_2=150)$. To balance model flexibility and robustness, we impose a low-rank structure with Tucker ranks $r_1=r_2=5$. The rank selection was guided by preliminary analysis of the singular value spectra. We randomly partition the data into a training set (200 positive and 250 negative scans) and a test set (117 positive and 147 negative scans). Using this split, we compare the performance of the proposed robust gradient descent (RGD) algorithm with that of vanilla gradient descent (VGD), which corresponds to using untruncated gradients. Both methods are used to estimate the low-rank tensor logistic model parameters.

Using each estimation method, we classify the testing data into four categories: true positive (TP), false positive (FP), true negative (TN), and false negative (FN). The performance metrics used for evaluation include: precision rate: $P = \text{TP}/(\text{TP} + \text{FP})$; recall rate: $R = \text{TP}/(\text{TP} + \text{FN})$; and F1 score: $F_1 = 2/(P^{-1} + R^{-1})$.
The precision, recall, and F1 scores for the RGD method are reported in Table \ref{tbl:1}, alongside the performance of the VGD method as a benchmark.

The results, summarized in Table \ref{tbl:1}, demonstrate that the RGD method significantly outperforms VGD across all three metrics. In particular, RGD achieves a precision of 0.954, recall of 0.880, and F1 score of 0.916, compared to VGD’s precision of 0.898, recall of 0.829, and F1 score of 0.862. These improvements indicate that robust gradient descent leads to more reliable and stable inference, particularly in the presence of heavy-tailed noise and potential outliers in the imaging data.

\begin{table}[!htp]
    \caption{Classification performance of VGD and RGD on chest CT images}
    \label{tbl:1}
    \begin{center}
    \begin{tabular}{c|ccc}
        \toprule
        Method & Precision & Recall & F1 Score\\
        \midrule
        VGD & 0.898 & 0.829 & 0.862 \\
        RGD & \textbf{0.954} & \textbf{0.880} & \textbf{0.916} \\
        \bottomrule
    \end{tabular}
    \end{center}
\end{table}

These findings highlight the practical advantages of the proposed robust tensor estimation framework, particularly in real-world applications involving noisy, high-dimensional, and potentially heavy-tailed data. The ability of RGD to maintain high classification performance in the presence of distributional uncertainties underscores its value for medical imaging diagnostics and other domains where robustness is critical.

\section{Conclusion and Discussion}
\label{sec:6}

We propose a unified and computationally efficient framework for robust tensor estimation, based on gradient descent with entrywise gradient truncation. By stabilizing the gradient updates, rather than modifying the loss or preprocessing the data, we achieve distributional robustness under heavy-tailed noise and covariates, while maintaining statistical optimality and computational scalability.

Applied to tensor linear regression, logistic regression, and PCA, our method attains minimax-optimal error rates under mild local moment conditions, without requiring sub-Gaussian assumptions. The approach is flexible and can incorporate alternative robust gradient mechanisms, such as median-of-means or rank-based estimators. It is also applicable to broader contamination models and settings with structured outliers.

Though our method achieves robust estimation under relaxed moment conditions, including covariates with $(2+2\lambda)$-th moments and noise with $(1+\epsilon)$-th moments, establishing minimax optimality in this regime remains challenging. Due to the inherent difficulty in deriving tight lower bounds for tensor estimation under such heavy-tailed covariate assumptions, we leave the question of minimax optimality in this setting as an important direction for future research.

\begin{figure}[!htp]
    \begin{centering}
    \includegraphics[width=0.9\textwidth]{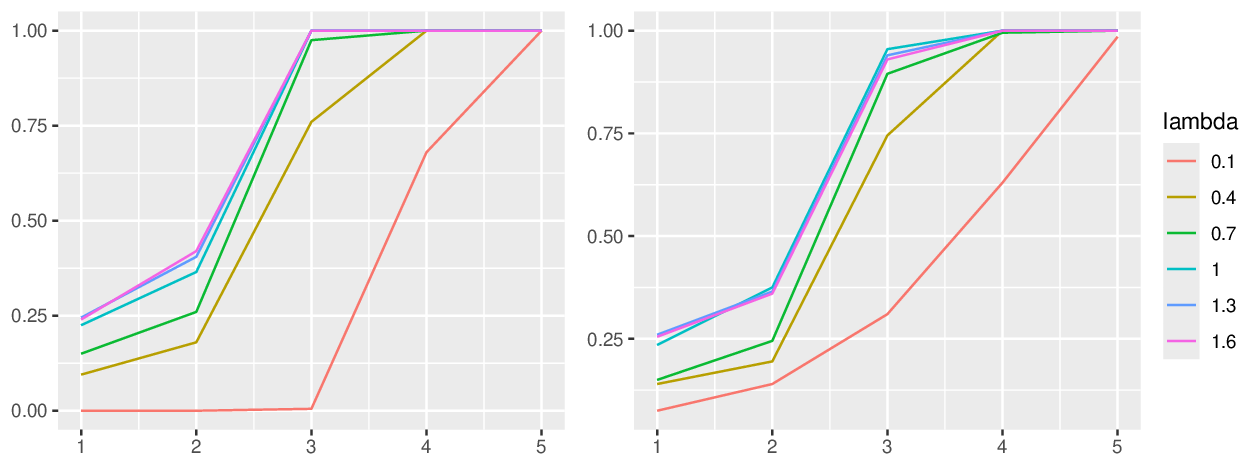}
    \vspace{-0.2cm}
    \caption{Average convergence proportion (y-axis) vs $m$ (x-axis) with varying $\lambda$ for Model I (left panel) and Model II (right panel) in Experiment 1}
    \label{fig:Exp1}
    \end{centering}
\end{figure}

\begin{figure}[!htp]
    \begin{centering}
    \includegraphics[width=0.9\textwidth]{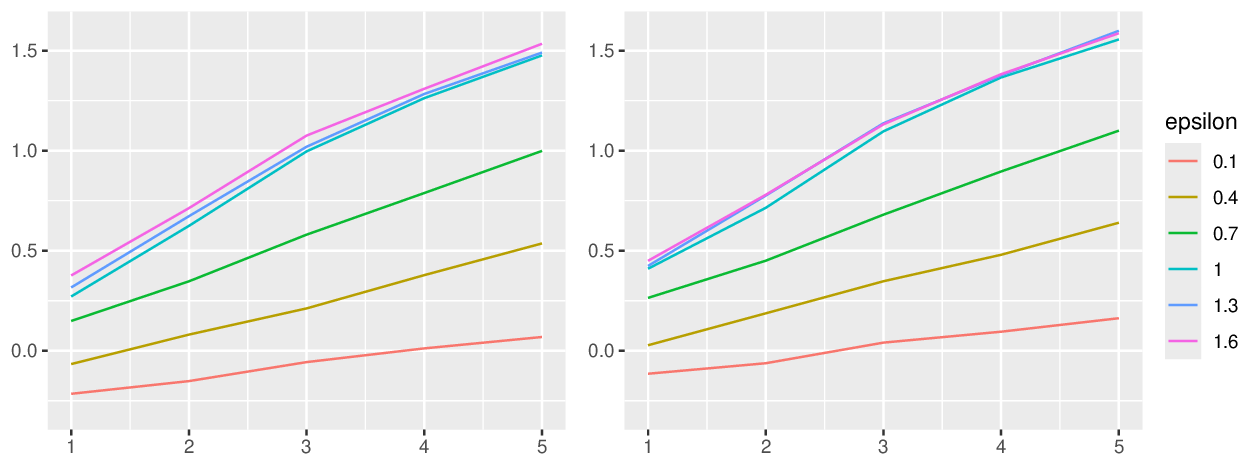}
    \vspace{-0.2cm}
    \caption{Average $-\log(\|\cm{\widehat{A}}-\cm{A}^*\|_\text{F}^2)$ (y-axis) vs $m$ (x-axis) with varying $\epsilon$ for Model I (left panel) and Model II (right panel) in Experiment 2}
    \label{fig:Exp2}
    \end{centering}
\end{figure}

\begin{figure}[!htp]
    \begin{centering}
    \includegraphics[width=0.9\textwidth]{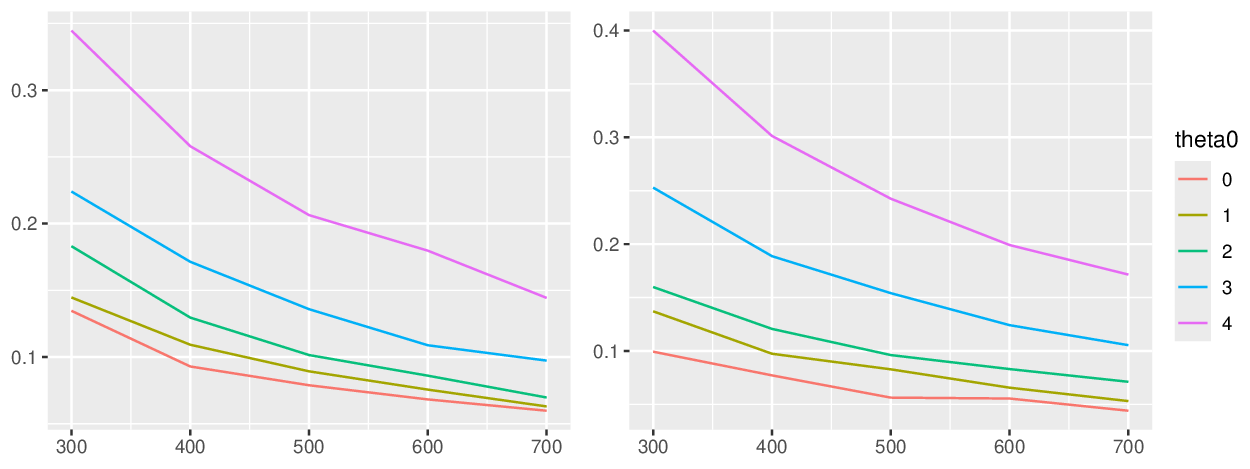}
    \vspace{-0.2cm}
    \caption{Average $\|\cm{\widehat{A}}-\cm{A}^*\|_\text{F}$ (y-axis) vs $n$ (x-axis) with varying $\theta_0$ for Model I (left panel) and Model II (right panel) in Experiment 3}
    \label{fig:Exp3}
    \end{centering}
\end{figure}

\begin{figure}[!htp]
    \begin{centering}
    \includegraphics[width=0.95\textwidth]{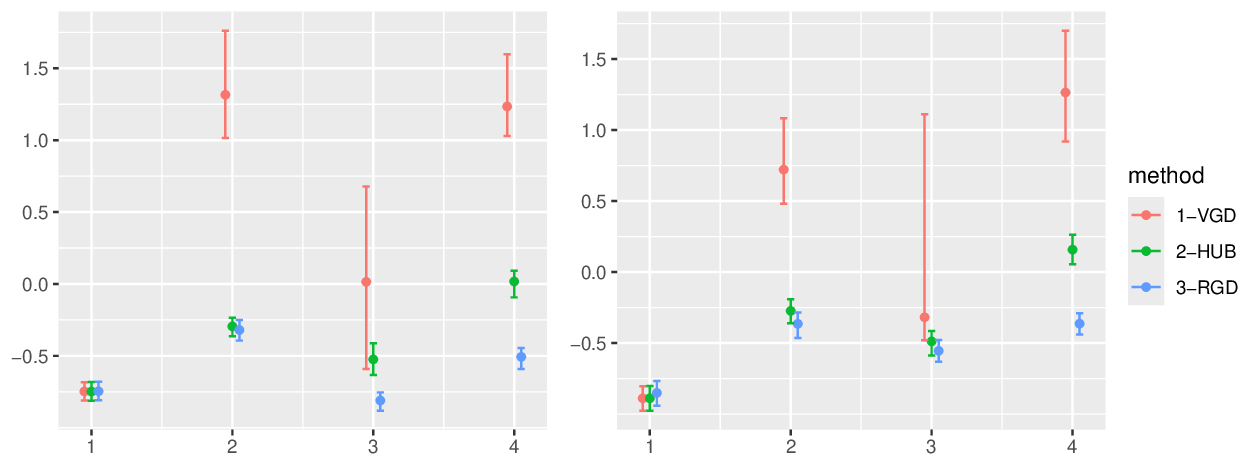}
    \vspace{-0.2cm}
    \caption{Average $\log(\|\cm{\widehat{A}}-\cm{A}^*\|_\text{F}^2)$ (y-axis) in different distributional cases (x-axis) by different methods for Model I (left panel) and Model II (right panel) in Experiment 4}
    \label{fig:Exp4}
    \end{centering}
\end{figure}

\begin{figure}[!htp]
    \begin{centering}
    \includegraphics[width=0.49\textwidth]{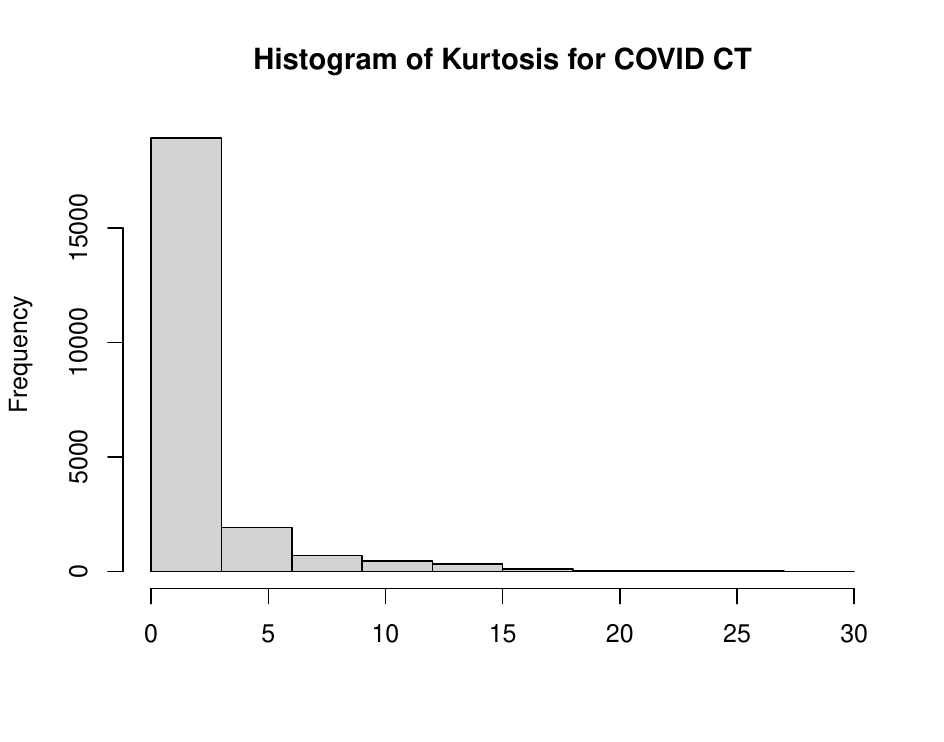}
    \includegraphics[width=0.49\textwidth]{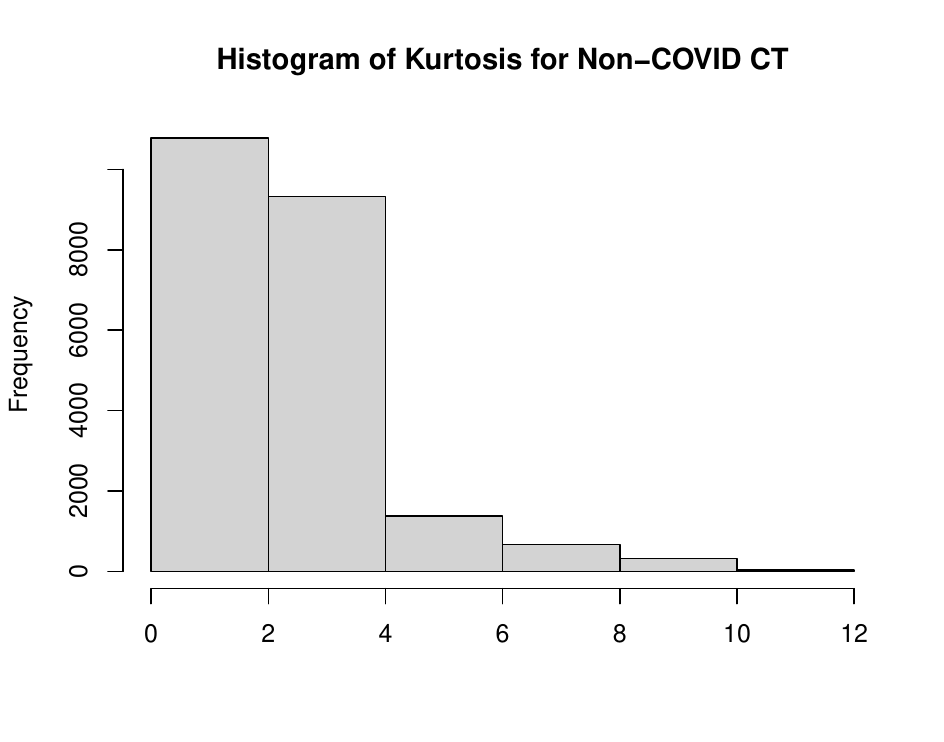}
    \vspace{-0.5cm}
    \caption{Histograms of kurtosis for COVID and non-COVID CT image pixels}
    \label{fig:CT_histogram}
    \end{centering}
\end{figure}

{

\bibliography{mybib.bib}

}

\newpage

\begin{center}
    \textbf{Supplementary Materials for ``Robust Gradient Descent Estimation for Tensor Models under Heavy-Tailed Distributions''}
\end{center}

This supplementary material provides all technical proofs of the theoretical results in the main article, as well as some discussions, examples, implementation details, and additionals numerical results. Specifically, the tensor algebra and notations are described in Appendix \ref{append:tensor_algebra_and_notations}. Discussions and examples of local moment conditions are provided in Appendix \ref{append:local_moments}. Computational and statistical analysis, particularlly the proofs of Theorems \ref{thm:1}-\ref{thm:PCA} are given in Appendices \ref{append:convergence} and \ref{append:B}. Initialization methods, as well as their theoretical guarantees, are given in Appendix \ref{append:init}. Additional simulation experiments and results, for tensor logistic regression and tensor PCA, are given in Appendix \ref{append:numerical}.

\begin{appendix}

\section{Tensor Algebra and Notations}
\label{append:tensor_algebra_and_notations}

Tensors are multi-dimensional arrays that generalize matrices to higher-order data. A $d$-th order tensor is represented as $\cm{X} \in \mathbb{R}^{p_1 \times p_2 \times \cdots \times p_d}$, where $p_k$ is the dimension along the $k$-th mode. In this article, we adopt the following notation:

\begin{itemize}
  \item \textbf{Vectors}: denoted by boldface lowercase letters, e.g., $\bm{x} \in \mathbb{R}^p$,
  \item \textbf{Matrices}: denoted by boldface uppercase letters, e.g., $\bm{X} \in \mathbb{R}^{p \times q}$,
  \item \textbf{Tensors}: denoted by boldface Euler letters, e.g., $\cm{X} \in \mathbb{R}^{p_1 \times \cdots \times p_d}$.
\end{itemize}

We refer readers to \citet{kolda2009tensor} for a comprehensive review of tensor operations and decompositions.

\paragraph{Mode-$k$ matricization}

The \textbf{mode-$k$ matricization} (or \textbf{unfolding}) of a tensor $\cm{X} \in \mathbb{R}^{p_1 \times \cdots \times p_d}$ is a matrix obtained by rearranging the fibers of $\cm{X}$ along the $k$-th mode into columns. The result is denoted by $\cm{X}_{(k)} \in \mathbb{R}^{p_k \times p_{-k}}$, where $p_{-k} = \prod_{\ell=1, \ell \neq k}^d p_\ell$.

Each column of $\cm{X}_{(k)}$ corresponds to a \textit{fiber} of $\cm{X}$ along mode $k$, stacked in lexicographic order. The element $(i_1, i_2, \ldots, i_d)$ of $\cm{X}$ is mapped to the $(i_k, j)$-th entry of $\cm{X}_{(k)}$, where the index $j$ is given by:
\begin{equation}
  j = 1 + \sum_{\substack{s=1 \\ s 
\neq k}}^d (i_s - 1) \cdot J_s^{(k)},
  \quad \text{with} \quad
  J_s^{(k)} = \prod_{\substack{\ell=1 \\ \ell < s \\ \ell 
\neq k}}^d p_\ell, \quad \text{and} \quad p_0 = 1.
\end{equation}
This operation is central to defining mode-$k$ products and understanding Tucker decompositions.

\paragraph{Mode-$k$ Product}

For a tensor $\cm{X} \in \mathbb{R}^{p_1 \times \cdots \times p_d}$ and a matrix $\bm{Y} \in \mathbb{R}^{q_k \times p_k}$, the \textbf{mode-$k$ product}, denoted $\cm{X} \times_k \bm{Y}$, results in a new tensor of size $p_1 \times \cdots \times p_{k-1} \times q_k \times p_{k+1} \times \cdots \times p_d$. Its entries are given by:
\begin{equation}
  \left( \cm{X} \times_k \bm{Y} \right)_{i_1 \cdots i_{k-1} j i_{k+1} \cdots i_d}
  = \sum_{i_k=1}^{p_k} \cm{X}_{i_1 \cdots i_k \cdots i_d} \cdot \bm{Y}_{j i_k},
  \quad \text{for all } j = 1, \ldots, q_k.
\end{equation}
This operation recombines the tensor along mode $k$ with the matrix $\bm{Y}$.

\paragraph{Generalized Inner Product}

For two tensors $\cm{X} \in \mathbb{R}^{p_1 \times \cdots \times p_d}$ and $\cm{Y} \in \mathbb{R}^{p_1 \times \cdots \times p_{d_0}}$ where $d \geq d_0$, their \textbf{generalized inner product} is defined as:
\begin{equation}
  \langle \cm{X}, \cm{Y} \rangle
  = \sum_{i_1=1}^{p_1} \cdots \sum_{i_{d_0}=1}^{p_{d_0}} \cm{X}_{i_1 \cdots i_{d_0} i_{d_0+1} \cdots i_d} \cdot \cm{Y}_{i_1 \cdots i_{d_0}},
\end{equation}
and it results in a $(d - d_0)$-th order tensor with entries indexed by $(i_{d_0+1}, \ldots, i_d)$.
In the special case where $d = d_0$, the generalized inner product reduces to the standard \textbf{Frobenius inner product}, and we define the \textbf{Frobenius norm} of $\cm{X}$ as:
\begin{equation}
  \| \cm{X} \|_{\text{F}} = \sqrt{ \langle \cm{X}, \cm{X} \rangle }.
\end{equation}

\paragraph{Outer Product}

The \textbf{outer product} of two tensors $\cm{X} \in \mathbb{R}^{p_1 \times \cdots \times p_{d_1}}$ and $\cm{Y} \in \mathbb{R}^{q_1 \times \cdots \times q_{d_2}}$ is denoted by $\cm{X} \circ \cm{Y}$ and results in a tensor of order $d_1 + d_2$ with entries:
\begin{equation}
  (\cm{X} \circ \cm{Y})_{i_1 \cdots i_{d_1} j_1 \cdots j_{d_2}}
  = \cm{X}_{i_1 \cdots i_{d_1}} \cdot \cm{Y}_{j_1 \cdots j_{d_2}},
  \quad \text{for all indices } i_k, j_\ell.
\end{equation}

\paragraph{Tucker Decomposition and Tucker Ranks}

The \textbf{Tucker rank} of a tensor $\cm{X} \in \mathbb{R}^{p_1 \times \cdots \times p_d}$ is a vector $(r_1, \ldots, r_d)$, where each $r_k$ is the rank of the mode-$k$ matricization $\cm{X}_{(k)}$, i.e.,
\begin{equation}
  r_k = \text{rank}(\cm{X}_{(k)}) \in \mathbb{N}, \quad \text{for } k = 1, \ldots, d.
\end{equation}

While Tucker ranks are defined via matricization ranks, they correspond to the number of components retained in the \textbf{Tucker decomposition} of $\cm{X}$. If $\cm{X}$ has Tucker ranks $(r_1, \ldots, r_d)$, it can be written as:
\begin{equation}
  \cm{X} = \cm{Y} \times_{j=1}^d \bm{Y}_j
  = \cm{Y} \times_1 \bm{Y}_1 \times_2 \bm{Y}_2 \cdots \times_d \bm{Y}_d,
\end{equation}
where $\bm{Y}_j \in \mathbb{R}^{p_j \times r_j}$ is the factor matrix for mode $j$, and $\cm{Y} \in \mathbb{R}^{r_1 \times \cdots \times r_d}$ is the core tensor.

The mode-$k$ matricization of $\cm{X}$ under the Tucker decomposition can be expressed as:
\begin{equation}
  \cm{X}_{(k)} = \cm{Y}_{(k)} \left( \otimes_{j=1,j\neq k}^d \bm{Y}_j \right)^\top,
\end{equation}
where $\otimes$ is the Kronecker product, and the product is taken over all modes except $k$. This structure is central to our algorithm and theoretical analysis, as it allows us to work with low-rank representations in high-dimensional spaces.

\newpage

\section{Local Moment Conditions of Tensors}
\label{append:local_moments}

In this appendix, we provide a detailed discussion of the \textit{local moment conditions} for tensors, which play a central role in the theoretical analysis of our robust tensor estimation framework. These conditions generalize traditional moment assumptions by focusing on the \textit{tail behavior of tensor components in low-dimensional subspaces} aligned with the underlying tensor structure. Such localization enables us to establish statistical guarantees under much weaker conditions than those required by global moment assumptions, particularly in the presence of heavy-tailed noise or covariates.

\subsection{Motivation: Why Local Moments?}

In our robust estimation framework, the gradient of the loss with respect to the full tensor $\cm{A}$ is projected onto the subspaces defined by the factor matrices $\{\bm{U}_j\}_{j=1}^d$ via multilinear projections. Specifically, the partial gradients used for updating each factor depend only on the projections:
\[
\cm{A} \times_{j=1}^d \bm{U}_j^\top,
\]
rather than the full tensor itself. Consequently, the statistical behavior of the gradients—and hence the convergence and risk properties of our algorithm—is governed by the \textit{distribution of these projected components}, not the ambient distribution of the full data.

When the data or noise exhibit heavy tails, global moment conditions (e.g., boundedness, sub-Gaussianity, or even finite fourth moments) may be too restrictive or entirely unavailable. However, if the projections of the data onto the relevant low-dimensional subspaces (induced by the true or approximate factor directions) have \textit{better-behaved tails}, then robust estimation remains feasible. This motivates the introduction of local moment conditions, which restrict attention to the distributions of tensor components within neighborhoods of the true factor subspaces.

\subsection{Definitions and Properties}

Let $\cm{A} \in \mathbb{R}^{p_1 \times \cdots \times p_d}$ be a tensor, and let $\{\bm{U}_j^* \in \mathbb{R}^{p_j \times r_j}\}_{j=1}^d$ denote the true (or target) factor matrices that define the low-rank structure of $\cm{A}^* = \cm{S}^* \times_{j=1}^d \bm{U}_j^*$. For each $j$, define the projection operator onto the column space of $\bm{U}_j^*$ as:
\[
\mathcal{P}_{\bm{U}_j^*} = \bm{U}_j^* (\bm{U}_j^{* \top} \bm{U}_j^*)^\dagger \bm{U}_j^{* \top},
\]
where $(\cdot)^\dagger$ denotes the Moore--Penrose pseudo-inverse. The angular deviation of a unit vector $\bm{v} \in \mathbb{R}^{p_j}$ from the subspace $\text{col}(\bm{U}_j^*)$ is measured by:
\[
\sin \arccos \left( \|\mathcal{P}_{\bm{U}_j^*} \bm{v}\|_2 \right).
\]

For a tolerance parameter $\delta \in [0,1]$, we define the set of admissible directions for mode $j$ as:
\begin{equation}
  \mathcal{V}(\bm{U}_j^*, \delta)
  =
  \left\{
    \bm{v} \in \mathbb{R}^{p_j} \;\bigg|\; \|\bm{v}\|_2 = 1 \;\text{and}\; \sin \arccos \left( \|\mathcal{P}_{\bm{U}_j^*} \bm{v}\|_2 \right) \leq \delta
  \right\}.
\end{equation}
Intuitively, $\mathcal{V}(\bm{U}_j^*, \delta)$ consists of all unit vectors that lie within an angle $\arcsin(\delta)$ of the column space of $\bm{U}_j^*$; smaller $\delta$ corresponds to stricter proximity to the true factor subspace.

With this, we define two types of local moment conditions:

\begin{definition}[Local Moments of Tensors]
  Let $\cm{T} \in \mathbb{R}^{p_1 \times \cdots \times p_d}$ be a tensor, and let $\{\bm{U}_j^*\}_{j=1}^d$ be fixed factor matrices.

  \begin{enumerate}
    \item The $\eta$-th all-mode local moment of $\cm{T}$ is defined as:
      \begin{equation}
        \textup{LM}_0(\cm{T}; \eta, \delta, \{\bm{U}_j^*\}_{j=1}^d)
        =
        \sup_{\bm{v}_j \in \mathcal{V}(\bm{U}_j^*, \delta),\; j=1,\dots,d}
        \mathbb{E}\left[ \left| \cm{T} \times_{j=1}^d \bm{v}_j^\top \right|^\eta \right].
      \end{equation}

    \item For $1 \leq k \leq d$, the $\eta$-th mode-$k$-excluded local moment is defined as:
      \begin{equation}
        \textup{LM}_k(\cm{T}; \eta, \delta, \{\bm{U}_j^*\}_{j=1}^d)
        =
        \sup_{\bm{v}_j \in \mathcal{V}(\bm{U}_j^*, \delta),\; 1 \leq l \leq p_k}
        \mathbb{E}\left[ \left| \cm{T} \times_{j=1,j \neq k}^d \bm{v}_j^\top \times_k \bm{c}_l^\top \right|^\eta \right],
      \end{equation}
      where $\bm{c}_l$ is the $l$-th canonical basis vector in $\mathbb{R}^{p_k}$.
  \end{enumerate}
\end{definition}

These definitions generalize the notion of moments to \textit{directions within low-dimensional subspaces}. The all-mode local moment $\text{LM}_0$ captures the overall tail behavior of the full multilinear projection $\cm{T} \times_{j=1}^d \bm{v}_j^\top$, while the mode-$k$-excluded local moment $\text{LM}_k$ focuses on projections excluding the $k$-th mode, which is particularly useful when estimating the $k$-th factor matrix.

\textbf{Properties:}
\begin{itemize}
  \item When $\delta = 1$, we have $\mathcal{V}(\bm{U}_j^*, 1) = \{\bm{v} : \|\bm{v}\|_2 = 1\}$, so $\text{LM}_0$ and $\text{LM}_k$ reduce to their global counterparts (i.e., moments over the full unit sphere).
  \item Smaller $\delta$ restricts attention to directions \textit{closer to the true factor subspaces}, where the projected data or gradients are often better behaved—even if the ambient distribution is heavy-tailed.
  \item These moments control the tails of projections that directly influence the gradient updates in our robust algorithm, thereby determining the stability and convergence of the estimation procedure.
\end{itemize}

\subsection{Example: Local Moments Under Directional Dependence (Used in Experiment 3)}\label{append:B.3}

Consider a tensor \(\cm{X} \in \mathbb{R}^{p \times p \times p}\) (with \(d = 3\) and \(p_1 = p_2 = p_3 = p\)) whose vectorized form \(\text{vec}(\cm{X}) \in \mathbb{R}^{p^3}\) follows a multivariate distribution with mean zero and a structured covariance matrix. Unlike the i.i.d.\ case, the entries of \(\cm{X}\) are not independent but exhibit dependence that varies with direction. This example is motivated by the setup in Experiment~3 of the main text, where the local moment of the tensor depends on the alignment between the underlying factor structure and the dominant directions of variation in the data.

Let the covariance matrix \(\bm{\Sigma} \in \mathbb{R}^{p^3 \times p^3}\) be such that the variance of projections of \(\text{vec}(\cm{X})\) onto certain directions is modulated by an angle parameter \(\theta \in [0, \pi/2]\). In this setup, the directions that align closely with the column spaces of the true factor matrices (denoted \(\bm{U}_1^*, \bm{U}_2^*, \bm{U}_3^*\)) exhibit different levels of variance depending on \(\theta\).

Now, consider the \textit{local second moment} of \(\cm{X}\) with respect to these subspaces, defined as
\[
\text{LM}_0(\cm{X}; 2, \delta, \{\bm{U}_j^*\}_{j=1}^3)
=
\sup_{\bm{v}_j \in \mathcal{V}(\bm{U}_j^*, \delta), \, j = 1, 2, 3}
\mathbb{E}\left[ \left| \cm{X} \times_{j=1}^3 \bm{v}_j^\top \right|^2 \right].
\]

Suppose that when \(\theta = 0\), the dominant directions of variation in \(\text{vec}(\cm{X})\) are nearly orthogonal to the column spaces of \(\bm{U}_1^*, \bm{U}_2^*, \bm{U}_3^*\). In this case, the projections \(\cm{X} \times_{j=1}^3 \bm{v}_j^\top\) for \(\bm{v}_j \in \mathcal{V}(\bm{U}_j^*, \delta)\) have relatively small variance, and the local second moment is close to a minimal value, say on the order of 1.

In contrast, when \(\theta = \pi/2\), the dominant directions of \(\text{vec}(\cm{X})\) align well with the column spaces of the true factors. The projections \(\cm{X} \times_{j=1}^3 \bm{v}_j^\top\) then capture a significant portion of the total variance, and the local second moment increases to a larger but still manageable value, say on the order of 5.

Importantly, the \textit{global second moment} of \(\cm{X}\)—defined as \(\sup_{\|\bm{v}\|_2 = 1} \mathbb{E}[|\cm{X} \times_{j=1}^3 \bm{v}^\top|^2]\)—remains unchanged across different values of \(\theta\). However, the \textit{local moment} varies considerably, depending on how well the projection directions align with the true factor subspaces.

This example demonstrates that the local moment condition is sensitive to the \textit{geometric alignment} of the data structure, even when global distributional properties are held fixed. It also highlights why the local moment framework is better suited to capturing the effective behavior of gradients in tensor estimation problems, particularly when the signal resides in a low-dimensional subspace defined by the underlying factors.

\subsection{Usefulness in Robust Tensor Estimation}

The local moment conditions are \textbf{crucial} for the theoretical analysis of our robust gradient descent framework. Specifically:

\begin{itemize}
  \item They allow us to control the tails of the projected gradients $\cm{T} \times_{j=1}^d \bm{v}_j^\top$, which are the quantities actually used in the robust gradient estimators.
  \item By focusing on low-dimensional subspaces where the signal resides (as defined by $\{\bm{U}_j^*\}$), we can obtain sharp, adaptive moment bounds without relying on heavy global assumptions.
  \item This localization leads to weaker sufficient conditions for statistical consistency and convergence, enabling robust estimation under heavy-tailed noise or covariates—even when traditional moment conditions fail.
\end{itemize}

In summary, the concept of \textbf{local moments} provides a principled way to extend classical moment-based analysis to the tensor setting, accounting for the intrinsic geometry of low-rank tensor models and ensuring robustness in real-world, heavy-tailed environments.

\newpage

\section{Convergence Analysis of Robust Gradient Descent}\label{append:convergence}

\subsection{Proofs of Theorem \ref{thm:1}}

The proof consists of five steps. In the first step, we introduce the notations and the regularity conditions in the following steps. In the second to fourth steps, we establish the convergence analysis of the estimation errors. Finally, in the last step, we verify the conditions given in the first steps recursively.\\

\noindent\textit{Step 1.} (Notations and conditions)

\noindent We first introduce the notations used in the proof.
At step $t$, we simplify the notations of the robust gradient estimators to 
\begin{equation}
    \cm{G}_0^{(t)} = \cm{G}(\bm{F}^{(t)}),~~\text{and}~~\bm{G}_k^{(t)} = \bm{G}(\bm{F}^{(t)}),
\end{equation}
for $k=1,\dots,d$ and $t=1,\dots,T$.
Denote $\bm{V}_k^{(t)}=(\otimes_{j\neq k}\bm{U}_j^{(t)})\cm{S}^{(t)\top}_{(k)}$,
\begin{equation}
    \bm{\Delta}_k^{(t)}=\bm{G}_k^{(t)}-\mathbb{E}[\nabla_k\mathcal{L}^{(t)}]=\bm{G}_k^{(t)}-\mathbb{E}[\nabla\mathcal{L}(\cm{A}^{(t)})_{(k)}\bm{V}_k^{(t)}],
\end{equation} 
and 
\begin{equation}
    \bm{\Delta}_0^{(t)}=\cm{G}_0^{(t)}-\mathbb{E}[\nabla_0\mathcal{L}^{(t)}]=\cm{G}_0^{(t)}-\mathbb{E}[\nabla\mathcal{L}(\cm{A}^{(t)})\times_{j=1}^d\bm{U}_j^{(t)\top}]
\end{equation}
as the robust gradient estimation errors. By the stability of the robust gradients, $\|\bm{\Delta}_k^{(t)}\|_\text{F}^2\leq\phi\|\cm{A}^{(t)}-\cm{A}^*\|^2_\text{F}+\xi_k^2$, for all $k=0,1,\dots,d$ and $t=1,2,\dots,T$. In addition, we assume $b\asymp\bar{\sigma}^{1/(d+1)}$, as required in Theorem \ref{thm:1}.

Let $\cm{A}^*=\cm{S}^*\times_{k=1}^d\bm{U}_k^*$ such that $\bm{U}_k^{*\top}\bm{U}_k^*=b^2\bm{I}_{r_k}$, for $k=1,\dots,d$. Define $\mathbb{O}_r=\{\bm{M}\in\mathbb{R}^{r\times r}:\bm{M}^\top\bm{M}=\bm{I}_r\}$ as the set of $r\times r$ orthogonal matrices. For each step $t=0,1,\dots,T$, we define
\begin{equation}
    \text{Err}^{(t)} = \min_{\bm{O}_k\in\mathbb{O}_{r_k},1\leq k\leq d}\left\{\sum_{k=1}^d\|\bm{U}^{(t)}_k-\bm{U}_k^*\bm{O}_k\|_\textup{F}^2+\|\cm{S}^{(t)}-\cm{S}^*\times_{j=1}^d\bm{O}_j^\top\|^2_\textup{F}\right\},
\end{equation}
and
\begin{equation}
    (\bm{O}_1^{(t)},\cdots,\bm{O}_d^{(t)}) = \argmin_{{\bm{O}_k\in\mathbb{O}_{r_k},1\leq k\leq d}}\left\{\sum_{k=1}^d\|\bm{U}^{(t)}_k-\bm{U}_k^*\bm{O}_k\|_\textup{F}^2+\|\cm{S}^{(t)}-\cm{S}^*\times_{j=1}^d\bm{O}_j^\top\|^2_\textup{F}\right\}.
\end{equation}
Here, $\text{Err}^{(t)}$ collects the combined estimation errors for all tensor decomposition components at step $t$, and $\bm{O}_k^{(t)}$'s are the optimal rotations used to handle the non-identifiability of the Tucker decomposition. 

Next, we discuss some additional conditions used in the convergence analysis. To ease presentation, we first assume that these conditions hold and verify them in the last step.

\noindent(C1) For any $t=0,1,\dots,T$ and $k=1,2,\dots,d$, $\|\cm{S}_{(k)}^{(t)}\|\leq C\bar{\sigma}b^{-d}$ and $\|\bm{U}_k^{(t)}\|\leq Cb$ for some absolute constant greater than one. Hence, $\|\bm{V}_k^{(t)}\|\leq\|\cm{S}_{(k)}^{(t)}\|\cdot\prod_{j\neq k}\|\bm{U}_j^{(t)}\|\leq C_d\bar{\sigma}b^{-1}$.

\noindent(C2) For any $t=0,1,\dots,T$, $\text{Err}^{(t)}\leq C\alpha\beta^{-1}b^2\kappa^{-2} $.\\

\noindent\textit{Step 2.} (Descent of $\text{Err}^{(t)}$)

\noindent By definition of $\text{Err}^{(t)}$ and $\bm{O}_k^{(t)}$'s,
\begin{equation}
    \begin{split}
        \text{Err}^{(t+1)} & = \sum_{k=1}^d\left\|\bm{U}_k^{(t+1)}-\bm{U}^*_k\bm{O}_k^{(t+1)}\right\|^2_\textup{F}+\left\|\cm{S}^{(t+1)}-\cm{S}^*\times_{j=1}^d\bm{O}_j^{(t+1)\top}\right\|^2_\textup{F}\\
        & \leq \sum_{k=1}^d\left\|\bm{U}_k^{(t+1)}-\bm{U}^*_k\bm{O}_k^{(t)}\right\|^2_\textup{F}+\left\|\cm{S}^{(t+1)}-\cm{S}^*\times_{j=1}^d\bm{O}_j^{(t)\top}\right\|^2_\textup{F}.
    \end{split}
\end{equation}

For each $k=1,\cdots,d$, since $\bm{U}_k^{(t+1)}=\bm{U}_k^{(t)}-\eta\bm{G}_k^{(t)}-a\eta\bm{U}_k^{(t)}(\bm{U}_k^{(t)\top}\bm{U}_k^{(t)}-b^2\bm{I}_{r_k})$, we have that for any $\zeta>0$, 
\begin{equation}
    \label{eq:U_grad}
    \begin{split}
        & \|\bm{U}^{(t+1)}_k-\bm{U}_k^*\bm{O}_k^{(t)}\|_\textup{F}^2\\
        = & \|\bm{U}_k^{(t)}-\bm{U}_k^*\bm{O}_k^{(t)}-\eta(\bm{G}_k^{(t)}+a\bm{U}_k^{(t)}(\bm{U}_k^{(t)\top}\bm{U}_k^{(t)}-b^2\bm{I}_{r_k}))\|_\textup{F}^2\\
        = & \|\bm{U}_k^{(t)}-\bm{U}_k^*\bm{O}_k^{(t)}-\eta(\mathbb{E}[\nabla_k\mathcal{L}^{(t)}]+a\bm{U}_k^{(t)}(\bm{U}_k^{(t)\top}\bm{U}_k^{(t)}-b^2\bm{I}_{r_k}))-\eta \bm{\Delta}_k^{(t)}\|_\textup{F}^2\\
        \leq & \|\bm{U}_k^{(t)}-\bm{U}_k^*\bm{O}_k^{(t)}-\eta(\mathbb{E}[\nabla_k\mathcal{L}^{(t)}]+a\bm{U}_k^{(t)}(\bm{U}_k^{(t)\top}\bm{U}_k^{(t)}-b^2\bm{I}_{r_k}))\|_\text{F}^2+\eta^2\|\bm{\Delta}_k^{(t)}\|_\textup{F}^2\\
        & + 2\eta\|\bm{\Delta}_k^{(t)}\|_\textup{F}\cdot\|\bm{U}_k^{(t)}-\bm{U}_k^*\bm{O}_k^{(t)}-\eta(\mathbb{E}[\nabla_k\mathcal{L}^{(t)}]+a\bm{U}_k^{(t)}(\bm{U}_k^{(t)\top}\bm{U}_k^{(t)}-b^2\bm{I}_{r_k}))\|_\text{F}\\
        \leq & (1+\zeta)\|\bm{U}_k^{(t)}-\bm{U}_k^*\bm{O}_k^{(t)}-\eta(\mathbb{E}[\nabla_k\mathcal{L}^{(t)}]+a\bm{U}_k^{(t)}(\bm{U}_k^{(t)\top}\bm{U}_k^{(t)}-b^2\bm{I}_{r_k}))\|_\text{F}^2\\
        & + (1+\zeta^{-1})\eta^2\|\bm{\Delta}_k^{(t)}\|_\text{F}^2,
    \end{split}
\end{equation}
where the last inequality stems from the mean inequality.

For the first term on the right hand side in \eqref{eq:U_grad}, we have the following decomposition
\begin{equation}\label{eq:decomposition_U}
    \begin{split}
        & \|\bm{U}_k^{(t)}-\bm{U}_k^*\bm{O}_k^{(t)}-\eta(\mathbb{E}[\nabla_k\mathcal{L}^{(t)}]+a\bm{U}_k^{(t)}(\bm{U}_k^{(t)\top}\bm{U}_k^{(t)}-b^2\bm{I}_{r_k}))\|_\text{F}^2 \\
        = & \|\bm{U}_k^{(t)}-\bm{U}_k^*\bm{O}_k^{(t)}\|_\text{F}^2 + \eta^2 \|\mathbb{E}[\nabla_k\mathcal{L}^{(t)}]+a\bm{U}_k^{(t)}(\bm{U}_k^{(t)\top}\bm{U}_k^{(t)}-b^2\bm{I}_{r_k})\|_\text{F}^2\\
        & - 2\eta\langle\bm{U}_k^{(t)}-\bm{U}_k^*\bm{O}_k^{(t)},\mathbb{E}[\nabla_k\mathcal{L}^{(t)}]\rangle
        - 2\eta a\langle\bm{U}_k^{(t)}-\bm{U}_k^*\bm{O}_k^{(t)},\bm{U}_k^{(t)}(\bm{U}_k^{(t)\top}\bm{U}_k^{(t)}-b^2\bm{I}_{r_k})\rangle.
    \end{split}
\end{equation}
Here, by condition (C1), the second term in \eqref{eq:decomposition_U} can be bounded by
\begin{equation}
    \begin{split}
        &\|\mathbb{E}[\nabla_k\mathcal{L}^{(t)}]+a\bm{U}_k^{(t)}(\bm{U}_k^{(t)\top}\bm{U}_k^{(t)}-b^2\bm{I}_{r_k})\|_\text{F}^2\\
        \leq & 2\|\mathbb{E}[\nabla\overline{\mathcal{L}}(\cm{A}^{(t)})]_{(k)}\bm{V}_k^{(t)}\|_\text{F}^2 + 2a^2\|\bm{U}_k^{(t)}(\bm{U}_k^{(t)\top}\bm{U}_k^{(t)}-b^2\bm{I}_{r_k})\|_\text{F}^2\\
        \leq & 2\|\bm{V}_k^{(t)}\|^2\|\mathbb{E}[\nabla\overline{\mathcal{L}}(\cm{A}^{(t)})]\|_\text{F}^2 + 2a^2\|\bm{U}_k^{(t)}\|^2\|\bm{U}_k^{(t)\top}\bm{U}_k^{(t)}-b^2\bm{I}_{r_k}\|_\text{F}^2\\
        \leq & C_db^{-2}\bar{\sigma}^2\|\mathbb{E}[\nabla\overline{\mathcal{L}}(\cm{A}^{(t)})]\|_\text{F}^2 + Ca^2b^2\|\bm{U}_k^{(t)\top}\bm{U}_k^{(t)}-b^2\bm{I}_{r_k}\|_\text{F}^2.
    \end{split}
\end{equation}
The third term in \eqref{eq:decomposition_U} can be rewritten as
\begin{equation}
    \begin{split}
        & \langle\bm{U}_k^{(t)}-\bm{U}_k^*\bm{O}_k^{(t)},\mathbb{E}[\nabla_k\mathcal{L}^{(t)}]\rangle\\
        = & \langle\cm{A}^{(t)}-\cm{S}^{(t)}\times_{j\neq k}\bm{U}_j^{(t)}\times_k\bm{U}_k^*\bm{O}_k^{(t)},\mathbb{E}[\nabla\overline{\mathcal{L}}(\cm{A}^{(t)})]-\mathbb{E}[\nabla\overline{\mathcal{L}}(\cm{A}^*)]\rangle\\
        = & \langle\cm{A}^{(t)}-\cm{A}_k^{(t)},\mathbb{E}[\overline{\mathcal{L}}(\cm{A}^{(t)})]-\mathbb{E}[\nabla\overline{\mathcal{L}}(\cm{A}^*)]\rangle,
    \end{split}
\end{equation}
where $\cm{A}_k^{(t)}:=\cm{S}^{(t)}\times_{j\neq k}\bm{U}_j^{(t)}\times_k\bm{U}_k^*\bm{O}_k^{(t)}$.
For the fourth term in \eqref{eq:decomposition_U}, we have
\begin{equation}
    \begin{split}
        & \left\langle\bm{U}_k^{(t)}-\bm{U}_k^*\bm{O}_k^{(t)},\bm{U}_k^{(t)}(\bm{U}_k^{(t)\top}\bm{U}_k^{(t)}-b^2\bm{I}_{r_k})\right\rangle\\
        = & \left\langle\bm{U}_k^{(t)\top}\bm{U}_k^{(t)}-\bm{U}_k^{(t)\top}\bm{U}_k^*\bm{O}_k^{(t)},\bm{U}_k^{(t)\top}\bm{U}_k^{(t)}-b^2\bm{I}_{r_k}\right\rangle\\
        = & \frac{1}{2}\left\langle\bm{U}_k^{(t)\top}\bm{U}_k^{(t)}-\bm{U}_k^{*\top}\bm{U}_k^*,\bm{U}_k^{(t)\top}\bm{U}_k^{(t)}-b^2\bm{I}_{r_k}\right\rangle\\
        & + \frac{1}{2}\left\langle\bm{U}_k^{*\top}\bm{U}_k^{*}-2\bm{U}_k^{(t)\top}\bm{U}_k^*\bm{O}_k^{(t)}+\bm{U}_k^{(t)\top}\bm{U}_k^{(t)},\bm{U}_k^{(t)\top}\bm{U}_k^{(t)}-b^2\bm{I}_{r_k}\right\rangle\\
        = & \frac{1}{2}\|\bm{U}_k^{(t)\top}\bm{U}_k^{(t)}-b^2\bm{I}_{r_k}\|_\text{F}^2\\
        & + \frac{1}{2}\left\langle(\bm{U}_k^*\bm{O}_k^{(t)}-\bm{U}_k^{(t)})^\top(\bm{U}_k^*\bm{O}_k^{(t)}-\bm{U}_k^{(t)}),\bm{U}_k^{(t)\top}\bm{U}_k^{(t)}-b^2\bm{I}_{r_k}\right\rangle\\
        \geq & \frac{1}{2}\|\bm{U}_k^{(t)\top}\bm{U}_k^{(t)}-b^2\bm{I}_{r_k}\|_\text{F}^2 - \frac{1}{2}\|\bm{U}_k^*\bm{O}_k^{(t)}-\bm{U}_k^{(t)}\|_\text{F}^2\cdot\|\bm{U}_k^{(t)\top}\bm{U}_k^{(t)}-b^2\bm{I}_{r_k}\|_\text{F} \\
        \geq & \frac{1}{2}\|\bm{U}_k^{(t)\top}\bm{U}_k^{(t)}-b^2\bm{I}_{r_k}\|_\text{F}^2 - \frac{1}{4}\|\bm{U}_k^*\bm{O}_k^{(t)}-\bm{U}_k^{(t)}\|_\text{F}^4-\frac{1}{4}\|\bm{U}_k^{(t)\top}\bm{U}_k^{(t)}-b^2\bm{I}_{r_k}\|_\text{F}^2 \\
        \geq & \frac{1}{4}\|\bm{U}_k^{(t)\top}\bm{U}_k^{(t)}-b^2\bm{I}_{r_k}\|_\text{F}^2-\frac{\text{Err}^{(t)}}{4}\|\bm{U}_k^{(t)}-\bm{U}_k^*\bm{O}_k^{(t)}\|_\text{F}^2,
    \end{split}
\end{equation}
where we use the fact that $\|\bm{U}_k^*\bm{O}_k^{(t)}-\bm{U}_k^{(t)}\|_\text{F}^2\leq \text{Err}^{(t)}$.

Hence, for any $k=1,2,\dots,d$,
\begin{equation}
    \begin{split}
        &\|\bm{U}_k^{(t)}-\bm{U}_k^*\bm{O}_k^{(t)}-\eta(\mathbb{E}[\nabla_k\mathcal{L}^{(t)}]+a\bm{U}_k^{(t)}(\bm{U}_k^{(t)\top}\bm{U}_k^{(t)}-b^2\bm{I}_{r_k}))\|_\text{F}^2 \\
        \leq & \|\bm{U}_k^{(t)}-\bm{U}_k^*\bm{O}_k^{(t)}\|_\text{F}^2 - 2Q_{k,1}^{(t)}\eta + Q_{k,2}^{(t)}\eta^2,
    \end{split}
\end{equation}
where
\begin{equation}
    Q_{k,1}^{(t)}=\langle\cm{A}^{(t)}-\cm{A}_k^{(t)},\mathbb{E}[\nabla\overline{\mathcal{L}}(\cm{A}^{(t)})]\rangle+\frac{a}{4}\left\|\bm{U}_k^{(t)\top}\bm{U}_k^{(t)}-b^2\bm{I}_{r_k}\right\|_\text{F}^2-\frac{a\text{Err}^{(t)}}{4}\left\|\bm{U}_k^{(t)}-\bm{U}_k^*\bm{O}_k^{(t)}\right\|_\text{F}^2
\end{equation} and
\begin{equation}
    Q_{k,2}^{(t)}=C_db^{-2}\bar{\sigma}^2\|\mathbb{E}[\nabla\overline{\mathcal{L}}(\cm{A}^{(t)})]\|_\text{F}^2 + Ca^2b^2\|\bm{U}_k^{(t)\top}\bm{U}_k^{(t)}-b^2\bm{I}_{r_k}\|_\text{F}^2.
\end{equation}

Similarly, for any $\zeta>0$,
\begin{equation}
    \begin{split}
        & \|\widetilde{\cm{S}}^{(t+1)}-\cm{S}^*\times_{k=1}^d\bm{O}_k^{(t)\top}\|_\text{F}^2
        = \|\cm{S}^{(t)}-\eta\cm{G}_0^{(t)}-\cm{S}^*\times_{k=1}^d\bm{O}_k^{(t)\top}\|_\text{F}^2\\
        = & \|\cm{S}^{(t)}-\cm{S}^*\times_{k=1}^d\bm{O}_k^{(t)\top}-\eta\mathbb{E}[\nabla_0\mathcal{L}^{(t)}]-\eta\bm{\Delta}_0^{(t)}\|_\text{F}^2\\
        \leq & (1+\zeta)\|\cm{S}^{(t)}-\cm{S}^*\times_{k=1}^d\bm{O}_k^{(t)\top}-\eta\mathbb{E}[\nabla_0\mathcal{L}^{(t)}]\|_\text{F}^2 + \eta^2(1+\zeta^{-1})\|\bm{\Delta}_0^{(t)}\|_\text{F}^2,
    \end{split}
\end{equation}
and
\begin{equation}
    \|\cm{S}^{(t)}-\cm{S}^*\times_{k=1}^d\bm{O}_k^{(t)\top}-\eta\mathbb{E}[\nabla_0\mathcal{L}^{(t)}]\|_\text{F}^2 \leq \|\cm{S}^{(t)}-\cm{S}^*\times_{k=1}^d\bm{O}_k^{(t)\top}\|_\text{F}^2 - 2Q_{0,1}^{(t)}\eta + Q_{0,2}^{(t)}\eta^2,
\end{equation}
where 
$$Q_{0,1}^{(t)}=\langle\cm{A}^{(t)}-\cm{A}_0^{(t)},\mathbb{E}[\nabla\overline{\mathcal{L}}(\cm{A}^{(t)})]\rangle \text{ with }\cm{A}_0^{(t)}=\cm{S}^*\times_{k=1}^d\bm{U}_k^{(t)}\bm{O}_k^{(t)\top}$$
and 
$Q_{0,2}^{(t)}=C_db^{2d}\|\mathbb{E}[\nabla\overline{\mathcal{L}}(\cm{A}^{(t)})]\|_\text{F}^2$.

Hence, combining the above results, we have
\begin{equation}
    \text{Err}^{(t+1)} \leq (1+\zeta)\left\{\text{Err}^{(t)} - 2\eta\sum_{k=0}^dQ^{(t)}_{k,1}+\eta^2\sum_{k=0}^dQ_{k,2}^{(t)}\right\} + (1+\zeta^{-1})\eta^2\sum_{k=0}^d\|\bm{\Delta}_k^{(t)}\|_\text{F}^2.
\end{equation}~

\noindent\textit{Step 3}. (Lower bound of $\sum_{k=0}^dQ^{(t)}_{k,1}$)

\noindent By definition of $Q^{(t)}_{k,1}$ for $k=0,\dots,d$, we have
\begin{equation}
    \begin{split}
        \sum_{k=0}^dQ_{k,1}^{(t)}
        = & \left\langle(d+1)\cm{A}^{(t)}-\sum_{k=0}^d\cm{A}_k^{(t)},\mathbb{E}[\nabla\overline{\mathcal{L}}(\cm{A}^{(t)})]\right\rangle\\
        & + a\sum_{k=1}^d\left\{\frac{1}{4}\|\bm{U}_k^{(t)\top}\bm{U}_k^{(t)}-b^2\bm{I}_{r_k}\|_\text{F}^2 - \frac{\text{Err}^{(t)}}{4}\|\bm{U}_k^{(t)}-\bm{U}_k^*\bm{O}_k^{(t)}\|_\text{F}^2\right\}.
    \end{split}
\end{equation}
For the first term, by RCG condition of $\overline{\mathcal{L}}$ and Cauchy's inequality,
\begin{equation}
    \begin{split}
        & \left\langle(d+1)\cm{A}^{(t)}-\sum_{k=0}^d\cm{A}_k^{(t)},\mathbb{E}[\nabla\overline{\mathcal{L}}(\cm{A}^{(t)})]\right\rangle=\langle\cm{A}^{(t)}-\cm{A}^*+\cm{H},\mathbb{E}[\nabla\overline{\mathcal{L}}(\cm{A}^{(t)})]\rangle\\
        = & \langle\cm{A}^{(t)}-\cm{A}^*,\mathbb{E}[\nabla\overline{\mathcal{L}}(\cm{A}^{(t)})]-\mathbb{E}[\nabla\overline{\mathcal{L}}(\cm{A}^*)]\rangle + \langle\cm{H},\mathbb{E}[\nabla\overline{\mathcal{L}}(\cm{A}^{(t)})]\rangle\\
        \geq & \frac{\alpha}{2}\|\cm{A}^{(t)}-\cm{A}^*\|_\text{F}^2 + \frac{1}{2\beta}\|\mathbb{E}[\nabla\overline{\mathcal{L}}(\cm{A}^{(t)})]\|_\text{F}^2 - \|\cm{H}\|_\text{F}\cdot\|\mathbb{E}[\nabla\overline{\mathcal{L}}(\cm{A}^{(t)})]\|_\text{F}\\
        \geq & \frac{\alpha}{2}\|\cm{A}^{(t)}-\cm{A}^*\|_\text{F}^2 + \frac{1}{2\beta}\|\mathbb{E}[\nabla\overline{\mathcal{L}}(\cm{A}^{(t)})]\|_\text{F}^2 - \frac{1}{4\beta}\|\mathbb{E}[\nabla\overline{\mathcal{L}}(\cm{A}^{(t)})]\|_\text{F}^2 - \beta\|\cm{H}\|_\text{F}^2\\
        = & \frac{\alpha}{2}\|\cm{A}^{(t)}-\cm{A}^*\|_\text{F}^2 + \frac{1}{4\beta}\|\mathbb{E}[\nabla\overline{\mathcal{L}}(\cm{A}^{(t)})]\|_\text{F}^2 - \beta\|\cm{H}\|_\text{F}^2
    \end{split}
\end{equation}
where $\cm{H}$ is the higher-order perturbation term in
\begin{equation}
    \cm{A}^* = \cm{A}_0^{(t)} + \sum_{k=1}^d(\cm{A}_k^{(t)}-\cm{A}^{(t)}) + \cm{H}. 
\end{equation}

By Lemma \ref{lemma:perturb}, we have
$\|\cm{H}\|_\text{F}\leq C_db^{-2}\bar{\sigma}\text{Err}^{(t)}$. Hence, by Lemma \ref{lemma:1}, $\sum_{k=0}^dQ_{k,1}^{(t)}$ can be lower bounded by
\begin{equation}
    \begin{split}
        & \sum_{k=0}^dQ_{k,1}^{(t)} \geq \frac{\alpha}{2}\|\cm{A}^{(t)}-\cm{A}^*\|_\text{F}^2 + \frac{1}{4\beta}\|\mathbb{E}[\nabla\overline{\mathcal{L}}(\cm{A}^{(t)})]\|_\text{F}^2 - C_d\beta b^{-4}\bar{\sigma}^2(\text{Err}^{(t)})^2\\
        & + \frac{a}{4}\sum_{k=1}^d\|\bm{U}_k^{(t)\top}\bm{U}_k^{(t)}-b^2\bm{I}_{r_k}\|_\text{F}^2 - \frac{a}{4}(\text{Err}^{(t)})^2\\
        \geq & \left\{C\alpha b^{2d}\kappa^{-2}- C_d\beta b^{-4}\bar{\sigma}^2\text{Err}^{(t)}-\frac{a\text{Err}^{(t)}}{4}\right\}\text{Err}^{(t)}\\
        & + \frac{1}{4\beta}\|\mathbb{E}[\nabla\overline{\mathcal{L}}(\cm{A}^{(t)})]\|_\text{F}^2 + \left(\frac{a}{4}-C_d\alpha b^{2d-2}\kappa^{-2}\right)\sum_{k=1}^d\|\bm{U}_k^{(t)\top}\bm{U}_k^{(t)}-b^2\bm{I}_{r_k}\|_\text{F}^2\\
        \geq & ~C\alpha b^{2d}\kappa^{-2}\text{Err}^{(t)} + \frac{1}{4\beta}\|\mathbb{E}[\nabla\overline{\mathcal{L}}(\cm{A}^{(t)})]\|_\text{F}^2+\left(\frac{a}{4}-C_d\alpha b^{2d-2}\kappa^{-2}\right)\sum_{k=1}^d\|\bm{U}_k^{(t)\top}\bm{U}_k^{(t)}-b^2\bm{I}_{r_k}\|_\text{F}^2.
    \end{split}
\end{equation}

\noindent\textit{Step 4.} (Convergence analysis)

We have the following bound for $\sum_{k=0}^dQ_{k,2}^{(t)}$
\begin{equation}
    \begin{split}
        \sum_{k=0}^dQ_{k,2}^{(t)} & \leq C_db^{2d}\|\mathbb{E}[\nabla\overline{\mathcal{L}}(\cm{A}^{(t)})]\|_\text{F}^2 + 3a^2b^2\sum_{k=1}^d\|\bm{U}_k^{(t)\top}\bm{U}_k^{(t)}-b^2\bm{I}_{r_k}\|_\text{F}^2.
    \end{split}
\end{equation}

Combining the results above, we have
\begin{equation}
    \begin{split}
        & \text{Err}^{(t)} - 2\eta\sum_{k=0}^dQ_{k,1}^{(t)} + \eta^2\sum_{k=0}^dQ_{k,2}^{(t)}\\
        \leq & \left(1-C\alpha b^{2d}\kappa^{-2}\eta\right)\text{Err}^{(t)} + \left(C_db^{2d}\eta^2-\frac{\eta}{4\beta}\right)\|\mathbb{E}[\nabla\overline{\mathcal{L}}(\cm{A}^{(t)})]\|_\text{F}^2 \\
        & + \left(3a^2b^2\eta^2+C_d\alpha b^{2d-2}\kappa^{-2}\eta-\frac{a\eta}{4}\right)\sum_{k=1}^d\|\bm{U}_k^{(t)\top}\bm{U}_k^{(t)}-b^2\bm{I}_{r_k}\|_\text{F}^2.
    \end{split}
\end{equation}
Taking $\eta=\eta_0b^{-2d}\beta^{-1}$ and $a=C_0b^{2d-2}\alpha\kappa^{-2}$ for some sufficiently small constants $\eta_0$ and $C_0$, we have
\begin{equation}
    \text{Err}^{(t)}-2\eta\sum_{k=0}^d Q_{k,1}^{(t)} + \eta^2\sum_{k=0}^dQ_{k,2}^{(t)} \leq (1-C\alpha\beta^{-1}\kappa^{-2})\text{Err}^{(t)}
\end{equation}
and
\begin{equation}
    \text{Err}^{(t+1)} \leq (1+\zeta)(1-\eta_0\alpha\beta^{-1}\kappa^{-2})\text{Err}^{(t)} + (1+\zeta^{-1})\eta^2\sum_{k=0}^d\|\bm{\Delta}_k^{(t)}\|_\text{F}^2.
\end{equation}
Taking $\zeta=\eta_0\alpha\beta^{-1}\kappa^{-2}/2$, we have
\begin{equation}
    \text{Err}^{(t+1)} \leq (1-\eta_0\alpha\beta^{-1}\kappa^{-2}/2)\text{Err}^{(t)}+C\alpha^{-1}\beta^{-1}\bar{\sigma}^{-4d/(d+1)}\kappa^2\sum_{k=0}^d\|\bm{\Delta}_k^{(t)}\|_\text{F}^2.
\end{equation}

By stability of the robust gradient estimators, for $k=0,1,\dots,d$ and $t=1,2,\dots,T$,
\begin{equation}
    \|\bm{\Delta}^{(t)}_k\|_\text{F}^2\leq\phi\|\cm{A}^{(t)}-\cm{A}^*\|_\text{F}^2 + \xi_k^2.
\end{equation}
Hence, as $\phi\lesssim \alpha^2\kappa^{-4}\bar{\sigma}^{2d/(d+1)}$, we have
\begin{equation}\label{eq:recursive}
    \begin{split}
        \text{Err}^{(t+1)} & \leq (1-\eta_0\alpha\beta^{-1}\kappa^{-2}/2)\text{Err}^{(t)}+C_d\alpha^{-1}\beta^{-1}\bar{\sigma}^{-4d/(d+1)}\kappa^2\left(\phi\|\cm{A}^{(t)}-\cm{A}^*\|_\text{F}^2+\sum_{k=0}^d\xi_k^2\right)\\
        & \leq (1-\eta_0\alpha\beta^{-1}\kappa^{-2}/2+C_d\alpha^{-1}\beta^{-1}\bar{\sigma}^{-2d/(d+1)}\kappa^2\phi)\text{Err}^{(t)}+C\alpha^{-1}\beta^{-1}\bar{\sigma}^{-4d/(d+1)}\kappa^2\sum_{k=0}^d\xi_k^2\\
        & \leq (1-C\alpha\beta^{-1}\kappa^{-2})\text{Err}^{(t)}+C\alpha^{-1}\beta^{-1}\bar{\sigma}^{-4d/(d+1)}\kappa^2\sum_{k=0}^d\xi_k^2\\
        & \leq (1-C\alpha\beta^{-1}\kappa^{-2})^{t+1}\text{Err}^{(0)}+C\alpha^{-2}\bar{\sigma}^{-4d/(d+1)}\kappa^4\sum_{k=0}^d\xi_k^2.
    \end{split}
\end{equation}

We apply Lemma \ref{lemma:1} again and obtain
\begin{equation}
    \begin{split}
        \|\cm{A}^{(t)}-\cm{A}^*\|_\text{F}^2 & \leq C\bar{\sigma}^{2d/(d+1)}\text{Err}^{(t+1)}\\
        & \leq C\bar{\sigma}^{2d/(d+1)}(1-C\alpha\beta^{-1}\kappa^{-2})^{t}\text{Err}^{(0)} + C\bar{\sigma}^{-2d/(d+1)}\alpha^{-2}\kappa^4\sum_{k=0}^d\xi_k^2\\
        & \leq C\kappa^2(1-C\alpha\beta^{-1}\kappa^{-2})^{t}\|\cm{A}^{(0)}-\cm{A}^*\|_\text{F}^2 + C\bar{\sigma}^{-2d/(d+1)}\alpha^{-2}\kappa^4\sum_{k=0}^d\xi_k^2.
    \end{split}
\end{equation}~

\noindent\textit{Step 5.} (Verfications of conditions)

Finally, we show the conditions (C1) and (C2) hold for all $t=1,2,\dots$. By Lemma \ref{lemma:1}, we have
\begin{equation}
    \text{Err}^{(0)} \leq C(\alpha/\beta)b^2\kappa^{-2}\leq Cb^2.
\end{equation}
By the recursive relationship in \eqref{eq:recursive}, by induction we can check that $\text{Err}^{(t)}\leq Cb^2$ for all $t=1,2,\dots,T$. Furthermore, it implies that
\begin{equation}
    \|\bm{U}_k^{(t)}\|\leq \|\bm{U}_k^*\|+\|\bm{U}_k^{(t)}-\bm{U}_k^*\bm{O}_k^{(t)}\|\leq Cb,~~k=1,2,\dots,d,
\end{equation}
and
\begin{equation}
    \max_{k}\|\cm{S}^{(t)}_{(k)}\|\leq\max_{k}\|\cm{S}^{*}_{(k)}\| + \max_{k}\|\cm{S}_{(k)}^{(t)}-\cm{S}^*\times_{j=1}^d\bm{O}_j^{(t)\top}\|\leq C\underline{\sigma}b^{-d},
\end{equation}
which completes the convergence analysis.

\subsection{Auxiliary Lemmas}

The first lemma shows the equivalence between $\|\cm{A}-\cm{A}^*\|_\text{F}^2$ and the combined error $E$, which is from Lemma E.2 in \citet{han2022optimal} and is presented here for self-containedness. The proof of Lemma \ref{lemma:1} can be found in \citet{han2022optimal} and hence is omitted.

\begin{lemma}\label{lemma:1}
    Suppose $\cm{A}^*=[\![\cm{S}^*;\bm{U}_1^*,\dots,\bm{U}_d^*]\!]$, $\bm{U}_k^{*\top}\bm{U}_k=b^2\bm{I}_{r_k}$, for $k=1,\dots,d$, $\bar{\sigma}=\max_{k}\|\cm{A}^*_{(k)}\|_\textup{sp}$, and $\underline{\sigma}=\min_{k}\sigma_{r_k}(\cm{A}^*_{(k)})$. Let $\cm{A}=[\![\cm{S};\bm{U}_1,\dots,\bm{U}_d]\!]$ be another Tucker low-rank tensor with $\bm{U}_k\in\mathbb{R}^{p_k\times r_k}$, $\|\bm{U}_k\|\leq(1+c_0)b$, and $\max_{k}\|\cm{S}_{(k)}\|\leq(1+c_0)\bar{\sigma}b^{-d}$ for some $c_0>0$. Define
    \begin{equation}
        E:=\min_{\mathbb{O}_k\in\mathbb{O}_{p_k,r_k}}\left\{\sum_{k=1}^d\|\bm{U}_k-\bm{U}_k^*\bm{O}_k\|_\textup{F}^2 + \left\|\cm{S}-[\![\cm{S}^*;\bm{O}_1^\top,\dots,\bm{O}_d^\top]\!]\right\|_\text{F}^2\right\}.
    \end{equation}
    Then, we have
    \begin{equation}
        \begin{split}
            & E \leq b^{-2d}(C+C_1b^{2d+2}\underline{\sigma}^{-2})\|\cm{A}-\cm{A}^*\|_\textup{F}^2 + 2b^{-2}C_1\sum_{k=1}^d\|\bm{U}_k^\top\bm{U}_k-b^2\bm{I}_{r_k}\|_\textup{F}^2,\\
            \text{and  }& \|\cm{A}-\cm{A}^*\|_\textup{F}^2 \leq Cb^{2d}(C+C_2\bar{\sigma}^2b^{-2(d+1)})E,
        \end{split}
    \end{equation}
    where $C_1,C_2>0$ are some constants related to $c_0$.
\end{lemma}

The second lemma is an upper bound for the second and higher-order terms in the perturbation of a tensor Tucker decomposition, as the higher-order generalization of Lemma E.3 in \citet{han2022optimal}.

\begin{lemma}\label{lemma:perturb}
    Suppose that $\cm{A}^*=\cm{S}^*\times_{k=1}^d\bm{U}_k^*$ and $\cm{A}=\cm{S}\times_{k=1}^d\bm{U}_k$ with $\|\bm{U}_k\|\asymp\|\bm{U}_k^*\|\asymp b$ and $\|\cm{S}_{(k)}\|\asymp\|\cm{S}_{(k)}^*\|\asymp \bar{\sigma}b^{-d}$. For $\bm{O}_k\in\mathbb{O}_{r_k}$, $1\leq k\leq d$, $\|\cm{H}\|_\textup{F}\leq C_db^{-2}\bar{\sigma}\textup{Err}$, where $\cm{H}=\cm{A}^*-\cm{A}_0-\sum_{k=1}^d(\cm{A}_{k}-\cm{A})$ and $\textup{Err}=\sum_{k=1}^d\|\bm{U}_k-\bm{U}_k^*\bm{O}_k\|_\textup{F}^2+\|\cm{S}-\cm{S}^*\times_{k=1}^d\bm{O}_k^\top\|_\textup{F}^2$. Then, $\|\cm{H}\|_\textup{F}\leq C_db^{-2}\bar{\sigma}\textup{Err}$.
\end{lemma}

\begin{proof} We have that
    \begin{equation}
    \begin{split}
        \|\cm{H}\|_\text{F} & \leq \sum_{j\neq k}\left\|\cm{S}^*\times_{i=j,k}(\bm{U}_i-\bm{U}^*_i\bm{O}_i)\times_{i\neq j,k}\bm{U}_j^*\bm{O}_j\right\|_\text{F}\\
        & + \sum_{j\neq k \neq l}\left\|\cm{S}^*\times_{i=j,k,l}(\bm{U}_i-\bm{U}^*_i\bm{O}_i)\times_{i\neq j,k,l}\bm{U}_j^*\bm{O}_j\right\|_\text{F}\\
        & + \cdots + \sum_{j}\|\cm{S}^*\times_{i\neq j}(\bm{U}_i-\bm{U}^*_i\bm{O}_i)\times_{i=j}\bm{U}_j^*\bm{O}_j\|_\text{F} + \|\cm{S}^*\times_{i=1}^d(\bm{U}_i-\bm{U}^*_i\bm{O}_i)\|_\text{F}\\
        & \leq \sum_{j\neq k}\left\|(\cm{S}\times_{k=1}^d\bm{O}_k-\cm{S}^*)\times_{i=j,k}(\bm{U}_i-\bm{U}^*_i\bm{O}_i)\times_{i\neq j,k}\bm{U}_j^*\bm{O}_j\right\|_\text{F}\\
        & + \sum_{j\neq k \neq l}\left\|(\cm{S}\times_{k=1}^d\bm{O}_k-\cm{S}^*)\times_{i=j,k,l}(\bm{U}_i-\bm{U}^*_i\bm{O}_i)\times_{i\neq j,k,l}\bm{U}_j^*\bm{O}_j\right\|_\text{F}\\
        & + \cdots + \sum_{j}\|(\cm{S}\times_{k=1}^d\bm{O}_k-\cm{S}^*)\times_{i\neq j}(\bm{U}_i-\bm{U}^*_i\bm{O}_i)\times_{i=j}\bm{U}_j^*\bm{O}_j\|_\text{F}\\
        & + \|(\cm{S}\times_{k=1}^d\bm{O}_k-\cm{S}^*)\times_{i=1}^d(\bm{U}_i-\bm{U}^*_i\bm{O}_i)\|_\text{F}\\
        & \leq \binom{d}{2}B_2B_1^{d-2}B_3 + \binom{d}{3}B_2B_1^{d-3}B_3^{3/2} + \cdots + dB_2B_1B_3^{(d-1)/2}+B_2B_3^{d/2}\\
        & + \binom{d}{2}B_1^{d-2}B_3^{3/2} + \binom{d}{3}B_1^{d-3}B_3^{2} + \cdots + dB_1B_3^{d/2}+B_3^{(d+1)/2} \leq C_db^{-2}\bar{\sigma}\textup{Err},
        \end{split}
    \end{equation}
    where $B_1=\max_k(\|\bm{U}^*_k\|,\|\bm{U}_k\|)$, $B_2=\max_k(\|\cm{S}^*_{(k)}\|,\|\cm{S}_{(k)}\|)$, and $B_3=\max_k(\|\bm{U}_k-\bm{U}_k^*\bm{O}_k\|_\text{F}^2,\|\cm{S}-\cm{S}^*\times_{k=1}^d\bm{O}_k\|_\text{F}^2)$.

\end{proof}

\newpage
\section{Properties of Robust Gradient Estimators}\label{append:B}

\subsection{General Proof Strategy}

The most essential part of the statistical analysis is to prove that the robust gradient estimators are stable. For $1\leq k\leq d$, the robust gradient estimator with respect to $\bm{U}_k$ is
\begin{equation}
    \bm{G}_k=\frac{1}{n}\sum_{i=1}^n\text{T}(\nabla\overline{\mathcal{L}}(\cm{A};z_i)_{(k)}\bm{V}_k,\tau).
\end{equation}
Note that
\begin{equation}\label{eq:stable_robust_decomp}
    \begin{split}
        \bm{G}_k - \nabla_k\mathcal{R} = & ~\frac{1}{n}\sum_{i=1}^n\text{T}(\nabla\overline{\mathcal{L}}(\cm{A};z_i)_{(k)}\bm{V}_k,\tau) - \mathbb{E}[\nabla\overline{\mathcal{L}}(\cm{A};z_i)_{(k)}\bm{V}_k]\\
        = &~ T_{k,1} + T_{k,2} + T_{k,3} + T_{k,4},
    \end{split}
\end{equation}
where
\begin{equation}
    \begin{split}
        T_{k,1} = & \mathbb{E}[\text{T}(\nabla\overline{\mathcal{L}}(\cm{A}^*;z_i)_{(k)}\bm{V}_k,\tau)] - \mathbb{E}[\nabla\overline{\mathcal{L}}(\cm{A}^*;z_i)_{(k)}\bm{V}_k],\\
        T_{k,2} = & \frac{1}{n}\sum_{i=1}^n\text{T}(\nabla\overline{\mathcal{L}}(\cm{A}^*;z_i)_{(k)}\bm{V}_k,\tau) - \mathbb{E}[\text{T}(\nabla\overline{\mathcal{L}}(\cm{A}^*;z_i)_{(k)}\bm{V}_k,\tau)],\\
        T_{k,3} = & \mathbb{E}[\nabla\overline{\mathcal{L}}(\cm{A}^*;z_i)_{(k)}\bm{V}_k] - \mathbb{E}[\nabla\overline{\mathcal{L}}(\cm{A};z_i)_{(k)}\bm{V}_k]\\
        & + \mathbb{E}[\text{T}(\nabla\overline{\mathcal{L}}(\cm{A};z_i)_{(k)}\bm{V}_k,\tau)] - \mathbb{E}[\text{T}(\nabla\overline{\mathcal{L}}(\cm{A}^*;z_i)_{(k)}\bm{V}_k,\tau)],\\
        T_{k,4} = & \frac{1}{n}\sum_{i=1}^n\text{T}(\nabla\overline{\mathcal{L}}(\cm{A};z_i)_{(k)}\bm{V}_k,\tau) - \frac{1}{n}\sum_{i=1}^n\text{T}(\nabla\overline{\mathcal{L}}(\cm{A}^*;z_i)_{(k)}\bm{V}_k,\tau)\\
        & - \mathbb{E}[\text{T}(\nabla\overline{\mathcal{L}}(\cm{A};z_i)_{(k)}\bm{V}_k,\tau)] + \mathbb{E}[\text{T}(\nabla\overline{\mathcal{L}}(\cm{A}^*;z_i)_{(k)}\bm{V}_k,\tau)].
    \end{split}
\end{equation}

Similarly, for $\cm{S}$, its robust gradient estimator is $$\cm{G}_0=\frac{1}{n}\sum_{i=1}^n\text{T}(\nabla\overline{\mathcal{L}}(\cm{A};z_i)\times_{j=1}^d\bm{U}_j^\top,\tau).$$ We can also decompose $\cm{G}_0 - \mathbb{E}{\nabla_0\mathcal{L}}$ into four components,
\begin{equation}\label{eq:stable_robust_decomp}
    \begin{split}
        \cm{G}_0 - \mathbb{E}[\nabla_0\mathcal{L}] = & ~\frac{1}{n}\sum_{i=1}^n\text{T}(\nabla\overline{\mathcal{L}}(\cm{A};z_i)\times_{j=1}^d\bm{U}_j^\top,\tau) - \mathbb{E}[\nabla\overline{\mathcal{L}}(\cm{A};z_i)\times_{j=1}^d\bm{U}_j^\top,\tau)]\\
        = & ~T_{0,1} + T_{0,2} + T_{0,3} +T_{0,4},
    \end{split}
\end{equation}
where
\begin{equation}
    \begin{split}
        T_{0,1} = & \mathbb{E}[\text{T}(\nabla\overline{\mathcal{L}}(\cm{A}^*;z_i)\times_{j=1}^d\bm{U}_j^\top,\tau)] - \mathbb{E}[\nabla\overline{\mathcal{L}}(\cm{A}^*;z_i)\times_{j=1}^d\bm{U}_j^\top],\\
        T_{0,2} = & \frac{1}{n}\sum_{i=1}^n\text{T}(\nabla\overline{\mathcal{L}}(\cm{A}^*;z_i)\times_{j=1}^d\bm{U}_j^\top,\tau) - \mathbb{E}[\text{T}(\nabla\overline{\mathcal{L}}(\cm{A}^*;z_i)\times_{j=1}^d\bm{U}_j^\top,\tau)],\\
        T_{0,3} = & \mathbb{E}[\nabla\overline{\mathcal{L}}(\cm{A}^*;z_i)\times_{j=1}^d\bm{U}_j^\top] - \mathbb{E}[\nabla\overline{\mathcal{L}}(\cm{A};z_i)\times_{j=1}^d\bm{U}_j^\top]\\
        & + \mathbb{E}[\text{T}(\nabla\overline{\mathcal{L}}(\cm{A};z_i)\times_{j=1}^d\bm{U}_j^\top,\tau)] - \mathbb{E}[\text{T}(\nabla\overline{\mathcal{L}}(\cm{A}^*;z_i)\times_{j=1}^d\bm{U}_j^\top,\tau)],\\
        T_{0,4} = & \frac{1}{n}\sum_{i=1}^n\text{T}(\nabla\overline{\mathcal{L}}(\cm{A};z_i)\times_{j=1}^d\bm{U}_j^\top,\tau) - \frac{1}{n}\sum_{i=1}^n\text{T}(\nabla\overline{\mathcal{L}}(\cm{A}^*;z_i)\times_{j=1}^d\bm{U}_j^\top,\tau)\\
        & - \mathbb{E}[\text{T}(\nabla\overline{\mathcal{L}}(\cm{A};z_i)\times_{j=1}^d\bm{U}_j^\top,\tau)] + \mathbb{E}[\text{T}(\nabla\overline{\mathcal{L}}(\cm{A}^*;z_i)\times_{j=1}^d\bm{U}_j^\top,\tau)].
    \end{split}
\end{equation}
To prove the stability of the robust gradient estimators, it suffices to give proper upper bounds of $\|T_{k,j}\|_\text{F}$ for $0\leq k\leq d$ and $1\leq j\leq 4$.

Here, $T_{k,1}$ is the truncation bias at the ground truth $\cm{A}^*$, and $T_{k,2}$ represents the deviation of the truncated estimation around its expectation. As each truncated gradient, $\text{T}(\nabla\overline{\mathcal{L}}(\cm{A};z_i)_{(k)}\bm{V}_k,\tau)$, is a bounded variable, we can apply the Bernstein inequality \citep{wainwright2019high} to achieve a sub-Gaussian-type concentration without the Gaussian distributional assumption on the data. The truncation parameter $\tau$ controls the magnitude of $\|T_{k,1}\|_\text{F}$ and $\|T_{k,2}\|_\text{F}$, and an optimal $\tau$ gives $\|T_{k,1}\|_\text{F}\asymp\|T_{k,2}\|_\text{F}\asymp\xi_k$. For $T_{k,3}$, given some regularity conditions, we can obtain an upper bound for the truncation bias of the second-order approximation error in $\|T_{k,3}\|_\text{F}$. Similarly, as $\text{T}(\nabla\overline{\mathcal{L}}(\cm{A};z_i)_{(k)}\bm{V}_k,\tau) - \text{T}(\nabla\overline{\mathcal{L}}(\cm{A}^*;z_i)_{(k)}\bm{V}_k,\tau)$ is bounded, we can also achieve a sub-Gaussian-type concentration and show that $\|T_{k,3}\|_\text{F}\asymp\|T_{k,4}\|_\text{F}\lesssim\phi^{1/2}\|\cm{A}-\cm{A}^*\|_\text{F}$. Hence, we can show that $\sum_{i=1}^4\|T_{k,i}\|_\text{F}^2\lesssim\phi\|\cm{A}-\cm{A}^*\|_\text{F}^2+\xi_k^2$. 

By controlling the truncation bias, deviation, and approximation errors, we demonstrate that the truncated gradient estimator is stable and achieves optimal performance under certain conditions.
A similar approach can be applied to the gradient with respect to the core tensor $\cm{S}$, establishing the stability of the corresponding robust estimator.

\newpage

\subsection{Proof of Theorem \ref{thm:linearregression}}

\begin{proof}

The proof consists of seven steps. In Step 1, we present the local moment bounds used in partial gradients. In Steps 2 to 6, we prove the stability of the robust gradient estimators for the general $1\leq t\leq T$ and, hence, we omit the notation $(t)$ for simplicity. Specifically, in Steps 2 to 5, we give the upper bounds for $\|T_{k,1}\|_\text{F},\dots,\|T_{k,4}\|_\text{F}$, respectively, for $1\leq k\leq d_0$. In Step 6, we extend the proof to the terms for the core tensor. In the last step, we apply the results to the local convergence analysis in Theorem \ref{thm:1} and verify the corresponding conditions. Throughout the first six steps, we assume that for each $1\leq k\leq d$, $\|\bm{U}_k\|\asymp \bar{\sigma}^{1/(d+1)}$, $\max_{1\leq k\leq d}\|\cm{S}_{(k)}\|\asymp \bar{\sigma}^{1/(d+1)}$, and $\|\sin\theta(\bm{U}_k,\bm{U}_k^*)\|\leq\delta$ and will verify them in the last step.\\

\noindent\textit{Step 1.} (Calculate local moments)

\noindent For any $1\leq k\leq d_0$, we let $r_k'=r_1r_2\cdots r_{d_0}/r_k$, $\bar{r}_{d_0}=r_{d_0+1}r_{d_0+2}\cdots r_d$, and
\begin{equation}
    \nabla\overline{\mathcal{L}}(\cm{A}^*;z_i)_{(k)}\bm{V}_k=\left[(\cm{X}_i\times_{j=1,j\neq k}^{d_0}\bm{U}_{j}^\top)_{(k)}\otimes\text{vec}(-\cm{E}_i\times_{j=1}^{d-d_0}\bm{U}_{d_0+j}^\top)^\top\right]\cm{S}_{(k)}^\top.
\end{equation}
Denote the columns of $\cm{S}_{(k)}^\top$ as $\cm{S}_{(k)}^\top=[\bm{s}_{k,1},s_{k,2},\dots,\bm{s}_{k,r_k}]$ such that $\text{vec}(\bm{S}_{k,j})=\bm{s}_{k,j}$. The $(l,m)$-th entry of $\nabla\overline{\mathcal{L}}(\cm{A}^*;z_i)_{(k)}\bm{V}_k$ is
\begin{equation}
    \begin{split}
        & \left(\left[(\cm{X}_i\times_{j=1,j\neq k}^{d_0}\bm{U}_j^\top)_{(k)}\otimes\text{vec}(-\cm{E}_i\times_{j=1}^{d-d_0}\bm{U}_{d_0+j}^\top)^\top\right]\bm{s}_{k,m}\right)_{l}\\
        & = \left[(\cm{X}_i\times_{j=1,j\neq k}^{d_0}\bm{U}_j^\top)_{(k)}\bm{S}_{k,m}\text{vec}(-\cm{E}_i\times_{j=1}^{d-d_0}\bm{U}_{d_0+j}^\top)\right]_l\\
        & = \bm{c}_l^\top(\cm{X}_i)_{(k)}(\otimes_{j=1,j\neq k}^{d_0}\bm{U}_j)\bm{S}_{k,m}(\otimes_{j=d_0+1}^d\bm{U}_j^\top)\bm{e}_i,
    \end{split}
\end{equation}
where $\bm{c}_l$ is the coordinate vector whose $l$-th entry is one and the others are zero, and $\bm{e}_i=\text{vec}(-\cm{E}_i)$.

For the fixed $\bm{U}_j$'s, let $\bm{M}_{k,1}=(\otimes_{j=1,j\neq k}^{d_0}\bm{U}_j)/\|(\otimes_{j=1,j\neq k}^{d_0}\bm{U}_j)\|$ and $\bm{c}_l^\top(\cm{X}_i)_{(k)}\bm{M}_{k,1}=(w^{(i)}_{k,l,1},\dots,w^{(i)}_{k,l,r_k'})$.
Similarly, let $\bm{M}_{k,2}=(\otimes_{j=d_0+1}^d\bm{U}_j^\top)/\|(\otimes_{j=d_0+1}^d\bm{U}_j^\top)\|$ and $\bm{M}_{k,2}\bm{e}_i=(z^{(i)}_{k,1},\dots,z^{(i)}_{k,\bar{r}_{d_0}})^\top$. 
By Assumption \ref{asmp:1}, $\mathbb{E}[|w^{(i)}_{k,l,j}|^{1+\epsilon}]\leq M_{x,1+\epsilon,\delta}$ and $\mathbb{E}[|z^{(i)}_{k,{m'}}|^{1+\epsilon}]\leq M_{e,1+\epsilon,\delta}$, for $j=1,2,\dots,r_k'$, $l=1,2,\dots,p_k$, and ${m'}=1,2,\dots,\bar{r}_{d_0}$. Let $\bm{M}_{k,3,m}=\bm{S}_{k,m}/\|\bm{S}_{k,m}\|$ and $\bm{M}_{k,3,m}\bm{M}_{k,2}\bm{e}_i=(z_{k,m,1}^{(i)},\dots,z_{k,m,r_k'}^{(i)})$, for $m=1,2,\dots,r_k$. Then, $\mathbb{E}[|z_{k,m,j}^{(i)}|^{1+\epsilon}|\cm{X}_i]\lesssim M_{e,1+\epsilon,\delta}$. Let $v^{(i)}_{k,j,l,m}=w^{(i)}_{k,l,j}z^{(i)}_{k,m,j}$, which satisfies that
\begin{equation}\label{eq:moment1}
    \mathbb{E}\left[|v_{k,j,l,m}^{(i)}|^{1+\epsilon}\right]=\mathbb{E}\left[|w_{k,j,l}^{(i)}|^{1+\epsilon}\cdot\mathbb{E}\left[|z_{k,m,j}^{(i)}|^{1+\epsilon}|\cm{X}_i\right]\right]\lesssim M_{x,1+\epsilon,\delta}\cdot M_{e,1+\epsilon,\delta} = M_{\text{eff},1+\epsilon,\delta}.
\end{equation}
Let $v^{(i)}_{k,l,m}=\sum_{j=1}^{r_k'}v_{k,j,l,m}^{(i)}$ and $\mathbb{E}[|v^{(i)}_{k,l,m}|^{1+\epsilon}]\lesssim M_{\text{eff},1+\epsilon,\delta}$.

For $d_0+1\leq k\leq d$, we let $r_k'=r_{d_0+1}r_{d_0+2}\cdots r_d/r_k$ and
\begin{equation}
    \nabla{\mathcal{L}}(\cm{A}^*;z_i)_{(k)}\bm{V}_k = [(-\cm{E}_i\times_{j=1,j\neq k-d_0}^{d-d_0}\bm{U}_{d_0+j}^\top)_{(k-d_0)}\otimes\text{vec}(\cm{X}_i\times_{j=1}^{d_0}\bm{U}_{j}^\top)^\top]\cm{S}_{(k)}^\top.
\end{equation}
The $(l,m)$-th entry of $n^{-1}\sum_{i=1}^n\nabla\overline{\mathcal{L}}(\cm{A}^*;z_i)_{(k)}\bm{V}_k$ is
\begin{equation}
    \begin{split}
        & \frac{1}{n}\sum_{i=1}^n([(-\cm{E}_i\times_{j=1,j\neq k-d_0}^{d-d_0}\bm{U}_j^\top)_{(k-d_0)}\otimes\text{vec}(\cm{X}_i\times_{j=1}^{d_0}\bm{U}_j^\top)^\top]\bm{s}_{k,m})_l\\
        & = \frac{1}{n}\sum_{i=1}^n\bm{c}_l^\top(-\cm{E}_i)_{(k-d_0)}(\otimes_{j=d_0+1,j\neq k}^d\bm{U}_j)\cdot\bm{S}_{k,m}(\otimes_{j=1}^{d_0}\bm{U}_j^\top)\bm{x}_i.
    \end{split}
\end{equation}
Let $\bm{M}_{k,1}=(\otimes_{j=d_0+1,j\neq k}^d\bm{U}_j)/\|\otimes_{j=d_0+1,j\neq k}^d\bm{U}_j\|$ and $\bm{c}_l^\top(-\cm{E}_i)_{(k-d_0)}\bm{M}_{k,1}=(u_{k,l,1}^{(i)},\dots,u_{k,l,r_k'}^{(i)})$. Let $\bm{M}_{k,2}=(\otimes_{j=1}^{d_0}\bm{U}_j^\top)/\|\otimes_{j=1}^{d_0}\bm{U}_j^\top\|$ and $\bm{M}_{k,2}\bm{x}_i=(s_{k,1}^{(i)},s_{k,2}^{(i)},\dots,s_{k,r_1r_2\cdots r_{d_0}}^{(i)})^\top$. 

By Assumption \ref{asmp:1}, $\mathbb{E}[|u_{k,l,j}^{(i)}|^{1+\epsilon}|\cm{X}_i]\leq M_{e,1+\epsilon,\delta}$ and $\mathbb{E}[|s_{k,j'}^{(i)}|^{1+\epsilon}]\leq M_{x,1+\epsilon,\delta}$, for $j'=1,2,\dots,r_1r_2\cdots r_{d_0}$ and $l=1,2,\dots,p_k$. Let $\bm{M}_{k,3,m}=\bm{S}_{k,m}/\|\bm{S}_{k,m}\|$ and $\bm{M}_{k,3,m}\bm{M}_2\bm{x}_i=(s_{k,m,1}^{(i)},s_{k,m,2}^{(i)},\dots,s_{k,m,r_k'}^{(i)})$, where $\mathbb{E}[|s_{k,m,j}^{(i)}|^{1+\epsilon}]\lesssim M_{x,1+\epsilon,\delta}$. Let $r_{k,j,l,m}^{(i)}=u_{k,l,j}^{(i)}s_{k,m,j}^{(i)}$ and
\begin{equation}\label{eq:moment2}
    \mathbb{E}\left[|r_{k,j,l,m}^{(i)}|^{1+\epsilon}\right]=\mathbb{E}\left[|u_{k,j,l}^{(i)}|^{1+\epsilon}\cdot\mathbb{E}\left[|s_{k,m,j}^{(i)}|^{1+\epsilon}|\cm{X}_i\right]\right]\lesssim M_{x,1+\epsilon,\delta}\cdot M_{e,1+\epsilon,\delta} = M_{\text{eff},1+\epsilon,\delta}.
\end{equation}

In addition, for any $1\leq k \leq d$, we let $\bm{V}_k=[\bm{v}_{k,1},\dots,\bm{v}_{k,r_k}]$. The $(l,m)$-th entry of $\nabla\overline{\mathcal{L}}(\cm{A}^*;z_i)_{(k)}\bm{V}_k - \nabla\overline{\mathcal{L}}(\cm{A};z_i)_{(k)}\bm{V}_k$ is 
\begin{equation}
    \begin{split}
        & \bm{c}_l^\top[\cm{X}_i\circ \langle\cm{A}^*-\cm{A},\cm{X}_i\rangle]_{(k)}\bm{v}_{k,m} \\
        & = (\bm{v}_{k,m}\otimes\bm{c}_l)^\top \text{vec}((\cm{X}_i)_{(k)})\text{vec}(\cm{X}_i)^\top\text{vec}(\cm{A}^*-\cm{A})\\
        & = (\bm{v}_{k,m}\otimes\bm{c}_l)^\top\bm{P}_k^\top\text{vec}(\cm{X}_i)\text{vec}(\cm{X}_i)^\top\text{vec}(\cm{A}^*-\cm{A}).
    \end{split}
\end{equation}
Let $\bm{w}_{k,m,l}=\bm{P}_k(\bm{v}_{k,m}\otimes\bm{c}_l)/\|\bm{P}_k(\bm{v}_{k,m}\otimes\bm{c}_l)\|_2$. Then, we have
\begin{equation}\label{eq:moment3}
    \begin{split}
        & \mathbb{E}\left[\left|q_{k,m,l}^{(i)}\right|^{1+\lambda}\right] := \mathbb{E}\left[\left|\bm{w}_{k,m,l}^\top\text{vec}(\cm{X}_i)\text{vec}(\cm{X}_i)^\top\text{vec}(\cm{A}^*-\cm{A})\right|^{1+\lambda}\right]\\
        & \leq \mathbb{E}\left[\left|\bm{w}_{k,m,l}^\top\text{vec}(\cm{X}_i)\right|^{2+2\lambda}\right]^{1/2}\cdot\mathbb{E}\left[\left|\text{vec}(\cm{X}_i)^\top\frac{\text{vec}(\cm{A}^*-\cm{A})}{\|\text{vec}(\cm{A}^*-\cm{A})\|_2}\right|^{2+2\lambda}\right]^{1/2}\cdot\|\cm{A}^*-\cm{A}\|_\text{F}^{1+\lambda}\\
        & \leq M_{x,2+2\lambda}\cdot \|\cm{A}-\cm{A}^*\|_\text{F}^{1+\lambda}.
    \end{split}
\end{equation}~

\noindent\textit{Step 2.} (Bound $\|T_{k,1}\|_\text{F}$)

We first bound the bias, namely $T_{k,1}$ in \eqref{eq:stable_robust_decomp}. We have that
\begin{equation}
    \|T_{k,1}\|_\text{F}^2\asymp\bar{\sigma}^{\frac{2d}{d+1}}\sum_{l=1}^{p_k}\sum_{m=1}^{r_k}\left|\mathbb{E}[v^{(i)}_{k,l,m}] - \mathbb{E}[\text{T}(v^{(i)}_{k,l,m},\tau_k)]\right|^2,
\end{equation}
where $\tau_k=\tau\cdot\|\otimes_{j=1,j\neq k}^d\bm{U}_j\|^{-1}\cdot(\max_{1\leq m\leq r_k}\|\bm{S}_{k,m}\|)^{-1}\asymp[nM_{\text{eff},1+\epsilon,\delta}/\log(\bar{p})]^{1/(1+\epsilon)}$.

For any $l=1,2,\dots,p_k$ and $m=1,2,\dots,r_k$, by definition of the truncation operator $\text{T}(\cdot,\cdot)$, local moment condition in \eqref{eq:moment1}, and Markov's inequality,
\begin{equation}
    \begin{split}
        & \left|\mathbb{E}\left[v^{(i)}_{k,l,m}\right] - \mathbb{E}\left[\text{T}(v^{(i)}_{k,l,m},\tau_k)\right]\right| \leq \mathbb{E}\left[|v^{(i)}_{k,l,m}|\cdot1\{|v^{(i)}_{k,l,m}|\geq\tau_k\}\right]\\
        \leq &~ \mathbb{E}\left[|v^{(i)}_{k,l,m}|^{1+\epsilon}\right]^{1/(1+\epsilon)}\cdot\mathbb{P}(|v^{(i)}_{k,l,m}|\geq\tau_k)^{\epsilon/(1+\epsilon)}\\
        \leq &~ \mathbb{E}\left[|v^{(i)}_{k,l,m}|^{1+\epsilon}\right]^{1/(1+\epsilon)}\left(\frac{\mathbb{E}\left[|v^{(i)}_{k,l,m}|^{1+\epsilon}\right]}{\tau_k^{1+\epsilon}}\right)^{\epsilon/(1+\epsilon)}\\
        \asymp &~ M_{\text{eff},1+\epsilon,\delta}\cdot\tau_k^{-\epsilon} \asymp \left[\frac{M_{\text{eff},1+\epsilon,\delta}^{1/\epsilon}\log(\bar{p})}{n}\right]^{\epsilon/(1+\epsilon)}
    \end{split}
\end{equation}
with truncation parameter $\tau_k\asymp\left[nM_{\text{eff},1+\epsilon,\delta}/\log(\bar{p})\right]^{1/(1+\epsilon)}$.

Hence, for $k=1,\dots,d_0$,
\begin{equation}
    \left\|T_{k,1}\right\|_\text{F}\lesssim \bar{\sigma}^{d/(d+1)}\sqrt{p_kr_k}\left[\frac{M_{\text{eff},1+\epsilon,\delta}^{1/\epsilon}\log(\bar{p})}{n}\right]^{\frac{\epsilon}{1+\epsilon}}.
\end{equation}
The results for $k=d_0+1,\dots,d$ can be derived similarly by the condition in \eqref{eq:moment2}.\\

\noindent\textit{Step 3.} (Bound $\|T_{k,2}\|_\text{F}$)

\noindent For $T_{k,2}$ in \eqref{eq:stable_robust_decomp} and $k=1,2,\dots,d_0$, similarly to $T_{k,1}$,
\begin{equation}
    \begin{split}
        &\|T_{k,2}\|_\text{F}^2
        \asymp\bar{\sigma}^{\frac{2d}{d+1}}\sum_{1\leq l\leq p_k,1\leq m\leq r_k}\left|\frac{1}{n}\sum_{i=1}^n\text{T}(v^{(i)}_{k,l,m},\tau_k)-\mathbb{E}[\text{T}(v^{(i)}_{k,l,m},\tau_k)]\right|^2.
    \end{split}
\end{equation}
For each $i=1,2,\dots,n$, it can be checked that
\begin{equation}
    \mathbb{E}\left[\text{T}(v^{(i)}_{k,l,m},\tau_k)^2\right] \leq \tau_k^{1-\epsilon}\cdot\mathbb{E}\left[|v^{(i)}_{k,l,m}|^{1+\epsilon}\right] \asymp \tau_k^{1-\epsilon}\cdot M_{\text{eff},1+\epsilon,\delta}.
\end{equation}
Thus, by the nature of truncation and local moment condition in \eqref{eq:moment1}, we have the upper bound for the variance
\begin{equation}
    \text{var}(\text{T}(v^{(i)}_{k,l,m},\tau_k)) \leq\mathbb{E}\left[\text{T}(v^{(i)}_{k,l,m},\tau_k)^2\right]\lesssim \tau_k^{1-\epsilon}\cdot M_{\text{eff},1+\epsilon,\delta}.
\end{equation}
Also, for any $q=3,4,\dots$, the higher-order moments satisfy that
\begin{equation}
    \begin{split}
        &\mathbb{E}\left[\left|\text{T}(v^{(i)}_{k,l,m},\tau_k)-\mathbb{E}[\text{T}(v^{(i)}_{k,l,m},\tau_k)]\right|^q\right]\leq (2\tau_k)^{q-2}\cdot\mathbb{E}\left[\left(\text{T}(v^{(i)}_{k,l,m},\tau_k)-\mathbb{E}[\text{T}(v^{(i)}_{k,l,m},\tau_k)]\right)^2\right].
    \end{split}
\end{equation}
By Bernstein's inequality, for any $1\leq l\leq p_k$, $1\leq m\leq r_k$, and $0<t\lesssim \tau_k^{-\epsilon}M_{\text{eff},1+\epsilon,\delta}$,
\begin{equation}
    \begin{split}
        & \mathbb{P}\left(\left|\frac{1}{n}\sum_{i=1}^n\text{T}(v^{(i)}_{k,l,m},\tau_k)-\mathbb{E}\text{T}(v^{(i)}_{k,l,m},\tau_k)\right|\geq t\right)
        \leq 2\exp\left(-\frac{nt^2}{4\tau_k^{1-\epsilon}M_{\text{eff},1+\epsilon,\delta}}\right).
    \end{split}
\end{equation}
Let $t=CM_{\text{eff},1+\epsilon,\delta}^{1/(1+\epsilon)}\log(\bar{p})^{\epsilon/(1+\epsilon)}n^{-\epsilon/(1+\epsilon)}$. Therefore, we have
\begin{equation}
    \begin{split}
        & \mathbb{P}\Bigg(\Bigg|\frac{1}{n}\sum_{i=1}^n\text{T}(v^{(i)}_{k,l,m},\tau_k)-\mathbb{E}\text{T}(v^{(i)}_{k,l,m},\tau_k)\Bigg|\gtrsim \left[\frac{M_{\text{eff},1+\epsilon,\delta}^{1/\epsilon}\log(\bar{p})}{n}\right]^{\epsilon/(1+\epsilon)}\Bigg)\\ 
        & \leq C\exp\left(-C\log(\bar{p})\right)
    \end{split}
\end{equation}
and
\begin{equation}
    \begin{split}
        & \mathbb{P}\left(\max_{\substack{1\leq l\leq p_k\\1\leq m\leq r_k}}\Bigg|\frac{1}{n}\sum_{i=1}^n\text{T}(v^{(i)}_{k,l,m},\tau_k)-\mathbb{E}\text{T}(v^{(i)}_{k,l,m},\tau_k)\Bigg|\gtrsim \left[\frac{M_{\text{eff},1+\epsilon,\delta}^{1/\epsilon}\log(\bar{p})}{n}\right]^{\epsilon/(1+\epsilon)}\right)\\
        & \leq Cp_kr_k\exp\left(-C\log(\bar{p})\right)\leq C\exp(-C\log(\bar{p})).
    \end{split}
\end{equation}
Hence, for $1\leq k\leq d_0$, with high probability at least $1-C\exp(-C\log(\bar{p}))$,
\begin{equation}
    \begin{split}
        &\left\|\frac{1}{n}\sum_{i=1}^n\text{T}(\nabla\overline{\mathcal{L}}(\cm{A}^*;z_i)_{(k)}\bm{V}_k,\tau) - \mathbb{E}[\text{T}(\nabla\overline{\mathcal{L}}(\cm{A}^*;z_i)_{(k)}\bm{V}_k,\tau)]\right\|_\text{F}\\
        & \lesssim \bar{\sigma}^{d/(d+1)}\sqrt{p_kr_k}\left[\frac{M_{\text{eff},1+\epsilon,\delta}^{1/\epsilon}\log(\bar{p})}{n}\right]^{\epsilon/(1+\epsilon)}.
    \end{split}
\end{equation}
Similarly, by the nature of truncation operator and local moment condition in \eqref{eq:moment3}, the same result can be obtained for $k=d_0+1,\dots,d$.\\

\noindent\textit{Step 4.} (Bound $\|T_{k,3}\|_\text{F}$ for $1\leq k\leq d_0$)

\noindent By definition, the $(l,m)$-th entry of $T_{k,3}$ can be bounded as
\begin{equation}
    \begin{split}
        |(T_{k,3})_{l,m}| & \asymp \bar{\sigma}^{\frac{d}{d+1}} \cdot \left|\mathbb{E}[q_{k,m,l}^{(i)}] - \mathbb{E}\left[\text{T}(q_{k,m,l}^{(i)} + v_{k,l,m}^{(i)},\tau_k) - \text{T}(v_{k,l,m}^{(i)},\tau_k)\right]\right|.
    \end{split}
\end{equation}
By the nature of truncation operator, local moment condition in \eqref{eq:moment3}, and Markov's inequality, similarly to step 2,
\begin{equation}
    \begin{split}
        & \left|\mathbb{E}[q_{k,m,l}^{(i)}] - \mathbb{E}\left[\text{T}(q_{k,m,l}^{(i)} + v_{k,l,m}^{(i)},\tau_k) - \text{T}(v_{k,l,m}^{(i)},\tau_k)\right]\right| \\
        & \leq \left|\mathbb{E}[q_{k,m,l}^{(i)}\cdot 1\{|(|v_{k,l,m}^{(i)}|\geq \tau_k)\cup(|q_{k,l,m}^{(i)} + v_{k,l,m}^{(i)}|\geq \tau_k)\}]\right| \\
        & \leq \left|\mathbb{E}[q_{k,m,l}^{(i)}\cdot 1\{|(|v_{k,l,m}^{(i)}|\geq \tau_k)\cup(|q_{k,l,m}^{(i)}|\geq \tau_k/2)\cup(|v_{k,l,m}^{(i)}|\geq \tau_k/2)\}]\right| \\
        & \leq \left|\mathbb{E}[q_{k,m,l}^{(i)}\cdot 1\{|q_{k,l,m}^{(i)}|\geq \tau_k/2\}]\right| + \left|\mathbb{E}[q_{k,m,l}^{(i)}\cdot 1\{|v_{k,l,m}^{(i)}|\geq \tau_k/2\}]\right|\\
        & \lesssim \mathbb{E}\left[|q_{k,m,l}^{(i)}|^{1+\lambda}\right]^{\frac{1}{1+\lambda}}\cdot\left(\frac{\mathbb{E}[|q_{k,l,m}^{(i)}|^{1+\lambda}]}{\tau_k^{1+\lambda}}\right)^{\frac{\lambda}{1+\lambda}} + \mathbb{E}\left[|q_{k,m,l}^{(i)}|^{1+\lambda}\right]^{\frac{1}{1+\lambda}}\cdot\left(\frac{\mathbb{E}[|v_{k,l,m}^{(i)}|^{1+\epsilon}]}{\tau_k^{1+\epsilon}}\right)^{\frac{\lambda}{1+\lambda}}\\
        & \lesssim \mathbb{E}\left[|q_{k,m,l}^{(i)}|^{1+\lambda}\right]\cdot\tau_k^{-\lambda} + \mathbb{E}\left[|q_{k,m,l}^{(i)}|^{1+\lambda}\right]^{\frac{1}{1+\lambda}}\cdot\mathbb{E}\left[|v_{k,m,l}^{(i)}|^{1+\epsilon}\right]^{(1+\epsilon)\lambda/(1+\lambda)}\cdot\tau_k^{-(1+\epsilon)\lambda/(1+\lambda)}\\
        & \lesssim \left\{\bar{\sigma}^\lambda M_{x,2+2\lambda}\left[\frac{\log(\bar{p})}{nM_{\text{eff},1+\epsilon,\delta}}\right]^{\frac{\lambda}{1+\epsilon}} + M_{x,2+2\lambda}^{1/(1+\lambda)}M_{\text{eff},1+\epsilon,\delta}^{\lambda/(1+\lambda)}\left[\frac{\log(\bar{p})}{nM_{\text{eff},1+\epsilon,\delta}}\right]^{\frac{\lambda}{1+\lambda}}\right\}\|\cm{A}-\cm{A}^*\|_\text{F}\\
        & \lesssim M_{x,2+2\lambda}^{1/(1+\lambda)}\left[\bar{\sigma}^\lambda M_{x,2+2\lambda}^{\lambda/(1+\lambda)}M_{\text{eff},1+\epsilon,\delta}^{-\lambda/(1+\epsilon)}\log(\bar{p})^{\frac{\lambda}{1+\epsilon}}n^{-\frac{\lambda}{1+\epsilon}} + \log(\bar{p})^{\frac{\lambda}{1+\lambda}}n^{-\frac{\lambda}{1+\lambda}}\right]\|\cm{A}-\cm{A}^*\|_\text{F}\\
        & \lesssim \bar{\sigma}^{\lambda}M_{x,2+2\lambda}\left[\frac{\log(\bar{p})}{n}\right]^{\min\left(\frac{\lambda}{1+\lambda},\frac{\lambda}{1+\epsilon}\right)}\|\cm{A}-\cm{A}^*\|_\text{F}.
    \end{split}
\end{equation}

Therefore, we have
\begin{equation}
    \begin{split}
        \|T_{k,3}\|_\text{F}^2 \lesssim &~ \bar{\sigma}^{2d/(d+1)}\phi_{\epsilon,\lambda,\delta}\|\cm{A}-\cm{A}^*\|_\text{F}^2,
    \end{split}
\end{equation}
where $\phi_{\lambda,\epsilon}=\bar{p}\bar{\sigma}^{2\lambda}M_{x,2+2\lambda}^2[\log(\bar{p})/n]^{2\min(\lambda/(1+\lambda),\lambda/(1+\epsilon))}.$ \\

\noindent\textit{Step 5.} (Bound $\|T_{k,4}\|_\text{F}$)

\noindent For $T_{k,4}$,
\begin{equation}
    \begin{split}   
        \|T_{k,4}\|_\text{F}^2 \asymp \bar{\sigma}^{\frac{2d}{d+1}}\sum_{1\leq l\leq p_k,1\leq m\leq r_k}\Bigg| & \frac{1}{n}\sum_{i=1}^n\left[\text{T}(q_{k,m,l}^{(i)},\tau_k + v_{k,m,l}^{(i)},\tau_k) - \text{T}(v_{k,m,l}^{(i)},\tau_k)\right]\\
        & - \mathbb{E}\left[\text{T}(q_{k,m,l}^{(i)},\tau_k + v_{k,m,l}^{(i)},\tau_k) - \text{T}(v_{k,m,l}^{(i)},\tau_k)\right] \Bigg|^2
    \end{split}
\end{equation}
For each $i=1,2,\dots,n$, we have $|\text{T}(q_{k,m,l}^{(i)},\tau_k + v_{k,m,l}^{(i)},\tau_k) - \text{T}(v_{k,m,l}^{(i)},\tau_k)|\leq 2\tau_k$, and hence,
\begin{equation}
    \begin{split}
        & \mathbb{E}[(\text{T}(q_{k,m,l}^{(i)},\tau_k + v_{k,m,l}^{(i)},\tau_k) - \text{T}(v_{k,m,l}^{(i)},\tau_k))^2] \leq \tau_k^{1-\lambda}\cdot\mathbb{E}[|q_{k,m,l}^{(i)}|^{1+\lambda}]\\
        & \asymp \tau_k^{1-\lambda}M_{x,2+2\lambda}\|\cm{A}-\cm{A}^*\|_\text{F}^{1+\lambda}.
    \end{split}
\end{equation}
In addition, for any $q=3,4,\dots$, the higher-order moments satisfy that
\begin{equation}
    \begin{split}
        & \mathbb{E}[(\text{T}(q_{k,m,l}^{(i)},\tau_k + v_{k,m,l}^{(i)},\tau_k) - \text{T}(v_{k,m,l}^{(i)},\tau_k))^q]\\
        & \leq (2\tau_k)^{q-2}\cdot\mathbb{E}[(\text{T}(q_{k,m,l}^{(i)},\tau_k + v_{k,m,l}^{(i)},\tau_k) - \text{T}(v_{k,m,l}^{(i)},\tau_k))^2].
    \end{split}
\end{equation}
By Bernstein's inequality, for any $1\leq l\leq p_k$ and $1\leq m\leq r_k$,
\begin{equation}
    \begin{split}   
        & \mathbb{P}\Bigg(\Bigg| \frac{1}{n}\sum_{i=1}^n\left[\text{T}(q_{k,m,l}^{(i)},\tau_k + v_{k,m,l}^{(i)},\tau_k) - \text{T}(v_{k,m,l}^{(i)},\tau_k)\right]\\
        & - \mathbb{E}\left[\text{T}(q_{k,m,l}^{(i)},\tau_k + v_{k,m,l}^{(i)},\tau_k) - \text{T}(v_{k,m,l}^{(i)},\tau_k)\right] \Bigg|\geq t\Bigg)\\
        & \leq 2\exp\left(-\frac{Cnt^2}{\tau_k^{1-\lambda}M_{x,2+2\lambda}\|\cm{A}-\cm{A}^*\|_\text{F}^{1+\lambda} + \tau_k t}\right).
    \end{split}
\end{equation}
If $\|\cm{A}-\cm{A}^*\|_\text{F}\lesssim M_{x,2+2\lambda}^{-1/(1+\lambda)}\cdot M_{\text{eff},1+\epsilon,\delta}^{1/(1+\epsilon)}$, letting $t=C[M_{\text{eff},1+\epsilon,\delta}^{1/\epsilon}\log(\bar{p})/n]^{\epsilon/(1+\epsilon)}$,
\begin{equation}
    \begin{split}
        & \mathbb{P}\Bigg(\max_{m,l}\Bigg| \frac{1}{n}\sum_{i=1}^n\left[\text{T}(q_{k,m,l}^{(i)},\tau_k + v_{k,m,l}^{(i)},\tau_k) - \text{T}(v_{k,m,l}^{(i)},\tau_k)\right]\\
        & - \mathbb{E}\left[\text{T}(q_{k,m,l}^{(i)},\tau_k + v_{k,m,l}^{(i)},\tau_k) - \text{T}(v_{k,m,l}^{(i)},\tau_k)\right] \Bigg|\geq C\left[\frac{M_{\text{eff},1+\epsilon,\delta}^{1/\epsilon}\log(\bar{p})}{n}\right]^{\frac{\epsilon}{1+\epsilon}}\Bigg)\\
        & \lesssim p_kr_k\exp(-C\log(\bar{p})) \leq C\exp(-C\log(\bar{p})). 
    \end{split}
\end{equation}
If $\|\cm{A}-\cm{A}^*\|_\text{F}\gtrsim M_{x,2+2\lambda}^{-1/(1+\lambda)}\cdot M_{\text{eff},1+\epsilon,\delta}^{1/(1+\epsilon)}$, then
\begin{equation}
    \|\cm{A}-\cm{A}^*\|_\text{F}^{1+\lambda} \lesssim \|\cm{A}-\cm{A}^*\|_\text{F}^2\cdot M_{x,2+2\lambda}^{(1-\lambda)/(1+\lambda)}\cdot M_{\text{eff},1+\epsilon,\delta}^{(\lambda-1)/(1+\epsilon)},
\end{equation}
and letting $t=CM_{x,2+2\lambda}[\log(\bar{p})/n]^{\min\left(\frac{\lambda}{1+\lambda},\frac{\lambda}{1+\epsilon}\right)}\|\cm{A}-\cm{A}^*\|_\text{F}$,
\begin{equation}
    \begin{split}
        & \mathbb{P}\Bigg[\max_{1\leq m\leq p_k,1\leq l\leq r_k}\Bigg| \frac{1}{n}\sum_{i=1}^n\left[\text{T}(q_{k,m,l}^{(i)},\tau_k + v_{k,m,l}^{(i)},\tau_k) - \text{T}(v_{k,m,l}^{(i)},\tau_k)\right]\\
        & - \mathbb{E}\left[\text{T}(q_{k,m,l}^{(i)},\tau_k + v_{k,m,l}^{(i)},\tau_k) - \text{T}(v_{k,m,l}^{(i)},\tau_k)\right] \Bigg|\geq t\Bigg]\\
        & \lesssim p_kr_k\exp(-C\log(\bar{p})) \leq C\exp(-C\log(\bar{p})). 
    \end{split}
\end{equation}

Combining these two cases, we have
\begin{equation}
    \|T_{k,4}\|_\text{F}^2 \lesssim \bar{\sigma}^{\frac{2d}{d+1}}\phi_{\lambda,\epsilon}\|\cm{A}-\cm{A}^*\|_\text{F}^2 + \bar{\sigma}^{\frac{2d}{d+1}} p_kr_k\left[\frac{M_{\text{eff},1+\epsilon,\delta}^{1/\epsilon}\log(\bar{p})}{n}\right]^{\frac{\epsilon}{1+\epsilon}}.
\end{equation}

Based on the results in steps 2 to 5, we have
\begin{equation}
    \begin{split}   
        \sum_{j=1}^4\|T_{k,j}\|_\text{F}^2 \lesssim &~ \bar{\sigma}^{\frac{2d}{d+1}}p_kr_k\left[\frac{M_{\text{eff},1+\epsilon,\delta}^{1/\epsilon}\log(\bar{p})}{n}\right]^{\frac{2\epsilon}{1+\epsilon}} + \bar{\sigma}^{\frac{2d}{d+1}}\phi_{\lambda,\epsilon}\|\cm{A}-\cm{A}^*\|_\text{F}^2.
    \end{split}
\end{equation}~

\noindent\textit{Step 6.} (Extension to core tensor)

For the partial gradient with respect to the core tensor $\cm{S}$, we have
\begin{equation}
    \nabla\overline{\mathcal{L}}(\cm{A}^*;z_i)\times_{j=1}^d\bm{U}_j^\top = (\cm{X}_i\times_{j=1}^{d_0}\bm{U}_j^\top)\circ(-\cm{E}_i\times_{j=d_0+1}^d\bm{U}_j^\top).
\end{equation}
Let $\bm{M}_{0,1}=\otimes_{j=1}^{d_0}\bm{U}_j/\|\otimes_{j=1}^{d_0}\bm{U}_j\|$ and $\bm{M}_{0,1}^\top\bm{x}_i=(w_{0,1}^{(i)},\dots,w_{0,r_1r_2\cdots r_{d_0}})^\top$, and let $\bm{M}_{0,2}=\otimes_{j=d_0+1}^{d}\bm{U}_j/\|\otimes_{j=d_0+1}^{d}\bm{U}_j\|$ and $\bm{M}_{0,2}^\top\bbm{e}_i=(z_{0,1}^{(i)},\dots,z_{0,r_{d_0+1}r_{d_0+2}\cdots r_d})^\top$. By Assumption \ref{asmp:1}, $\mathbb{E}[|w_{0,j}^{(i)}|^{1+\epsilon}|\cm{X}_i]\leq M_{x,1+\epsilon,\delta}$ and $\mathbb{E}[|z_{0,m}^{(i)}|^{1+\epsilon}|\cm{X}_i]\leq M_{e,1+\epsilon,\delta}$, for all $j=1,2,\dots,r_1r_2\cdots r_{d_0}$ and $m=1,2,\dots,r_{d_0+1}r_{d_0+2}\cdots r_d$. Let $v_{0,j,m}^{(i)}=w_{0,j}^{(i)}z_{0,m}^{(i)}$.

In a similar fashion, we can show that with probability at least $1-C\exp(-C\log(\bar{p}))$,
\begin{equation}
    \begin{split}
        \|T_{0,1}\|_\text{F}&\lesssim\bar{\sigma}^{d/(d+1)}\sqrt{r_1r_2\cdots r_d}\left[\frac{M_{\text{eff},1+\epsilon,\delta}^{1/\epsilon}\log(\bar{p})}{n}\right]^{\epsilon/(1+\epsilon)},\\
        \|T_{0,2}\|_\text{F}&\lesssim\bar{\sigma}^{d/(d+1)}\sqrt{r_1r_2\cdots r_d}\left[\frac{M_{\text{eff},1+\epsilon,\delta}^{1/\epsilon}\log(\bar{p})}{n}\right]^{\epsilon/(1+\epsilon)},\\
        \|T_{0,3}\|_\text{F}&\lesssim C\phi_{\lambda,\epsilon}^{1/2}\bar{\sigma}^{d/(d+1)}\|\cm{A}-\cm{A}^*\|_\text{F},\\
        \|T_{0,4}\|_\text{F}&\lesssim C\phi_{\lambda,\epsilon}^{1/2}\bar{\sigma}^{d/(d+1)}\|\cm{A}-\cm{A}^*\|_\text{F} + \bar{\sigma}^{d/(d+1)}\sqrt{r_1r_2\cdots r_d}\left[\frac{M_{\text{eff},1+\epsilon,\delta}^{1/\epsilon}\log(\bar{p})}{n}\right]^{\epsilon/(1+\epsilon)}.
    \end{split}
\end{equation}
Hence, with probability at least $1-C\exp(-C\log(\bar{p}))$,
\begin{equation}
    \|\cm{G}_0-\mathbb{E}[\nabla_0\mathcal{L}]\|_\text{F}^2\lesssim \bar{\sigma}^{\frac{2d}{d+1}}\phi_{\lambda,\epsilon}\|\cm{A}-\cm{A}^*\|_\text{F}^2+\bar{\sigma}^{\frac{2d}{d+1}}\prod_{k=1}^dr_k\left[\frac{M_{\text{eff},1+\epsilon,\delta}^{1/\epsilon}\log(\bar{p})}{n}\right]^{\frac{2\epsilon}{1+\epsilon}}.
\end{equation}~

\noindent\textit{Step 7.} (Verify the conditions and conclude the proof)

In the last step, we apply the results above to Theorem \ref{thm:1}. First, we examine the conditions in Theorem \ref{thm:1} hold. Under Assumption \ref{asmp:1}, by Lemma 3.11 in \citet{bubeck2015convex}, we can show that the RCG condition in Definition \ref{def:2} is implied by the restricted strong convexity and strong smoothness with $\alpha=\alpha_x$ and $\beta=\beta_x$.

Next, we show the stability of the robust gradient estimators for all $t=1,2,\dots,T$.
By matrix perturbation theory, if $\|\cm{A}^{(0)}-\cm{A}^*\|_\text{F}\leq\sqrt{\alpha_x/\beta_x}\underline{\sigma}\kappa^{-2}\delta$, we have $\|\sin\Theta(\bm{U}_k^{(0)},\bm{U}_k^*)\|\leq\delta$ for all $k=1,\dots,d$. After a finite number of iterations, $C_T$, with probability at least $1-C_T\exp(-C\log(\bar{p}))$, we can have $\|\sin\Theta(\bm{U}_k^{(C_T)},\bm{U}_k^*)\|\leq\delta'<(4\sqrt{2})^{-1}$.

For any $l\neq k$ and any tensor $\cm{B}\in\mathbb{R}^{p_1\times\cdots\times p_d}$, $(\cm{B}\times_{j\neq k}\bm{U}_j^\top)_{(l)}=\bm{U}_l^\top\cm{B}_{(l)}(\otimes_{j\neq l}\bm{U}_j')$, where $\bm{U}_j'=\bm{U}_j$ for $j\neq k$ and $\bm{U}_k'=\bm{I}_{r_k}$.
For any $\bm{U}_l\in\mathcal{C}(\bm{U}_l^*,\delta')$, we have $\|\bm{U}_l-\bm{U}^*_l\bm{O}_l\|\leq\sqrt{2}\|\sin\Theta(\bm{U}_l,\bm{U}_l^*)\|\leq\sqrt{2}\delta'$ for some $\bm{O}_l\in\mathbb{O}^{r_k\times r_k}$. Let $\bm{\Delta}_l=\bm{U}_l-\bm{U}^*_l\bm{O}_l$ and decompose $\bm{\Delta}_l=\bm{\Delta}_{l,1}+\bm{\Delta}_{l,2}$ where $\langle\bm{\Delta}_{l,1},\bm{\Delta}_{l,2}\rangle=0$ and $\bm{\Delta}_{l,1}/\|\bm{\Delta}_{l,1}\|,\bm{\Delta}_{l,2}/\|\bm{\Delta}_{l,2}\|\in\mathcal{C}(\bm{U}_l^*,\delta')$. Thus, we have $\|\bm{\Delta}_{l,1}\|\leq\sqrt{2}\delta'$ and $\|\bm{\Delta}_{l,2}\|\leq\sqrt{2}\delta'$.

Denote $\xi=\sup_{\bm{U}_l\in\mathcal{C}(\bm{U}_l^*,\delta')}\|\bm{U}_l^\top\cm{B}_{(l)}(\otimes_{j\neq l}\bm{U}_j')\|_\text{F}$. Then, since
\begin{equation}
    \begin{split}
        & \|\bm{U}_l^\top\cm{B}_{(l)}(\otimes_{j\neq l}\bm{U}_j')\|_\text{F}\\
        & \leq \|(\bm{U}^*_l\bm{O}_l)^\top\cm{B}_{(l)}(\otimes_{j\neq l}\bm{U}_j')\|_\text{F} + \|\bm{\Delta}_l^\top\cm{B}_{(l)}(\otimes_{j\neq l}\bm{U}_j')\|_\text{F}\\
        & \leq \|(\bm{U}^*_l\bm{O}_l)^\top\cm{B}_{(l)}(\otimes_{j\neq l}\bm{U}_j')\|_\text{F} + \|\bm{\Delta}_{l,1}\|\cdot\|(\bm{\Delta}_{l,1}/\|\bm{\Delta}_{l,1}\|)^\top\nabla\cm{B}_{(l)}(\otimes_{j\neq l}\bm{U}_j')\|_\text{F}\\
        & + \|\bm{\Delta}_{l,2}\|\cdot\|(\bm{\Delta}_{l,2}/\|\bm{\Delta}_{l,1}\|)^\top\nabla\cm{B}_{(l)}(\otimes_{j\neq l}\bm{U}_j')\|_\text{F},
    \end{split}
\end{equation}
we have that
\begin{equation}
    \xi \leq \|(\bm{U}^*_l\bm{O}_l)^\top\nabla\cm{B}_{(l)}(\otimes_{j\neq l}\bm{U}_j')\|_\text{F} + (\|\bm{\Delta}_{l,1}\| + \|\bm{\Delta}_{l,2}\|)\xi,
\end{equation}
that is, taking $\delta'=1/8$,
\begin{equation}
    \xi \leq (1-2\sqrt{2}\delta')^{-1}\|(\bm{U}^*_l\bm{O}_l)^\top\nabla\cm{B}_{(l)}(\otimes_{j\neq l}\bm{U}_j')\|_\text{F}\leq2\|(\bm{U}^*_l\bm{O}_l)^\top\nabla\cm{B}_{(l)}(\otimes_{j\neq l}\bm{U}_j')\|_\text{F}.
\end{equation}
Hence, for the iterate $t=1,2,\dots,T$, combining the results in steps 1 to 6, we have that with probability at least $1-C\exp(-C\log(\bar{p}))$, for any $k=1,2,\dots,d$
\begin{equation}
    \|\bm{G}_k^{(t)}-\mathbb{E}[\nabla_k\mathcal{L}^{(t)}]\|_\text{F}^2 \lesssim \phi_{\lambda,\epsilon}\bar{\sigma}^{2d/(d+1)}\|\cm{A}^{(t)}-\cm{A}^*\|_\text{F}^2 + \bar{\sigma}^{2d/(d+1)}(p_kr_k)\left[\frac{M_{\text{eff},1+\epsilon,\delta}^{1/\epsilon}\log(\bar{p})}{n}\right]^{\frac{2\epsilon}{1+\epsilon}}
\end{equation}
and
\begin{equation}
    \|\cm{G}_0^{(t)}-\mathbb{E}[\nabla_0\mathcal{L}^{(t)}]\|_\text{F}^2\lesssim \phi_{\lambda,\epsilon}\bar{\sigma}^{2d/(d+1)}\|\cm{A}^{(t)}-\cm{A}^*\|_\text{F}^2 + \bar{\sigma}^{2d/(d+1)}\prod_{k=1}^dr_k\left[\frac{M_{\text{eff},1+\epsilon,\delta}^{1/\epsilon}\log(\bar{p})}{n}\right]^{\frac{2\epsilon}{1+\epsilon}}.
\end{equation}
As the sample size satisfies
\begin{equation}
    n\gtrsim \left[\sqrt{\bar{p}}\alpha_x^{-1}\kappa^2 M_{x,2+2\lambda}\bar{\sigma}^\lambda\right]^{(1+\max(\lambda,\epsilon))/\lambda}\log(\bar{p}),
\end{equation}
plugging these into Theorem \ref{thm:1}, we have that for all $t=1,2,\dots,T$ and $k=1,2,\dots,d$,
\begin{equation}
    \begin{split}
        \text{Err}^{(t)}
        & \leq (1-\eta_0\alpha_x\beta_x^{-1}\kappa^{-2}/2)^t\text{Err}^{(0)}+C\alpha_x^{-2}\bar{\sigma}^{-4d/(d+1)}\kappa^2\sum_{k=0}^d\|\bm{\Delta}_k^{(t)}\|_\text{F}^2\\
        & \leq \text{Err}^{(0)} + C\alpha_x^{-2}\bar{\sigma}^{-2d/(d+1)}\kappa^4\left(\prod_{k=1}r_k+\sum_{k=1}^dp_kr_k\right)\left[\frac{M_{\text{eff},1+\epsilon,\delta}^{1/\epsilon}\log(\bar{p})}{n}\right]^{\frac{2\epsilon}{1+\epsilon}}
    \end{split}
\end{equation}
and
\begin{equation}
    \begin{split}
        \|\cm{A}^{(t)}-\cm{A}^*\|_\text{F}^2 & \lesssim \kappa^2(1-C\alpha_x\beta_x^{-1}\kappa^{-2})^t\|\cm{A}^{(0)}-\cm{A}^*\|_\text{F}^2\\
        &+\kappa^4\alpha_x^{-2}\left(\sum_{k=1}^dp_kr_k+\prod_{k=1}^dr_k\right)\left[\frac{M_{\text{eff},1+\epsilon,\delta}^{1/\epsilon}\log(\bar{p})}{n}\right]^{2\epsilon/(1+\epsilon)}.
    \end{split}
\end{equation}
Finally, for all $t=1,2,\dots,T$ and $k=1,2,\dots,d$,
\begin{equation}
    \begin{split}
        \|\sin\Theta(\bm{U}_k^{(t)},\bm{U}_k^*)\|^2& \leq \bar{\sigma}^{-2/(d+1)}\text{Err}^{(t)}\\
        &\leq\bar{\sigma}^{\frac{-2}{d+1}}\text{Err}^{(0)}+C\kappa^4\alpha_x^{-2}\bar{\sigma}^{-2}d_\text{eff}\left[\frac{M_{\text{eff},1+\epsilon,\delta}^{1/\epsilon}\log(\bar{p})}{n}\right]^{\frac{2\epsilon}{1+\epsilon}}\leq\delta^2.
    \end{split}
\end{equation}

\end{proof}

\subsection{Proof of Theorem \ref{thm:logistic}}

\begin{proof}

The proof consists of six steps. In the first five steps, we prove the stability of the robust gradient estimators for the general $1\leq t\leq T$ and, hence, we omit the notation $(t)$ for simplicity. Specifically, in the first four steps, we give the upper bounds for $\|T_{k,1}\|_\text{F},\dots,\|T_{k,4}\|_\text{F}$, respectively, for $1\leq k\leq d$. In the fifth step, we extend the proof to the terms for the core tensor. In the last step, we apply the results to the local convergence analysis in Theorem \ref{thm:1} and verify the corresponding conditions.
Throughout the first five steps, we assume that for each $1\leq k\leq d$, $\|\bm{U}_k\|\asymp \bar{\sigma}^{1/(d+1)}$ and $\|\sin\Theta(\bm{U}_k,\bm{U}_k^*)\|\leq\delta$ and will verify them in the last step.\\

\noindent\textit{Step 1.} (Calculate local moments)

\noindent For any $1\leq k\leq d$, we let $r_k'=r_1r_2\cdots r_d/r_k$ and
\begin{equation}
    \nabla\overline{\mathcal{L}}(\cm{A}^*;z_i)_{(k)}\bm{V}_k = \left(\frac{\exp(\langle\cm{X}_i,\cm{A}^*\rangle)}{1+\exp(\langle\cm{X}_i,\cm{A}^*\rangle)}-y_i\right)(\cm{X}_i)_{(k)}(\otimes_{j=1,j\neq k}^d\bm{U}_j)\cm{S}_{(k)}^\top.
\end{equation}
Let $\bm{M}_k=(\otimes_{j=1,j\neq k}^{d}\bm{U}_j)/\|\otimes_{j=1,j\neq k}^{d}\bm{U}_j\|$ and $\bm{c}_l^\top(\cm{X}_i)_{(k)}\bm{M}_{k}=(w^{(i)}_{k,l,1},w^{(i)}_{k,l,2},\dots,w^{(i)}_{k,l,r_k'})$. By Assumption \ref{asmp:logistic}, $\mathbb{E}[|w_{k,l,j}^{(i)}|^{2}]\leq M_{x,2,\delta}$ for $l=1,2,\dots,p_k$ and $j=1,2,\dots,r_k'$. Let $\bm{N}_{k}=\cm{S}_{(k)}/\|\cm{S}_{(k)}\|$ and $\bm{c}_l^\top(\cm{X}_i)_{(k)}\bm{M}_k\bm{N}_{k}^\top=(z_{k,l,1}^{(i)},z_{k,l,2}^{(i)},\dots,z_{k,l,r_k}^{(i)})$. Then, $\mathbb{E}[|z^{(i)}_{k,l,m}|^{2}]\lesssim M_{x,2,\delta}$.
Also, denote $q_i(\cm{A})=\exp(\langle\cm{X}_i,\cm{A}\rangle)/[1+\exp(\langle\cm{X}_i,\cm{A}\rangle)]$ for any $\cm{A}$. Let $v_{k,l,m}^{(i)}=(q_i(\cm{A}^*)-y_i)z_{k,l,m}^{(i)}$, which satisfies that
\begin{equation}
    \mathbb{E}\left[|v_{k,l,j}^{(i)}|^2\right] = \mathbb{E}\left[\mathbb{E}\left[|q_i(\cm{A}^*)-y_i|^2|\cm{X}_i\right]\cdot|z_{k,l,j}^{(i)}|^2\right] \leq M_{x,2,\delta}.
\end{equation}

For any $1\leq k\leq d$, we let $\bm{V}_k=[\bm{v}_{k,1},\dots,\bm{v}_{k,r_k}]$. The $(l,m)$-th entry of $\nabla\overline{\mathcal{L}}(\cm{A}^*;z_i)_{(k)}\bm{V}_k - \nabla\overline{\mathcal{L}}(\cm{A};z_i)_{(k)}\bm{V}_k$ is $(q_i(\cm{A}^*)-q_i(\cm{A}))\bm{c}_l^\top(\cm{X}_i)_{(k)}\bm{v}_{k,m}$. Since $\exp(t)/(1+\exp(t))$ is a 1-Lipschitz function, we have $|q_i(\cm{A}^*)-q_i(\cm{A})|\leq |\langle\cm{X}_i,\cm{A}-\cm{A}^*\rangle|$. Let $\bm{w}_{k,m,l}=\bm{P}_k(\bm{v}_{k,m}\otimes\bm{c}_l)/\|\bm{P}_k(\bm{v}_{k,m}\otimes\bm{c}_l)\|_2$. Then, we have
\begin{equation}
    \begin{split}
        & \mathbb{E}[|s_{k,m,l}^{(i)}|^{1+\lambda}]:= \mathbb{E}\left[|\bm{w}_{k,m,l}^\top\text{vec}(\cm{X}_i)\text{vec}(\cm{X}_i)^\top\text{vec}(\cm{A}-\cm{A}^*)|^{1+\lambda}\right]\\
        & \leq \mathbb{E}\left[|\bm{w}_{k,m,l}^\top\text{vec}(\cm{X}_i)|^{2+2\lambda}\right]^{1/2}\cdot\mathbb{E}\left[\left|\text{vec}(\cm{X}_i)^\top\frac{\text{vec}(\cm{A}-\cm{A}^*)}{\|\text{vec}(\cm{A}-\cm{A}^*)\|_2}\right|^{2+2\lambda}\right]^{1/2}\cdot\|\cm{A}-\cm{A}^*\|_\text{F}^{1+\lambda}\\
        & \leq M_{x,2+2\lambda}\cdot\|\cm{A}-\cm{A}^*\|_\text{F}^{1+\lambda}.
    \end{split}
\end{equation}~

\noindent\textit{Step 2.} (Bound $\|T_{k,1}\|_\text{F}$)

\noindent We first bound the bias $\|T_{k,1}\|_\text{F}$.
Let $\tau_k=\tau/\|\otimes_{j=1,j\neq k}^d\bm{U}_j\|\asymp [nM_{x,2,\delta}/\log(\bar{p})]^{1/2}$. Then,
\begin{equation}
    \|T_{1,k}\|_\text{F}^2\asymp\bar{\sigma}^{2d/(d+1)}\sum_{l=1}^{p_k}\sum_{j=1}^{r_k}\left|\mathbb{E}\left[(q_i(\cm{A}^*)-y_i)z_{k,l,j}^{(i)}\right]-\mathbb{E}\left[\text{T}(q_i(\cm{A}^*)-y_i)z_{k,l,j}^{(i)},\tau_k)\right]\right|^2.
\end{equation}

By Holder's inequality and Markov's inequality,
\begin{equation}
    \begin{split}
        & \left|\mathbb{E}\left[(q_i(\cm{A}^*)-y_i)z_{k,l,j}^{(i)}\right]-\mathbb{E}\left[\text{T}((q_i(\cm{A}^*)-y_i)z_{k,l,j}^{(i)},\tau_k)\right]\right|\\
        \leq & \mathbb{E}\left[|(q_i(\cm{A}^*)-y_i)z_{k,l,j}^{(i)}|\cdot1\{|(q_i(\cm{A}^*)-y_i)z_{k,l,j}^{(i)}|\geq\tau_k\}\right]\\
        \leq & \mathbb{E}\left[|(q_i(\cm{A}^*)-y_i)z_{k,l,j}^{(i)}|^{2}\right]^{1/2}\cdot\mathbb{P}(|(q_i(\cm{A}^*)-y_i)z_{k,l,j}^{(i)}|\geq\tau_k)^{1/2}\\
        \leq & \mathbb{E}\left[|(q_i(\cm{A}^*)-y_i)z_{k,l,j}^{(i)}|^{2}\right]^{1/2}\cdot\left(\frac{\mathbb{E}\left[|(q_i(\cm{A}^*)-y_i)z_{k,l,j}^{(i)}|^{2}\right]}{\tau_k^{2}}\right)^{1/2}\\
        \asymp &~ M_{x,2,\delta}\cdot\tau_k^{-1}\asymp\left[\frac{M_{x,2,\delta}\log(\bar{p})}{n}\right]^{1/2}.
    \end{split}
\end{equation}
Hence, we have $\left\|T_{k,1}\right\|_\text{F}^2 \lesssim\bar{\sigma}^{2d/(d+1)}p_kr_kM_{x,2,\delta}\log(\bar{p})/n$.\\

\noindent\textit{Step 2}. (Bound $\|T_{k,2}\|_\text{F}$)

\noindent For $T_{k,2}$ in \eqref{eq:stable_robust_decomp}, it can be checked that
\begin{equation}
    \begin{split}
        & \mathbb{E}[\text{T}(v_{k,l,j}^{(i)},\tau_k)^2] \leq \mathbb{E}\left[|v_{k,l,j}^{(i)}|^{2}\right]\lesssim M_{x,2,\delta}.
    \end{split}
\end{equation}
Thus, $\text{var}(\text{T}(v_{k,l,j}^{(i)},\tau_k))\leq\mathbb{E}[\text{T}(v_{k,l,j}^{(i)},\tau_k)^2]\lesssim M_{x,2,\delta}$.
Also, for any $s=3,4,\dots$, the higher-order moments satisfy that
\begin{equation}
    \mathbb{E}\left[(\text{T}(v_{k,l,j}^{(i)},\tau_k)-\mathbb{E}[\text{T}(v_{k,l,j}^{(i)},\tau_k)])^s\right]\leq (2\tau_k)^{s-2}\cdot\mathbb{E}\left[(\text{T}((v_{k,l,j}^{(i)},\tau_k)-\mathbb{E}[\text{T}(v_{k,l,j}^{(i)},\tau_k)])^2\right].
\end{equation}

By Bernstein's inequality, for any $0<t<(2\tau_k)^{-1}M_{x,2,\delta}$,
\begin{equation}
    \begin{split}   
        &\mathbb{P}\left(\left|\frac{1}{n}\sum_{i=1}^n\text{T}(v_{k,l,j}^{(i)},\tau_k)-\mathbb{E}\text{T}(v_{k,l,j}^{(i)},\tau_k)\right|>t\right)\leq 2\exp\left(-\frac{nt^2}{4M_{x,2,\delta}}\right)
    \end{split}
\end{equation}
Let $t=CM_{x,2,\delta}^{1/2}\log(\bar{p})^{1/2}n^{-1/2}$. Therefore, we have
\begin{equation}
    \begin{split}
        &\mathbb{P}\Bigg(\left|\frac{1}{n}\sum_{i=1}^n\text{T}(v_{k,l,j}^{(i)},\tau_k)-\mathbb{E}\text{T}(v_{k,l,j}^{(i)},\tau_k)\right|>C\left[\frac{M_{x,2,\delta}\log(\bar{p})}{n}\right]^{\frac{1}{2}}\Bigg)
        \leq C\exp(-C\log(\bar{p}))
    \end{split}
\end{equation}
and
\begin{equation}
    \begin{split}
        &\mathbb{P}\Bigg(\max_{\substack{1\leq j\leq p_k\\1\leq l\leq r_k}}\left|\frac{1}{n}\sum_{i=1}^n\text{T}(v_{k,l,j}^{(i)},\tau_k)-\mathbb{E}\left[\text{T}(v_{k,l,j}^{(i)},\tau_k)\right]\right|>C\left[\frac{M_{x,2,\delta}\log(\bar{p})}{n}\right]^{\frac{1}{2}}\Bigg) \\
        &\leq Cp_kr_k\exp(-C\log(\bar{p}))\leq C\exp(-C\log(\bar{p})).
    \end{split}
\end{equation}
Therefore, with probability at least $1-C\exp(-C\log(\bar{p}))$,
\begin{equation}
    \left\|T_{k,2}\right\|_\text{F}^2 \lesssim \bar{\sigma}^{2d/(d+1)}p_kr_kM_{x,2,\delta}\log(\bar{p})/n.
\end{equation}~

\noindent\textit{Step 3.} (Bound $\|T_{k,3}\|_\text{F}$)

\noindent By definition, the $(l,m)$-th entry of $T_{k,3}$ can be bounded as
\begin{equation}
    |(T_{k,3})_{l,m}| \asymp \bar{\sigma}^{\frac{d}{d+1}}\cdot\left|\mathbb{E}[s_{k,m,l}^{(i)}] - \mathbb{E}\left[\text{T}(s_{k,m,l}^{(i)}+v_{k,m,l}^{(i)},\tau_k) - \text{T}(v_{k,m,l}^{(i)},\tau_k)\right]\right|.
\end{equation}
By the nature of truncation operator, moment condition, and Markov's inequality,
\begin{equation}
    \begin{split}
        & \left|\mathbb{E}[s_{k,m,l}^{(i)}] - \mathbb{E}\left[\text{T}(s_{k,m,l}^{(i)} + v_{k,l,m}^{(i)},\tau_k) - \text{T}(v_{k,l,m}^{(i)},\tau_k)\right]\right| \\
        & \leq \left|\mathbb{E}[s_{k,m,l}^{(i)}\cdot 1\{|(|v_{k,l,m}^{(i)}|\geq \tau_k)\cup(|s_{k,l,m}^{(i)} + v_{k,l,m}^{(i)}|\geq \tau_k)\}]\right| \\
        & \leq \left|\mathbb{E}[s_{k,m,l}^{(i)}\cdot 1\{|(|v_{k,l,m}^{(i)}|\geq \tau_k)\cup(|s_{k,l,m}^{(i)}|\geq \tau_k/2)\cup(|v_{k,l,m}^{(i)}|\geq \tau_k/2)\}]\right| \\
        & \leq \left|\mathbb{E}[s_{k,m,l}^{(i)}\cdot 1\{|s_{k,l,m}^{(i)}|\geq \tau_k/2\}]\right| + \left|\mathbb{E}[s_{k,m,l}^{(i)}\cdot 1\{|v_{k,l,m}^{(i)}|\geq \tau_k/2\}]\right|\\
        & \lesssim \mathbb{E}\left[|s_{k,m,l}^{(i)}|^{1+\lambda}\right]^{\frac{1}{1+\lambda}}\cdot\left(\frac{\mathbb{E}[|s_{k,l,m}^{(i)}|^{1+\lambda}]}{\tau_k^{1+\lambda}}\right)^{\frac{\lambda}{1+\lambda}} + \mathbb{E}\left[|s_{k,m,l}^{(i)}|^{1+\lambda}\right]^{\frac{1}{1+\lambda}}\cdot\left(\frac{\mathbb{E}[|v_{k,l,m}^{(i)}|^{2}]}{\tau_k^{2}}\right)^{\frac{\lambda}{1+\lambda}}\\
        & \lesssim \mathbb{E}\left[|s_{k,m,l}^{(i)}|^{1+\lambda}\right]\cdot\tau_k^{-\lambda} + \mathbb{E}\left[|s_{k,m,l}^{(i)}|^{1+\lambda}\right]^{\frac{1}{1+\lambda}}\cdot\mathbb{E}\left[|v_{k,m,l}^{(i)}|^{2}\right]^{2\lambda/(1+\lambda)}\cdot\tau_k^{-2\lambda/(1+\lambda)}\\
        & \lesssim \left\{\bar{\sigma}^\lambda M_{x,2+2\lambda}\left[\frac{\log(\bar{p})}{nM_{x,2,\delta}}\right]^{\frac{\lambda}{2}} + M_{x,2+2\lambda}^{1/(1+\lambda)}M_{x,2,\delta}^{\lambda/(1+\lambda)}\left[\frac{\log(\bar{p})}{nM_{x,2,\delta}}\right]^{\frac{\lambda}{1+\lambda}}\right\}\|\cm{A}-\cm{A}^*\|_\text{F}\\
        & \lesssim M_{x,2+2\lambda}^{1/(1+\lambda)}\left[\bar{\sigma}^\lambda M_{x,2+2\lambda}^{\lambda/(1+\lambda)}M_{x,2,\delta}^{-\lambda/2}\log(\bar{p})^{\frac{\lambda}{2}}n^{-\frac{\lambda}{2}} + \log(\bar{p})^{\frac{\lambda}{1+\lambda}}n^{-\frac{\lambda}{1+\lambda}}\right]\|\cm{A}-\cm{A}^*\|_\text{F}\\
        & \lesssim \bar{\sigma}^{\lambda}M_{x,2+2\lambda}\left[\log(\bar{p})/n\right]^{\frac{\lambda}{1+\lambda}}\|\cm{A}-\cm{A}^*\|_\text{F}.
    \end{split}
\end{equation}
Therefore, we have
\begin{equation}
    \|T_{k,3}\|_\text{F}^2 \lesssim \bar{\sigma}^{2d/(d+1)}\phi_{\lambda}\|\cm{A}-\cm{A}^*\|_\text{F}^2,
\end{equation}
where $\phi_\lambda = \bar{\sigma}^{2\lambda}M_{x,2+2\lambda}^2\bar{p}[\log(\bar{p})/n]^{2\lambda/(1+\lambda)}$.\\

\noindent\textit{Step 4.} (Bound $\|T_{k,4}\|_\text{F}$)

\noindent For $T_{k,4}$,
\begin{equation}
    \begin{split}
        \|T_{k,4}\|_\text{F}^2 \asymp \bar{\sigma}^{\frac{2d}{d+1}}\sum_{l,m} & \Bigg|\frac{1}{n}\sum_{i=1}^n\left[\text{T}(s_{k,m,l}^{(i)}+v_{k,m,l}^{(i)},\tau_k) - \text{T}(v_{k,m,l}^{(i)},\tau_k) \right] \\
        & -\mathbb{E}\left[\text{T}(s_{k,m,l}^{(i)}+v_{k,m,l}^{(i)},\tau_k) - \text{T}(v_{k,m,l}^{(i)},\tau_k) \right]\Bigg|^2.
    \end{split}
\end{equation}
For each $i=1,2,\dots,n$, we have $|\text{T}(s_{k,m,l}^{(i)}+v_{k,m,l}^{(i)},\tau_k) - \text{T}(v_{k,m,l}^{(i)},\tau_k)|\leq 2\tau_k$ and hence
\begin{equation}
    \begin{split}
        & \mathbb{E}[(\text{T}(s_{k,m,l}^{(i)}+v_{k,m,l}^{(i)},\tau_k) - \text{T}(v_{k,m,l}^{(i)},\tau_k))^2] \leq (2\tau_k)^{1-\lambda}\cdot\mathbb{E}[|s_{k,m,l}^{(i)}|^{1+\lambda}]\\
        & \asymp \tau_k^{1-\lambda}M_{x,2+2\lambda}\|\cm{A}-\cm{A}^*\|_\text{F}^{1+\lambda}. 
    \end{split}
\end{equation}
In addition, for any $q=3,4,\dots$, the higher-order moments satisfy that
\begin{equation}
    \begin{split}
        & \mathbb{E}[(\text{T}(s_{k,m,l}^{(i)}+v_{k,m,l}^{(i)},\tau_k) - \text{T}(v_{k,m,l}^{(i)},\tau_k))^q] \\
        & \leq (2\tau_k)^{q-2}\cdot\mathbb{E}[(\text{T}(s_{k,m,l}^{(i)}+v_{k,m,l}^{(i)},\tau_k) - \text{T}(v_{k,m,l}^{(i)},\tau_k))^2].
    \end{split}
\end{equation}
By Bernstein's inequality, for any $1\leq l\leq p_k$ and $1\leq m\leq r_k$,
\begin{equation}
    \begin{split}   
        & \mathbb{P}\Bigg(\Bigg| \frac{1}{n}\sum_{i=1}^n\left[\text{T}(s_{k,m,l}^{(i)},\tau_k + v_{k,m,l}^{(i)},\tau_k) - \text{T}(v_{k,m,l}^{(i)},\tau_k)\right]\\
        & - \mathbb{E}\left[\text{T}(s_{k,m,l}^{(i)},\tau_k + v_{k,m,l}^{(i)},\tau_k) - \text{T}(v_{k,m,l}^{(i)},\tau_k)\right] \Bigg|\geq t\Bigg)\\
        & \leq 2\exp\left(-\frac{Cnt^2}{\tau_k^{1-\lambda}M_{x,2+2\lambda}\|\cm{A}-\cm{A}^*\|_\text{F}^{1+\lambda} + \tau_k t}\right).
    \end{split}
\end{equation}
If $\|\cm{A}-\cm{A}^*\|_\text{F}\lesssim M_{x,2+2\lambda}^{-1/(1+\lambda)}\cdot M_{x,2,\delta}^{1/2}$, letting $t=C[M_{x,2,\delta}\log(\bar{p})/n]^{1/2}$,
\begin{equation}
    \begin{split}
        & \mathbb{P}\Bigg(\max_{m,l}\Bigg| \frac{1}{n}\sum_{i=1}^n\left[\text{T}(s_{k,m,l}^{(i)},\tau_k + v_{k,m,l}^{(i)},\tau_k) - \text{T}(v_{k,m,l}^{(i)},\tau_k)\right]\\
        & - \mathbb{E}\left[\text{T}(s_{k,m,l}^{(i)},\tau_k + v_{k,m,l}^{(i)},\tau_k) - \text{T}(v_{k,m,l}^{(i)},\tau_k)\right] \Bigg|\geq C\left[\frac{M_{x,2,\delta}\log(\bar{p})}{n}\right]^{1/2}\Bigg)\\
        & \lesssim p_kr_k\exp(-C\log(\bar{p})) \leq C\exp(-C\log(\bar{p})). 
    \end{split}
\end{equation}
If $\|\cm{A}-\cm{A}^*\|_\text{F}\gtrsim M_{x,2+2\lambda}^{-1/(1+\lambda)}\cdot M_{x,2,\delta}^{1/2}$, then
\begin{equation}
    \|\cm{A}-\cm{A}^*\|_\text{F}^{1+\lambda} \lesssim \|\cm{A}-\cm{A}^*\|_\text{F}^2\cdot M_{x,2+2\lambda}^{(1-\lambda)/(1+\lambda)}\cdot M_{x,2,\delta}^{(\lambda-1)/2},
\end{equation}
and letting $t=CM_{x,2+2\lambda}[\log(\bar{p})/n]^{\frac{\lambda}{1+\lambda}}\|\cm{A}-\cm{A}^*\|_\text{F}$,
\begin{equation}
    \begin{split}
        & \mathbb{P}\Bigg[\max_{1\leq m\leq p_k,1\leq l\leq r_k}\Bigg| \frac{1}{n}\sum_{i=1}^n\left[\text{T}(s_{k,m,l}^{(i)},\tau_k + v_{k,m,l}^{(i)},\tau_k) - \text{T}(v_{k,m,l}^{(i)},\tau_k)\right]\\
        & - \mathbb{E}\left[\text{T}(s_{k,m,l}^{(i)},\tau_k + v_{k,m,l}^{(i)},\tau_k) - \text{T}(v_{k,m,l}^{(i)},\tau_k)\right] \Bigg|\geq t\Bigg]\\
        & \lesssim p_kr_k\exp(-C\log(\bar{p})) \leq C\exp(-C\log(\bar{p})). 
    \end{split}
\end{equation}
Combining these two cases, we have
\begin{equation}
    \|T_{k,4}\|_\text{F}^2 \lesssim \bar{\sigma}^{\frac{2d}{d+1}}\phi_{\lambda}\|\cm{A}-\cm{A}^*\|_\text{F}^2 + \bar{\sigma}^{\frac{2d}{d+1}} p_kr_kM_{x,2,\delta}\log(\bar{p})/n.
\end{equation}

Based on the results in steps 2 to 5, we have
\begin{equation}
    \begin{split}   
        \sum_{j=1}^4\|T_{k,j}\|_\text{F}^2 \lesssim &~ \bar{\sigma}^{\frac{2d}{d+1}}p_kr_kM_{x,2,\delta}\log(\bar{p})/n + \bar{\sigma}^{\frac{2d}{d+1}}\phi_{\lambda}\|\cm{A}-\cm{A}^*\|_\text{F}^2.
    \end{split}
\end{equation}~

\noindent\textit{Step 6.} (Extension to core tensor)

In a similar fashion, we can show that with probability at least $1-C\exp(-C\log(\bar{p}))$,
\begin{equation}
    \begin{split}
        \|T_{0,1}\|_\text{F}&\lesssim\bar{\sigma}^{d/(d+1)}\sqrt{r_1r_2\cdots r_d}\left[\frac{M_{x,2,\delta}\log(\bar{p})}{n}\right]^{1/2},\\
        \|T_{0,2}\|_\text{F}&\lesssim\bar{\sigma}^{d/(d+1)}\sqrt{r_1r_2\cdots r_d}\left[\frac{M_{x,2,\delta}\log(\bar{p})}{n}\right]^{1/2},\\
        \|T_{0,3}\|_\text{F}&\lesssim C\phi_{\lambda}^{1/2}\bar{\sigma}^{d/(d+1)}\|\cm{A}-\cm{A}^*\|_\text{F},\\
        \|T_{0,4}\|_\text{F}&\lesssim C\phi_{\lambda}^{1/2}\bar{\sigma}^{d/(d+1)}\|\cm{A}-\cm{A}^*\|_\text{F} + \bar{\sigma}^{d/(d+1)}\sqrt{r_1r_2\cdots r_d}\left[\frac{M_{x,2,\delta}\log(\bar{p})}{n}\right]^{1/2}.
    \end{split}
\end{equation}
Hence, with probability at least $1-C\exp(-C\log(\bar{p}))$,
\begin{equation}
    \|\cm{G}_0-\mathbb{E}[\nabla_0\mathcal{L}]\|_\text{F}^2\lesssim \bar{\sigma}^{\frac{2d}{d+1}}\phi_{\lambda}\|\cm{A}-\cm{A}^*\|_\text{F}^2+\bar{\sigma}^{\frac{2d}{d+1}}\prod_{k=1}^dr_k\left[\frac{M_{x,2,\delta}\log(\bar{p})}{n}\right].
\end{equation}~

\noindent\textit{Step 7.} (Verify the conditions and conclude the proof)

In the last step, we apply the results above to Theorem \ref{thm:1}. First, we examine the conditions in Theorem \ref{thm:1} hold. Under Assumption \ref{asmp:1}, by Lemma 3.11 in \citet{bubeck2015convex}, we can show that the RCG condition in Definition \ref{def:2} is implied by the restricted strong convexity and strong smoothness with $\alpha=\alpha_x$ and $\beta=\beta_x$.

Next, we show the stability of the robust gradient estimators for all $t=1,2,\dots,T$.
By matrix perturbation theory, if $\|\cm{A}^{(0)}-\cm{A}^*\|_\text{F}\leq\sqrt{\alpha_x/\beta_x}\underline{\sigma}\kappa^{-2}\delta$, we have $\|\sin\Theta(\bm{U}_k^{(0)},\bm{U}_k^*)\|\leq\delta$ for all $k=1,\dots,d$. After a finite number of iterations, $C_T$, with probability at least $1-C_T\exp(-C\log(\bar{p}))$, we can have $\|\sin\Theta(\bm{U}_k^{(C_T)},\bm{U}_k^*)\|\leq\delta'<(4\sqrt{2})^{-1}$.

For any $l\neq k$ and any tensor $\cm{B}\in\mathbb{R}^{p_1\times\cdots\times p_d}$, $(\cm{B}\times_{j\neq k}\bm{U}_j^\top)_{(l)}=\bm{U}_l^\top\cm{B}_{(l)}(\otimes_{j\neq l}\bm{U}_j')$, where $\bm{U}_j'=\bm{U}_j$ for $j\neq k$ and $\bm{U}_k'=\bm{I}_{r_k}$.
For any $\bm{U}_l\in\mathcal{C}(\bm{U}_l^*,\delta')$, we have $\|\bm{U}_l-\bm{U}^*_l\bm{O}_l\|\leq\sqrt{2}\|\sin\Theta(\bm{U}_l,\bm{U}_l^*)\|\leq\sqrt{2}\delta'$ for some $\bm{O}_l\in\mathbb{O}^{r_k\times r_k}$. Let $\bm{\Delta}_l=\bm{U}_l-\bm{U}^*_l\bm{O}_l$ and decompose $\bm{\Delta}_l=\bm{\Delta}_{l,1}+\bm{\Delta}_{l,2}$ where $\langle\bm{\Delta}_{l,1},\bm{\Delta}_{l,2}\rangle=0$ and $\bm{\Delta}_{l,1}/\|\bm{\Delta}_{l,1}\|,\bm{\Delta}_{l,2}/\|\bm{\Delta}_{l,2}\|\in\mathcal{C}(\bm{U}_l^*,\delta')$. Thus, we have $\|\bm{\Delta}_{l,1}\|\leq\sqrt{2}\delta'$ and $\|\bm{\Delta}_{l,2}\|\leq\sqrt{2}\delta'$.

Denote $\xi=\sup_{\bm{U}_l\in\mathcal{C}(\bm{U}_l^*,\delta')}\|\bm{U}_l^\top\cm{B}_{(l)}(\otimes_{j\neq l}\bm{U}_j')\|_\text{F}$. Then, since
\begin{equation}
    \begin{split}
        & \|\bm{U}_l^\top\cm{B}_{(l)}(\otimes_{j\neq l}\bm{U}_j')\|_\text{F}\\
        & \leq \|(\bm{U}^*_l\bm{O}_l)^\top\cm{B}_{(l)}(\otimes_{j\neq l}\bm{U}_j')\|_\text{F} + \|\bm{\Delta}_l^\top\cm{B}_{(l)}(\otimes_{j\neq l}\bm{U}_j')\|_\text{F}\\
        & \leq \|(\bm{U}^*_l\bm{O}_l)^\top\cm{B}_{(l)}(\otimes_{j\neq l}\bm{U}_j')\|_\text{F} + \|\bm{\Delta}_{l,1}\|\cdot\|(\bm{\Delta}_{l,1}/\|\bm{\Delta}_{l,1}\|)^\top\nabla\cm{B}_{(l)}(\otimes_{j\neq l}\bm{U}_j')\|_\text{F}\\
        & + \|\bm{\Delta}_{l,2}\|\cdot\|(\bm{\Delta}_{l,2}/\|\bm{\Delta}_{l,1}\|)^\top\nabla\cm{B}_{(l)}(\otimes_{j\neq l}\bm{U}_j')\|_\text{F},
    \end{split}
\end{equation}
we have that
\begin{equation}
    \xi \leq \|(\bm{U}^*_l\bm{O}_l)^\top\nabla\cm{B}_{(l)}(\otimes_{j\neq l}\bm{U}_j')\|_\text{F} + (\|\bm{\Delta}_{l,1}\| + \|\bm{\Delta}_{l,2}\|)\xi,
\end{equation}
that is, taking $\delta'=1/8$,
\begin{equation}
    \xi \leq (1-2\sqrt{2}\delta')^{-1}\|(\bm{U}^*_l\bm{O}_l)^\top\nabla\cm{B}_{(l)}(\otimes_{j\neq l}\bm{U}_j')\|_\text{F}\leq2\|(\bm{U}^*_l\bm{O}_l)^\top\nabla\cm{B}_{(l)}(\otimes_{j\neq l}\bm{U}_j')\|_\text{F}.
\end{equation}
Hence, for the iterate $t=1,2,\dots,T$, combining the results in steps 1 to 6, we have that with probability at least $1-C\exp(-C\log(\bar{p}))$, for any $k=1,2,\dots,d$
\begin{equation}
    \|\bm{G}_k^{(t)}-\mathbb{E}[\nabla_k\mathcal{L}^{(t)}]\|_\text{F}^2 \lesssim \phi_{\lambda}\bar{\sigma}^{2d/(d+1)}\|\cm{A}^{(t)}-\cm{A}^*\|_\text{F}^2 + \bar{\sigma}^{2d/(d+1)}(p_kr_k)\left[\frac{M_{x,2,\delta}\log(\bar{p})}{n}\right]
\end{equation}
and
\begin{equation}
    \|\cm{G}_0^{(t)}-\mathbb{E}[\nabla_0\mathcal{L}^{(t)}]\|_\text{F}^2\lesssim \phi_{\lambda}\bar{\sigma}^{2d/(d+1)}\|\cm{A}^{(t)}-\cm{A}^*\|_\text{F}^2 + \bar{\sigma}^{2d/(d+1)}\prod_{k=1}^dr_k\left[\frac{M_{x,2,\delta}\log(\bar{p})}{n}\right].
\end{equation}
As the sample size satisfies
\begin{equation}
    n\gtrsim \bar{p}^{1/\lambda}\alpha_x^{-2/\lambda}\kappa^{4/\lambda}M_{x,2+2\lambda}^{2/\lambda}\bar{\sigma}^{2}\log(\bar{p}),
\end{equation}
plugging these into Theorem \ref{thm:1}, we have that for all $t=1,2,\dots,T$ and $k=1,2,\dots,d$,
\begin{equation}
    \begin{split}
        \text{Err}^{(t)}
        & \leq (1-\eta_0\alpha_x\beta_x^{-1}\kappa^{-2}/2)^t\text{Err}^{(0)}+C\alpha_x^{-2}\bar{\sigma}^{-4d/(d+1)}\kappa^2\sum_{k=0}^d\|\bm{\Delta}_k^{(t)}\|_\text{F}^2\\
        & \leq \text{Err}^{(0)} + C\alpha_x^{-2}\bar{\sigma}^{-2d/(d+1)}\kappa^4\left(\prod_{k=1}r_k+\sum_{k=1}^dp_kr_k\right)\left[\frac{M_{x,2,\delta}\log(\bar{p})}{n}\right]
    \end{split}
\end{equation}
and
\begin{equation}
    \begin{split}
        \|\cm{A}^{(t)}-\cm{A}^*\|_\text{F}^2 & \lesssim \kappa^2(1-C\alpha_x\beta_x^{-1}\kappa^{-2})^t\|\cm{A}^{(0)}-\cm{A}^*\|_\text{F}^2\\
        &+\kappa^4\alpha_x^{-2}\left(\sum_{k=1}^dp_kr_k+\prod_{k=1}^dr_k\right)\left[\frac{M_{x,2,\delta}\log(\bar{p})}{n}\right].
    \end{split}
\end{equation}
Finally, for all $t=1,2,\dots,T$ and $k=1,2,\dots,d$,
\begin{equation}
    \begin{split}
        \|\sin\Theta(\bm{U}_k^{(t)},\bm{U}_k^*)\|^2& \leq \bar{\sigma}^{-2/(d+1)}\text{Err}^{(t)}\\
        &\leq\bar{\sigma}^{\frac{-2}{d+1}}\text{Err}^{(0)}+C\kappa^4\alpha_x^{-2}\bar{\sigma}^{-2}d_\text{eff}\left[\frac{M_{x,2,\delta}\log(\bar{p})}{n}\right]\leq\delta^2.
    \end{split}
\end{equation}

\end{proof}

\subsection{Proof of Theorem \ref{thm:PCA}}

The proof consists of six steps. In the first five steps, we prove the stability of the robust gradient estimators for the general $1\leq t\leq T$ and, hence, we omit the notation $(t)$ for simplicity. Specifically, in the first four steps, we give the upper bounds for $\|T_{k,1}\|_\text{F},\dots,\|T_{k,4}\|_\text{F}$, respectively, for $1\leq k\leq d$. In the fifth step, we extend the proof to the terms for the core tensor. In the last step, we apply the results to the local convergence analysis in Theorem \ref{thm:1} and verify the corresponding conditions.
Throughout the first five steps, we assume that for each $1\leq k\leq d$, $\|\bm{U}_k\|\asymp \bar{\sigma}^{1/(d+1)}$ and $\|\sin\Theta(\bm{U}_k,\bm{U}_k^*)\|\leq\delta$ and will verify them in the last step.\\

\noindent\textit{Step 1.} (Calculate local moments)

\noindent For any $1\leq k\leq d$,
\begin{equation}
    \nabla\overline{\mathcal{L}}(\cm{A}^*;z)_{(k)}\bm{V}_k = (-\cm{E})_{(k)}(\otimes_{j\neq k}\bm{U}_j)\cm{S}_{(k)}^\top
\end{equation}
and
\begin{equation}
    \nabla\overline{\mathcal{L}}(\cm{A};z)_{(k)}\bm{V}_k - \nabla\overline{\mathcal{L}}(\cm{A}^*;z)_{(k)}\bm{V}_k = (\cm{A} - \cm{A}^*)_{(k)}\bm{V}_k.
\end{equation}

Let $\bm{M}_{k,1}=(\otimes_{j=1,j\neq k}^d\bm{U}_j)/\|\otimes_{j=1,j\neq k}^d\bm{U}_j\|$ and its columns as $\bm{M}_{k,1}=[\bm{m}_{k,1},\bm{m}_{k,2},\dots,\bm{m}_{k,r_k'}]$. Let $z_{k,j,l}=-\bm{c}_j^\top\cm{E}_{(k)}\bm{m}_{k,l}$, and by Assumption \ref{asmp:logistic}, $\mathbb{E}[|z_{k,j,l}|^{1+\epsilon}]\leq M_{e,1+\epsilon,\delta}$. Let $\bm{M}_{k,2}=\cm{S}_{(k)}^\top/\|\cm{S}_{(k)}\|$ and, hence, we have $-\bm{c}_j^\top\cm{E}_{(k)}\bm{M}_{k,1}\bm{M}_{k,2}\bm{c}_m=w_{k,j,m}$, for $1\leq j\leq p_k$ and $1\leq m\leq r_k$, satisfying $\mathbb{E}[|w_{k,j,m}|^{1+\epsilon}]\lesssim M_{e,1+\epsilon,\delta}$. In addition, let $s_{k,m,l}$ be the $(j,m)$-th entry of $(\nabla\overline{\mathcal{L}}(\cm{A};z)_{(k)}\bm{V}_k - \nabla\overline{\mathcal{L}}(\cm{A}^*;z)_{(k)})\bm{M}_{k,1}\bm{M}_{k,2}$. Then, we have $\mathbb{E}[|s_{k,j,m}|^2]\lesssim\|\cm{A}-\cm{A}^*\|_\text{F}^2$.\\

\noindent\textit{Step 2.} (Bound $\|T_{k,j}\|_\text{F}$)

\noindent We first bound the bias, for any $1\leq j\leq p_k$ and $1\leq m\leq r_k$,
\begin{equation}
    \begin{split}
        & |\mathbb{E}[w_{k,j,m}]-\mathbb{E}[\text{T}(w_{k,j,m},\tau_k)]| \leq \mathbb{E}[|w_{k,j,m}|\cdot1\{|w_{k,j,m}|\geq\tau_k\}]\\
        \leq & \mathbb{E}\left[|w_{k,j,m}|^{1+\epsilon}\right]^{1/(1+\epsilon)}\cdot\mathbb{P}(|w_{k,j,m}|\geq\tau_k)^{\epsilon/(1+\epsilon)}\\
        \leq & \mathbb{E}\left[|w_{k,j,m}|^{1+\epsilon}\right]^{1/(1+\epsilon)}\cdot\left(\frac{\mathbb{E}[|w_{k,j,m}|^{1+\epsilon}]}{\tau_k^{1+\epsilon}}\right)^{\epsilon/(1+\epsilon)}\\
        \lesssim & M_{e,1+\epsilon,\delta}\cdot\tau_k^{-\epsilon} \asymp M_{e,1+\epsilon,\delta}^{1/(1+\epsilon)}\cdot\kappa^{-2}
    \end{split}
\end{equation}
with the truncation parameter $\tau_k\asymp M_{e,1+\epsilon,\delta}^{1/(1+\epsilon)}\cdot\kappa^{2/\epsilon}$. Hence,
\begin{equation}
    \|T_{k,1}\|_\text{F}\lesssim\bar{\sigma}^{d/(d+1)}\sqrt{p_kr_k}M_{e,1+\epsilon,\delta}^{1/(1+\epsilon)}\kappa^{-2}.
\end{equation}

Note that
\begin{equation}
    \mathbb{E}\left[\text{T}(w_{k,j,m},\tau_k)^2\right]\leq \tau_k^{1-\epsilon}\cdot\mathbb{E}[|w_{k,j,m}|^{1+\epsilon}]\asymp\tau_k^{1-\epsilon}\cdot M_{e,1+\epsilon,\delta}.
\end{equation}
Thus, we have $\text{var}(\text{T}(w_{k,j,m},\tau_k))\leq\mathbb{E}[\text{T}(w_{k,j,m},\tau_k)^2]\leq \tau_k^{1-\epsilon}M_{e,1+\epsilon,\delta}$. Also, for any $s=3,4,\dots$, the higher-order moments satisfy that
\begin{equation}
    \mathbb{E}\left[(\text{T}(w_{k,j,m},\tau_k)-\mathbb{E}[\text{T}(w_{k,j,m},\tau_k)])^s\right]
    \leq (2\tau_k)^{s-2}\mathbb{E}\left[(\text{T}(w_{k,j,m},\tau_k)-\mathbb{E}[\text{T}(w_{k,j,m},\tau_k)])^2\right].
\end{equation}

By Bernstein's inequality, for any $0<t\leq\tau_k^{-\epsilon}M_{e,1+\epsilon,\delta}$,
\begin{equation}
    \mathbb{P}(|\text{T}(w_{k,j,m},\tau_k)-\mathbb{E}[\text{T}(w_{k,j,m},\tau_k)]|\geq t)\leq 2\exp\left(-\frac{t^2}{4\tau_k^{1-\epsilon}M_{e,1+\epsilon,\delta}}\right).
\end{equation}
Letting $t=CM_{e,1+\epsilon,\delta}^{1/(1+\epsilon)}\kappa^{-2}$, since $\underline{\sigma}/M_{e,1+\epsilon,\delta}^{1/(1+\epsilon)}\gtrsim\sqrt{\bar{p}}$ we have
\begin{equation}
    \mathbb{P}(|\text{T}(w_{k,j,m},\tau_k)-\mathbb{E}[\text{T}(w_{k,j,m},\tau_k)]|\geq CM_{e,1+\epsilon,\delta}^{1/(1+\epsilon)}\kappa^{-2})\leq C\exp\left(-C\log(\bar{p})\right)
\end{equation}
and
\begin{equation}
    \mathbb{P}\left(\max_{\substack{1\leq j\leq p_k\\1\leq m\leq r_k}}|T(w_{k,j,m},\tau_k)-\mathbb{E}[T(w_{k,j,m},\tau_k)]|\geq CM_{e,1+\epsilon,\delta}^{1/(1+\epsilon)}\kappa^{-2}\right)\leq C\exp\left(-C\log(\bar{p})\right).
\end{equation}

Hence, with probabilty at least $1-C\exp(-C\log(\bar{p}))$,
\begin{equation}
    \|T_{k,2}\|_\text{F}\lesssim \bar{\sigma}^{d/(d+1)}\sqrt{p_kr_k}M_{e,1+\epsilon,\delta}^{1/(1+\epsilon)}\kappa^{-2}.
\end{equation}

By definition, the $(j,m)$-th entry of $T_{k,3}$ can be bounded as
\begin{equation}
    |T_{k,3}| \asymp \bar{\sigma}^{d/(d+1)}\cdot\left|\mathbb{E}[s_{k,j,m}]-\mathbb{E}[\text{T}(s_{k,j,m}+w_{k,j,m},\tau_k) - \text{T}(w_{k,j,m},\tau_k)]\right|.
\end{equation}
Then, by Markov's inequality,
\begin{equation}
    \begin{split}
        & \left|\mathbb{E}[s_{k,j,m}] - \mathbb{E}\left[\text{T}(s_{k,j,m} + w_{k,j,m}^{(i)},\tau_k) - \text{T}(w_{k,j,m},\tau_k)\right]\right| \\
        & \leq \left|\mathbb{E}[s_{k,j,m}^{(i)}\cdot 1\{|(|w_{k,j,m}|\geq \tau_k)\cup(|s_{k,j,m} + w_{k,j,m}|\geq \tau_k)\}]\right| \\
        & \leq \left|\mathbb{E}[s_{k,j,m}\cdot 1\{|(|w_{k,j,m}|\geq \tau_k)\cup(|s_{k,j,m}|\geq \tau_k/2)\cup(|w_{k,j,m}|\geq \tau_k/2)\}]\right| \\
        & \leq \left|\mathbb{E}[s_{k,j,m}\cdot 1\{|s_{k,j,m}|\geq \tau_k/2\}]\right| + \left|\mathbb{E}[s_{k,j,m}\cdot 1\{|v_{k,l,m}|\geq \tau_k/2\}]\right|\\
        & \lesssim \mathbb{E}\left[|s_{k,j,m}|^{2}\right]^{\frac{1}{2}}\cdot\left(\frac{\mathbb{E}[|s_{k,l,m}|^{2}]}{\tau_k^{2}}\right)^{\frac{1}{2}} + \mathbb{E}\left[|s_{k,j,m}|^{2}\right]^{\frac{1}{2}}\cdot\left(\frac{\mathbb{E}[|v_{k,l,m}|^{2}]}{\tau_k^{2}}\right)^{\frac{1}{2}}\\
        & \lesssim \mathbb{E}\left[|s_{k,m,l}|^{2}\right]\cdot\tau_k^{-1} + \mathbb{E}\left[|s_{k,m,l}|^{2}\right]^{\frac{1}{2}}\cdot\mathbb{E}\left[|v_{k,m,l}|^{2}\right]\cdot\tau_k^{-1}.
    \end{split}
\end{equation}
In addition, by Bernstein's inequality, we have
\begin{equation}
    \|T_{k,3}\|_\text{F}^2 + \|T_{k,4}\|_\text{F}^2 \lesssim \bar{\sigma}^{2d/(d+1)}\|\cm{A}-\cm{A}^*\|_\text{F}^2.
\end{equation}
Furthermore, the bounds for the core tensor can be developed similarly.\\

\noindent\textit{Step 3.} (Verify the conditions and conclude the proof)

\noindent First, as $\nabla\overline{\mathcal{L}}(\cm{A};\cm{Y})=\cm{A}-\cm{Y}$, the RCG condition holds with $\alpha=2$ and $\beta=2$:
\begin{equation}
    \langle\mathbb{E}[\nabla\overline{\mathcal{L}}(\cm{A};\cm{Y})],\cm{A}-\cm{A}^*\rangle = \|\cm{A}-\cm{A}^*\|_\text{F}^2.
\end{equation}

Plugging the results above in Theorem \ref{thm:1}, by the same finite covering arguments in the proof of Theorems \ref{thm:linearregression} and \ref{thm:logistic}, we can obtain the results and complete the proof.

\newpage

\section{Initialization and Implementation}
\label{append:init}

\subsection{Heavy-tailed Tensor Linear Regression}

Consider the tensor linear model:
\[
\cm{Y}_i = \langle \cm{A}^*, \cm{X}_i \rangle + \cm{E}_i,
\]
where:
\begin{itemize}
  \item \(\cm{A}^* \in \mathbb{R}^{p_1 \times \cdots \times p_d}\) is the unknown tensor of interest,
  \item \(\cm{X}_i \in \mathbb{R}^{p_1 \times \cdots \times p_{d_0}}\) is the covariate, potentially following a heavy-tailed distribution,
  \item \(\cm{Y}_i \in \mathbb{R}^{p_{d_0+1} \times \cdots \times p_{d}}\) is the response, and
  \item \(\cm{E}_i \in \mathbb{R}^{p_{d_0+1} \times \cdots \times p_{d}}\) is the heavy-tailed noise term.
\end{itemize}
According to Theorem \ref{thm:1}, our goal is to find the initial value $\bm{F}^{(0)}=(\cm{G}^{(0)},\bm{U}_1^{(0)},\dots,\bm{U}_d^{(0)})$ satisfying $\textup{Err}(\bm{F}^{(0)})\leq C\alpha_x\beta_x^{-1}\bar{\sigma}^{2/(d+1)}\kappa^{-2}$. To achieve that, we first transform the tensor linear model to
\begin{equation}
  \bm{y}_i = \bm{A}^*\bm{x}_i + \bm{e}_i,
\end{equation}
where $\bm{y}_i=\text{vec}(\cm{Y}_i)\in\mathbb{R}^{p_y}$, $\bm{A}^*=\text{mat}(\cm{A}^*)\in\mathbb{R}^{p_y\times p_x}$, $\bm{x}_i=\text{vec}(\cm{X}_i)\in\mathbb{R}^{p_x}$, and $\bm{e}_i=\text{vec}(\cm{E}_i)\in\mathbb{R}^{p_y}$.

Specifically, we have the following initialization procedure: 
\begin{itemize}
  \item[1.] We apply the vector truncation to $\bm{x}_i$
  \begin{equation}
    \widetilde{\bm{x}}_i(\omega) = \frac{\min(\|\bm{x}_i\|_2,\omega)}{\|\bm{x}_i\|_2}\bm{x}_i.
  \end{equation}
  \item[2.] We use the nuclear norm regularized Huber regression \citep{tan2023sparse}
\begin{equation}
  \widetilde{\bm{A}}(\omega,\delta,\lambda_{\text{nuc}}) = \arg\min \frac{1}{n}\sum_{i=1}^n\rho_\delta(\bm{y}_i-\bm{A}\widetilde{\bm{x}}_i(\omega)) + \lambda_{\text{nuc}}\|\bm{A}\|_{\text{nuc}}.
\end{equation}
  \item[3.] We apply higher-order orthogonal iteration to $\widetilde{\bm{A}}(\omega,\delta,\lambda)$ to obtain $(\widetilde{\cm{G}},\widetilde{\bm{U}}_1,\dots,\widetilde{\bm{U}}_d)$. Finally, we set $\bm{F}^{(0)}=(b^{-d}\widetilde{\cm{G}},b\bm{U}_1,\dots,b\bm{U}_1)$, where $b$ is the regularization parameter in Algorithm \ref{alg:1}.
\end{itemize}

Next, we state the theoretical properties of the proposed initial estimator. 
\begin{proposition}\label{prop:init1}
  Suppose the tuning parameters satisfy $\omega\asymp(p_x^\lambda M_{x,2+2\lambda} n)^{1/(2+2\lambda)}$, $\delta\asymp(nM_{e,1+\epsilon}/\log(p_xp_y))^{1/(1+\epsilon)}$, $\lambda_{\textup{nuc}}\asymp M_{e,1+\epsilon}^{1/(1+\epsilon)}(\log(p_xp_y)/n)^{\epsilon/(1+\epsilon)}$. If the sample size $n$ satisfies
  \begin{equation}
    n\gtrsim (p_x+p_y)^{(1+\epsilon)/(2\epsilon)}\log(p_xp_y)\alpha_x^{(2+2\epsilon)/\epsilon}\beta_x^{-(1+\epsilon)/\epsilon} + \alpha_x^{(1+\lambda)/\lambda}p_xM_{x,2+2\lambda}^{1/\lambda},
  \end{equation}
  we have $\textup{Err}(\bm{F}^{(0)})\lesssim \alpha_x\beta_x^{-1}\bar{\sigma}^{2/(d+1)}\kappa^{-2}$.
\end{proposition}

\begin{proof}
  By Proposition 3.2 in \citet{wang2023rate}, when $\omega\asymp(p_x^\lambda M_{x,2+2\lambda} n)^{1/(2+2\epsilon)}$, with high probability,
  \begin{equation}
    \left\|\frac{1}{n}\sum_{i=1}^n\widetilde{\bm{x}}_i(\omega)\widetilde{\bm{x}}_i(\omega)^\top - \mathbb{E}[\bm{x}_i\bm{x}_i^\top] \right\| \lesssim \left(\frac{p_xM_{x,2+2\lambda}^{1/\lambda}}{n}\right)^{\lambda/(1+\lambda)}.
  \end{equation}
  Therefore, when $n\gtrsim \alpha_x^{(1+\lambda)/\lambda}p_xM_{x,2+2\lambda}^{1/\lambda}$, the above error is bounded by $C\alpha_x$ for some sufficiently small $C$. In other words, by this vector norm truncation with parameter $\omega$, the heavy-tailed covariates are well controlled.

  Next, by Theorem 1 in \citet{tan2023sparse}, with $\delta\asymp(nM_{e,1+\epsilon}/\log(p_xp_y))^{1/(1+\epsilon)}$ and $\lambda_{\text{nuc}}\asymp M_{e,1+\epsilon}^{1/(1+\epsilon)}(\log(p_xp_y)/n)^{\epsilon/(1+\epsilon)}$,
  \begin{equation}
    \|\widetilde{\bm{A}}(\omega,\delta,\lambda_{\text{nuc}})-\bm{A}^*\|_\text{F}\lesssim \alpha_x^{-1}M_{e,1+\epsilon}^{1/(1+\epsilon)}\sqrt{r(p_x+p_y)}(\log(p_xp_y)/n)^{\epsilon/(1+\epsilon)}.
  \end{equation}
  Finally, by perturbation bound in \citet{ZX18}, when the sample size satisfies $n\gtrsim (p_x+p_y)^{(1+\epsilon)/(2\epsilon)}\log(p_xp_y)\alpha_x^{(2+2\epsilon)/\epsilon}\beta_x^{-(1+\epsilon)/\epsilon}$, the intial bound is satisfied.

\end{proof}

\subsection{Heavy-tailed Tensor Logistic Regression}

Consider the tensor logistic regression with negative log-likelihood loss function
\begin{equation}
  \overline{\mathcal{L}}(\cm{A};z_i) = \log(1+\exp(\langle\cm{X}_i,\cm{A}\rangle)) - y_i\langle\cm{X}_i,\cm{A}\rangle.
\end{equation}
According to Theorem \ref{thm:1}, our goal is to find $\bm{F}^{(0)}$ such that $\text{Err}(\bm{F}^{(0)})\leq C\alpha_x\beta_x^{-1}\bar{\sigma}^{2/(d+1)}\kappa^{-2}$. To initialize it, we transform $\cm{X}_i$ and $\cm{A}$ to matrices $\bm{X}_i$ and $\bm{A}$. Note that such transformation is not unique, but we try to make the difference between the numbers of rows and columns to be as small as possible. Denote the dimension of $\bm{A}$ as $q_1\times q_2$.

Specifically, we have the following initialization procedure: 
\begin{itemize}
  \item[1.] Similarly to \citep{zhu2021taming}, we apply the vector truncation to $\bm{x}_i$
  \begin{equation}
    \widetilde{\bm{X}}_i(\omega) = \frac{\min(\|\bm{X}_i\|_\text{F},\omega)}{\|\bm{X}_i\|_\text{F}}\bm{x}_i.
  \end{equation}
  \item[2.] We use the nuclear norm regularized logistic regression
\begin{equation}
  \widetilde{\bm{A}}(\omega,\delta,\lambda_{\text{nuc}}) = \arg\min \frac{1}{n}\sum_{i=1}^n[\log(1+\exp(\langle\widetilde{\bm{X}}_i(\omega),\bm{A}\rangle)) - y_i\langle\widetilde{\bm{X}}_i(\omega),\bm{A}\rangle] + \lambda_{\text{nuc}}\|\bm{A}\|_{\text{nuc}}.
\end{equation}
  \item[3.] We apply higher-order orthogonal iteration to $\widetilde{\bm{A}}(\omega,\delta,\lambda_{\text{nuc}})$ to obtain $(\widetilde{\cm{G}},\widetilde{\bm{U}}_1,\dots,\widetilde{\bm{U}}_d)$. Finally, we set $\bm{F}^{(0)}=(b^{-d}\widetilde{\cm{G}},b\bm{U}_1,\dots,b\bm{U}_1)$, where $b$ is the regularization parameter in Algorithm \ref{alg:1}.
\end{itemize}

\begin{proposition}\label{prop:init2}
  Suppose the tuning parameters satisfy $\omega\asymp(p_x^\lambda M_{x,2+2\lambda} n)^{1/(2+2\lambda)}$, $\lambda_{\textup{nuc}}\asymp (\log(q_1q_2)/n)^{1/2}$. If the sample size $n$ satisfies
  \begin{equation}
    n\gtrsim (q_1+q_2)\log(q_1q_2)\alpha_x^4\beta_x^{-2} + \alpha_x^{(1+\lambda)/\lambda}(q_1+q_2)M_{x,2+2\lambda}^{1/\lambda},
  \end{equation}
  we have $\textup{Err}(\bm{F}^{(0)})\lesssim \alpha_x\beta_x^{-1}\bar{\sigma}^{2/(d+1)}\kappa^{-2}$.
\end{proposition}

The proof is similar to that of Proposition \ref{prop:init1}, and hence is omitted for brevity.

\subsection{Heavy-tailed Tensor PCA}

For the tensor PCA
\begin{equation}
  \cm{Y} = \cm{A}^* + \cm{E},
\end{equation}
according to Theorem \ref{thm:1}, our goal is to find $\bm{F}^{(0)}$ such that $\text{Err}(\bm{F}^{(0)})\leq C\alpha_x\beta_x^{-1}\bar{\sigma}^{2/(d+1)}\kappa^{-2}$. To initialize it, we consider the pseudo-Huber tensor decomposition. Please refer to \citet{shen2023quantile} for detailed implementation.

\newpage

\section{Additional Simulation Experiments}\label{append:numerical}

\subsection{Tensor Logistic Regression}

We consider the tensor logistic regression model.
\begin{equation}\label{eq:logistic_reg}
  \begin{split}  
    \text{Model III}:\quad & \mathbb{P}(y_i=1|\cm{X}_i)=\frac{\exp(\langle\cm{X}_i,\cm{A}^*\rangle)}{1+\exp(\langle\cm{X}_i,\cm{A}^*\rangle)},\\ 
    & \mathbb{P}(y_i=0|\cm{X}_i)=\frac{1}{1+\exp(\langle\cm{X}_i,\cm{A}^*\rangle)},
  \end{split}
\end{equation}
where $\cm{X}_i\in\mathbb{R}^{10\times10\times10}$ and $\cm{A}^*=\sqrt{10}\cdot\bm{1}_{10}\circ\bm{1}_{10}\circ\bm{1}_{10}=\cm{S}^*\times_{j=1}^3\bm{U}_j^*$.

We consider three simulation experiments for the tensor logistic regression model. The first experiment is designed to verify how the tail behavior of the covariates, quantified by $\lambda$, is related to the computational performance of the proposed method. The second experiment aims for verifying the local moment effect. The third experiment is designed to compare the performance of vanilla gradient descent and robust gradient descent methods.

\subsubsection{Experiment 5: Dependence on Tail Behavior of Covariates}

In this model, we consider that all entries in $\cm{X}_i$ (or $\bm{X}_i$) are independent and follow the Student's $t_{2+2\lambda}$ distibution, and $y_i$ is generated by \eqref{eq:logistic_reg}. We vary $\lambda\in\{0.1,0.4,0.7,1.0,1.3,1.6\}$ and set the sample size as $n=10\times 2^{m}$, where $m\in\{1,2,3,4,5\}$. For the generated data, we apply the proposed RGD method with initial values set to the ground truth, $a=b=1$, step size $\eta=10^{-3}$, truncation threshold $\tau=\sqrt{n/\log(\bar{p})}$, and number of iterations $T=300$. 

In this experiment, we aim to verify whether the RGD iterates converge and to explore the relationship between the emprical convergence rate and $\lambda$. According to Theorem \ref{thm:logistic}, if the iterates converge, then $\|\cm{A}^{(t)}-\cm{A}^*\|_\text{F}^2$ lie in a region with small radius. To empirically assess convergence, we compute the sample standard deviation of $\|\cm{A}^{(t)}-\cm{A}^*\|_\text{F}^2$ over iterations $t=251,\dots,300$, and label the algorithm as having converged only if this quantity is smaller than $\bar{p}\log(\bar{p})/(100n)$. 

For each pair of $\lambda$ and $m$, we replicate the entire procedure  200 times and summarize the proportion of replications that achieve convergence versus $m$ in Figure \ref{fig:logistic_Exp1}. The results confirm that the smaller value of $\lambda$, corresponding to heavy-tailed covariates, leads to a greater sample size requirement for convergence. However, for $\lambda\geq1$, the convergence patterns across different $m$ are similar, which is consistent with the theoretical sample size requirement derived in Theorem \ref{thm:logistic}.

\begin{figure}[!htp]
    \begin{centering}
    \includegraphics[width=0.75\textwidth]{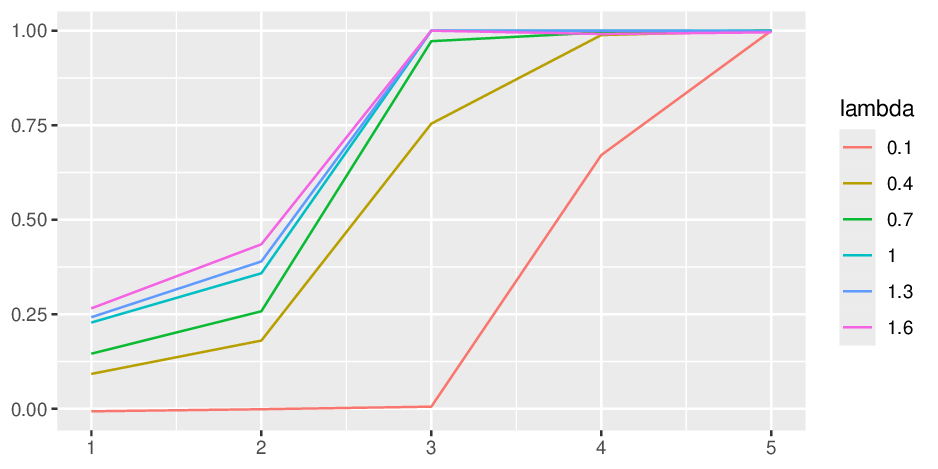}
    \vspace{-0.2cm}
    \caption{Average convergence proportion (y-axis) vs $m$ (x-axis) with varying $\lambda$ for Model III in Experiment 5}
    \label{fig:logistic_Exp1}
    \end{centering}
\end{figure}

\subsubsection{Experiment 6: Dependence on Local Moment Conditions}

Similarly to Experiment 3 in the main article, we consider the vectorized covariate $\text{vec}(\cm{X}_i)$ (or $\text{vec}(\bm{X}_i)$) follows a multivariate Gaussian distribution with mean zero and covariance $(\otimes_{j=1}^{d_0}\bm{\Sigma}_\delta)$, where $\bm{\Sigma}_\theta=0.5\bm{I}_{10}+0.5\bm{v}_\theta\bm{v}_\theta^\top$, where $\bm{v}_\theta=\sin(\theta)\bm{1}_{10} + \cos(\theta)\bm{w}$ and $\bm{w}=(1,-1,1,-1,\dots,1,-1)^\top\in\mathbb{R}^{10}$. In this setup, the entries in covariates are dependent, and the dependency is governed by the angle parameter $\theta\in[0,\pi/2]$. Specifically, when $\theta=\pi/2$, the vector $\bm{v}_\theta$ aligns with $\bm{1}_{10}$, which coincides with the true factor directions $\bm{U}_1^*=\bm{U}_2^*=\bm{U}_3^*$, resulting in a large local moment condition. When $\theta=0$, the correlation direction $\bm{v}_\theta=\bm{w}$ is orthogonal to the true factors, leading to a much smaller local moment. See more information in Appendix \ref{append:B.3}.

We consider $\theta=\theta_0\pi/8$ with $\theta_0\in\{0,1,2,3,4\}$ and set $n\in\{300,400,500,600,700\}$. For each pair of $\theta_0$ and $n$, we replicate the procedure 200 times and summarize the average of $\|\cm{A}^{(T)}-\cm{A}^*\|_\text{F}^2$ versus $n$ in Figure \ref{fig:logistic_Exp2}. As $\theta_0$ increases, the local moments increase, and the average estimation errors increase accordingly, further validating the importance of leveraging local moment conditions as emphasized in our theoretical analysis.

\begin{figure}[!htp]
    \begin{centering}
    \includegraphics[width=0.75\textwidth]{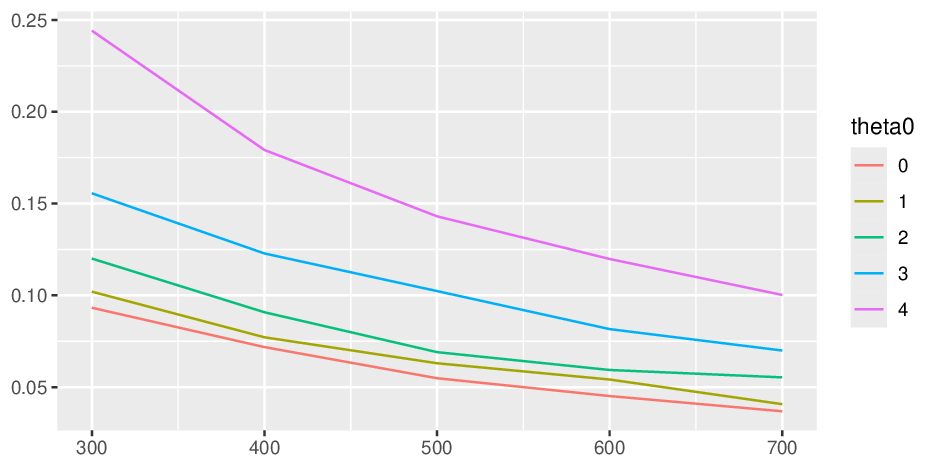}
    \vspace{-0.2cm}
    \caption{Average $\|\cm{\widehat{A}}-\cm{A}^*\|_\text{F}$ (y-axis) vs $n$ (x-axis) with varying $\theta_0$ for Model III in Experiment 6}
    \label{fig:logistic_Exp2}
    \end{centering}
\end{figure}

\subsubsection{Experiment 7: Method Comparison}

For tensor logistic regression, two distributional cases are adopted for covariates: (1) $N(0,1)$ and (2) $t_{2.1}$. All entries in covariates are independent, and we set $n=500$. We apply the proposed RGD algorithm, as well as the vanilla gradient descent (VGD) as the competitor, to the data generated. For both methods, intial values are obtained in a data-driven manner as suggested in Appendix \ref{append:init} of the supplentary materials. We set $a=b=1$, $\eta=10^{-3}$, $T=300$, and the truncation parameter $\tau$ is selected via five-fold cross-validation. 

For each distributional setting, we replicate the procedure 200 times and summarize the average of $\log(\|\cm{A}^{(T)}-\cm{A}^*\|_\text{F}^2)$, as well as their upper and lower quartiles, for the above four cases in Figure \ref{fig:logistic_Exp3}. When the covariates are light-tailed, the performances of these two estimation methods are nearly identical. However, in the heavy-tailed case, the performance of VGD deteriorates significantly, with estimation errors much larger than those of RGD.

\begin{figure}[!htp]
    \begin{centering}
    \includegraphics[width=0.75\textwidth]{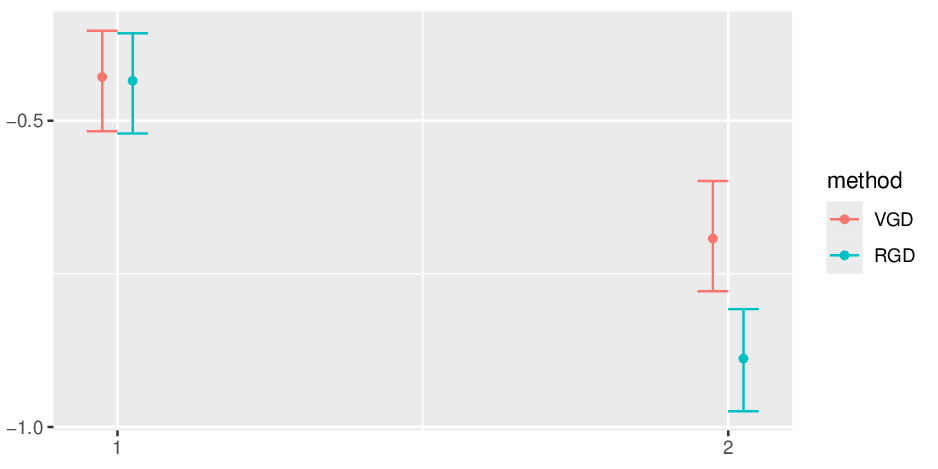}
    \vspace{-0.2cm}
    \caption{Average $\log(\|\cm{\widehat{A}}-\cm{A}^*\|_\text{F}^2)$ (y-axis) in different distributional cases (x-axis) by different methods for Model III in Experiment 6}
    \label{fig:logistic_Exp3}
    \end{centering}
\end{figure}

\subsection{Tensor PCA}

We consider the tensor PCA model.
\begin{equation}
  \text{Model IV}: \cm{Y} = \cm{A}^* + \cm{E},
\end{equation}
$\cm{A}^*=\sqrt{10}\cdot\bm{1}_{10}\circ\bm{1}_{10}\circ\bm{1}_{10}=\cm{S}^*\times_{j=1}^3\bm{U}_j^*$.

We consider three simulation experiments for the tensor PCA. The first experiment is designed to verify how the tail behavior of the noise, quantified by $\epsilon$, is related to the computational performance of the proposed method. The second experiment aims for verifying the local moment effect. The third experiment is designed to compare the performance of vanilla gradient descent, Huber estimator, and robust gradient descent methods.

\subsubsection{Experiment 8: Dependence on Tail Behavior of Noise}

We consider that the noise follow a $t_{1+\epsilon}$ distribution, vary $\epsilon\in\{0.1,0.4,0.7,1.0,1.3,1.6\}$, and set the sample size as $n=200\times 2^{m}$, where $m\in\{1,2,3,4,5\}$. For the generated data, we apply the proposed RGD method with the same tuning as in Experiment 2. According to Theorem \ref{thm:PCA}, after a sufficent number of iterations, $-\log(\|\cm{A}^{(T)}-\cm{A}^*\|_\text{F}^2) = C(\bar{p},\epsilon) + C[\epsilon_{\text{eff}}/(1+\epsilon_{\text{eff}})]m$, where $C(\bar{p},\epsilon)$ is a constant depending on $\bar{p}$ and $\epsilon$. 

Therefore, for each pair of $\epsilon$ and $m$, we replicate the procedure 200 times and summarize the average of negative log errors versus $m$ in Figure \ref{fig:Exp1_PCA}. For each value of $\epsilon$, the average negative log errors exhibit a linear relationship with $m$. Notably, the slope of this linear relationship shows a smooth transition: when $\epsilon\in(0,1)$, the slope increases as $\epsilon$ increases; when $\epsilon\geq1$, the slopes stablize. These empirical findings are similar to those in Experiment 2, and verify the smooth transition in statistical convergence rate as stated in Theorem \ref{thm:PCA}.

\begin{figure}[!htp]
    \begin{centering}
    \includegraphics[width=0.75\textwidth]{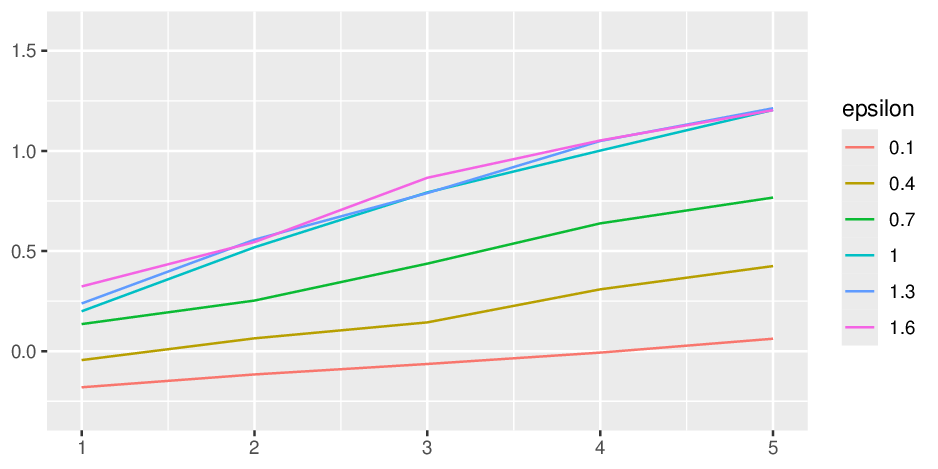}
    \vspace{-0.2cm}
    \caption{Average $-\log(\|\cm{\widehat{A}}-\cm{A}^*\|_\text{F}^2)$ (y-axis) vs $m$ (x-axis) with varying $\epsilon$ for Model IV in Experiment 8}
    \label{fig:Exp1_PCA}
    \end{centering}
\end{figure}

\subsubsection{Experiment 9: Dependence on Local Moment Conditions}

Similarly to Experiments 3 and 6, we consider the vectorized covariate $\text{vec}(\cm{X}_i)$ (or $\text{vec}(\bm{E}_i)$) follows a multivariate Gaussian distribution with mean zero and covariance $(\otimes_{j=1}^{d_0}\bm{\Sigma}_\delta)$, where $\bm{\Sigma}_\theta=0.5\bm{I}_{10}+0.5\bm{v}_\theta\bm{v}_\theta^\top$, where $\bm{v}_\theta=\sin(\theta)\bm{1}_{10} + \cos(\theta)\bm{w}$ and $\bm{w}=(1,-1,1,-1,\dots,1,-1)^\top\in\mathbb{R}^{10}$.

We consider $\theta=\theta_0\pi/8$ with $\theta_0\in\{0,1,2,3,4\}$ and set $n\in\{300,400,500,600,700\}$. For each pair of $\theta_0$ and $n$, we replicate the procedure 200 times and summarize the average of $\|\cm{A}^{(T)}-\cm{A}^*\|_\text{F}^2$ versus $n$ in Figure \ref{fig:PCA_Exp2}. As $\theta_0$ increases, the local moments increase, and the average estimation errors increase accordingly, further validating the importance of leveraging local moment conditions as emphasized in our theoretical analysis.

\begin{figure}[!htp]
    \begin{centering}
    \includegraphics[width=0.75\textwidth]{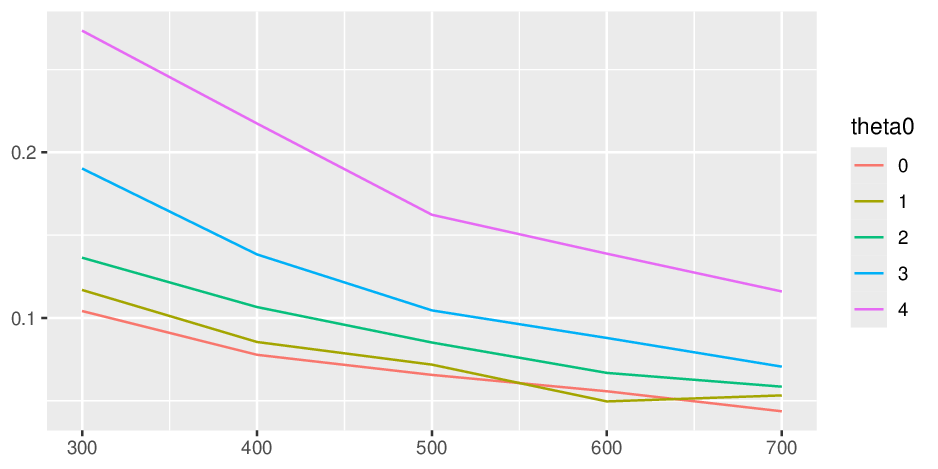}
    \vspace{-0.2cm}
    \caption{Average $\|\cm{\widehat{A}}-\cm{A}^*\|_\text{F}$ (y-axis) vs $n$ (x-axis) with varying $\theta_0$ for Model IV in Experiment 9}
    \label{fig:PCA_Exp2}
    \end{centering}
\end{figure}

\subsubsection{Experiment 10: Method Comparison}

For tensor PCA, two distributional cases are adopted for the noise: (1) $N(0,1)$ and (2) $t_{1.2}$ distribution. All entries in noise are independent, and we set $n=500$. We apply the proposed RGD algorithm, as well as the vanilla gradient descent (VGD) and adaptive Huber regression (HUB) as competitors, to the data generated. For all methods, intial values are obtained in a data-driven manner as suggested in Appendix \ref{append:init} of the supplentary materials. We set $a=b=1$, $\eta=10^{-3}$, $T=300$, and the truncation parameter $\tau$ is selected via five-fold cross-validation. 

For each model and distributional setting, we replicate the procedure 200 times and summarize the average of $\log(\|\cm{A}^{(T)}-\cm{A}^*\|_\text{F}^2)$, as well as their upper and lower quartiles, for the above four cases in Figure \ref{fig:PCA_Exp3}. When noise is light-tailed, the performances of three estimation methods are nearly identical. However, in heavy-tailed cases, the performance of VGD deteriorates significantly, with estimation errors much larger than those of the other two methods. Overall, the RGD method consistently yields the smallest estimation errors across all three methods. These numerical findings confirm the robustness and efficiency of the proposed method in handling heavy-tailed tensor PCA.

\begin{figure}[!htp]
    \begin{centering}
    \includegraphics[width=0.75\textwidth]{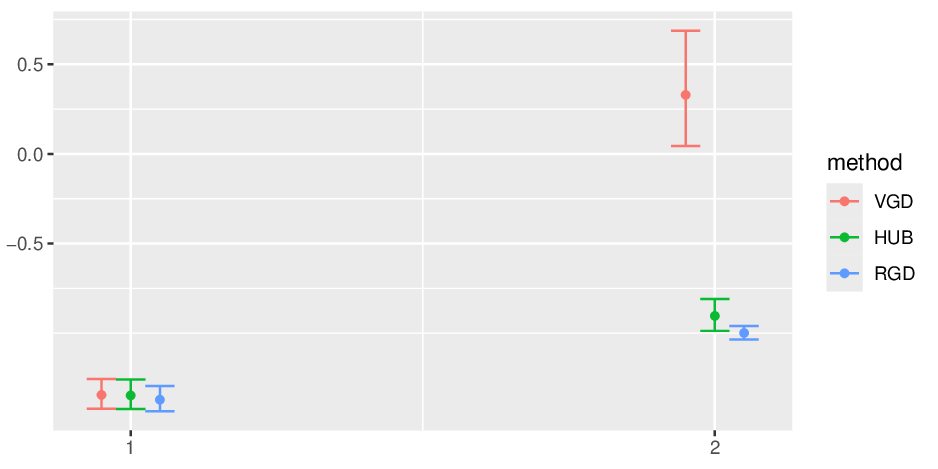}
    \vspace{-0.2cm}
    \caption{Average $\log(\|\cm{\widehat{A}}-\cm{A}^*\|_\text{F}^2)$ (y-axis) in different distributional cases (x-axis) by different methods for Model IV in Experiment 10}
    \label{fig:PCA_Exp3}
    \end{centering}
\end{figure}

\end{appendix}

\end{document}